%% file: main.tex
\documentclass[twoside,leqno]{article}

\usepackage[letterpaper]{geometry}
\usepackage{thm-restate}

\usepackage[T1]{fontenc}
\usepackage{amsfonts}
\usepackage{graphicx}
\usepackage[inline]{enumitem}
\usepackage{booktabs}
\usepackage{xcolor}
\usepackage{hyperref} 
\usepackage{wrapfig}
\usepackage{epstopdf}
\ifpdf
  \DeclareGraphicsExtensions{.eps,.pdf,.png,.jpg}
\else
  \DeclareGraphicsExtensions{.eps}
\fi

\hypersetup{
    colorlinks=true,
    linktoc=all,
    allcolors=blue,
}

\usepackage{siamproceedings}

\usepackage[algo2e, linesnumbered, boxed, vlined]{algorithm2e}
\DontPrintSemicolon
\SetKwFor{RepTimes}{Repeat}{times:}{end}
\SetKwProg{function}{Function}{:}{}
\SetKwProg{parfor}{parallel for}{\:do}{}
\SetKwProg{while}{while}{\:do}{}
\SetKwProg{Elif}{else if}{\:do}{}

\usepackage{etoolbox}
\csletcs{algorithm}{algorithm2e}
\csletcs{endalgorithm}{endalgorithm2e}
\csletcs{algorithm*}{algorithm2e*}
\csletcs{endalgorithm*}{endalgorithm2e*}
\csletcs{listofalgorithms}{listofalgorithmes}

\usepackage{amsopn}

\makeatletter
\newcommand{\multiline}[1]{%
  \begin{tabularx}{\dimexpr\linewidth-\ALG@thistlm}[t]{@{}X@{}}
    #1
  \end{tabularx}
}
\makeatother

\numberwithin{equation}{section}
\newtheorem{thm}[equation]{Theorem}

\newtheorem{lem}[equation]{Lemma}

\input{macro}
\title{Faster Parallel Batch-Dynamic Algorithms \\ 
for Low Out-Degree Orientation
}

\author{Guy E. Blelloch\footnotemark[1] \\ guyb@cs.cmu.edu \and Andrew C. Brady\footnotemark[1]  \\ acbrady2020@gmail.com \and  Laxman Dhulipala\footnotemark[2] \\ laxman@umd.edu \and Jeremy T. Fineman\footnotemark[3] \\ jf474@georgetown.edu \and   Kishen N. Gowda\footnotemark[2] \\ kishen19@umd.edu \and  Chase Hutton\footnotemark[2] \\ chutton6@umd.edu}

\date{}

\begin{document}

\renewcommand{\thefootnote}{\fnsymbol{footnote}}

\maketitle


\footnotetext[1]{Carnegie Mellon University, Pittsburgh, PA }
\footnotetext[2]{University of Maryland, College Park, MD}
\footnotetext[3]{Georgetown University, Washington, DC}

\setcounter{footnote}{0}
\renewcommand{\thefootnote}{\arabic{footnote}}




\fancyfoot[R]{\scriptsize{Copyright \textcopyright\ 2027\\
Copyright for this paper is retained by authors}}


\input{0_abstract}
\clearpage


\makeatletter
\renewcommand{\tableofcontents}{%
  \section*{\contentsname}\mbox{}\par
  \@starttoc{toc}%
}
\makeatother


\clearpage


\input{1_intro.tex}

\input{new_tech_overview}

\input{3_background}

\input{4_amortized_algo}
\input{5_algorithm}




\input{potentialsMain}

\input{appendix_batch_counter}

\input{actualAnalysis}

\input{we_worstcase}

\input{appendix_deterministicBag}
\input{appendix_skyline}

\input{10_conclusion}

\myparagraph{Acknowledgements} We thank the anonymous reviewer who informed us of Lemma \ref{lem:boostOrientation} \cite{chekuri2024adaptive}. This paper is based upon work performed while attending the AlgoPARC Workshop on Parallel Algorithms and Data Structures at the University of Hawaii at Manoa, in part supported by the National Science Foundation under Grant CCF2452276. 
This work was also supported by NSF grants CCF2403235, CNS231719, CCF2119352,  CCF1919223, CCF1918989, and CCF2106759.  

\myparagraph{Acknowledgements: AI Disclosure} In addition to human review, AI (Claude, Gemini, ChatGPT) was used to check for correctness and confirm analysis, both during the research process and while writing the paper. AI was also used for proofreading, paper formatting, and to assist in literature review. AI was not used for any substantive idea or content generation. The authors take full responsibility for the correctness, originality, and integrity of the work. 

\bibliographystyle{siamplain}
\bibliography{strings.bib, refs.bib}

\clearpage

\input{appendix_prelims}
\input{appendix_potential}
\input{appendix_apps}


\end{document}

%% file: macro.tex
\ifx\nocomment\undefined
\newcommand{\guy}[1]{{\color{brown}{\bf Guy:} #1}}
\newcommand{\jeremy}[1]{{\color{violet}{\bf Jeremy:} #1}}
\newcommand{\ac}[1]{{\color{magenta}{\bf Andrew:} #1}}
\newcommand{\ch}[1]{{\color{cyan}{\bf Chase:} #1}}
\newcommand{\kishen}[1]{{\color{blue}{\bf Kishen:} #1}}
\newcommand{\laxman}[1]{{\color{purple}{\bf Laxman:} #1}}
\else
\newcommand{\guy}[1]{{}}
\newcommand{\jeremy}[1]{{}}
\newcommand{\ac}[1]{{}}
\newcommand{\ch}[1]{{}}
\newcommand{\kishen}[1]{{}}
\newcommand{\laxman}[1]{{}}
\fi

\newcommand{\whp}{whp}
\newcommand{\myparagraph}[1]{\vspace{.05in}\noindent {\bf #1.}}

\newcommand{\cp}{\left\lceil \frac{c}{\log n}\right\rceil}

\newcommand{\modelop}[1]{\texttt{#1}}
\newcommand{\forkins}{\modelop{fork}}

\newcommand{\thread}{thread}

\newcommand{\paradepth}{span}
\newcommand{\depth}{\paradepth}

\newcommand{\bIns}{\textsc{BatchInsert}\xspace}

\newcommand{\threshcutoff}{\ensuremath{\alpha_{\text{cut}}\xspace}}
\newcommand{\threshstatic}{\ensuremath{\alpha_{\text{stat}}\xspace}}
\newcommand{\threshoffline}{\ensuremath{\alpha_{\text{offl}}\xspace}}

\newcommand{\defnn}[1]{\emph{\textbf{#1}}}

\renewcommand{\paragraph}[1]{\medskip\noindent{\bf #1}\xspace}

\newcommand{\vlow}{V^\text{low}}
\newcommand{\vhigh}{V^\text{high}}

%% file: 0_abstract.tex
\begin{abstract}

A low out-degree orientation directs each edge of an undirected graph with the goal of minimizing the maximum out-degree of a vertex. 
In the parallel batch-dynamic setting, one can insert or delete batches of edges, and the goal is to process the entire batch in parallel with work per edge similar to that of a single sequential update and with span (or depth) for the entire batch that is polylogarithmic.  
In this paper we present work-efficient parallel batch-dynamic algorithms for maintaining a low out-degree orientation of an undirected graph, both in the amortized and worst-case settings.
All results herein achieve polylogarithmic \depth{}; the focus of this paper is on minimizing the work, which varies across results.  

In the amortized setting, we give a parallel batch-dynamic algorithm that maintains a $O(c)$-orientation in $O(\log n)$ work per update in expectation, where $c$ is a known upper bound on the arboricity over the update sequence. Our algorithm has optimal work in expectation. Note that $\Theta(c)$ is the minimum possible max out-degree over the entire sequence. 
This result is the parallelization of the classic dynamic orientation algorithm of Brodal and Fagerberg~[WADS~'99], and, in this setting, is a logarithmic factor faster than the best prior parallel amortized algorithm of Liu et al.~[SPAA '22]. 

In the worst-case setting, we give an $O(c+\log n)$-orientation with worst-case expected work per update $O(\log n)$. This is work-efficient, matching the best known sequential dynamic work of Berglin and Brodal ~[Algorithmica~'20], and implies the existence of an $O(c)$-orientation algorithm with $O(\log^2 n)$ worst-case expected work per update. Our algorithm significantly improves, in the setting where $c$ is a fixed upper bound on arboricity, upon the parallel algorithm of Ghaffari and Koo~[SPAA~'25], which maintains a $O(c)$-orientation with $O(\log^9 n)$ worst-case work per edge with high probability (whp). 

Both of our algorithms also have deterministic bounds with an additional logarithmic factor in the work. Similar to Berglin and Brodal, by adjusting the parameters of our worst-case algorithm, we achieve tradeoffs where the work is faster than $O(\log n)$ in expectation per update but the max out-degree is worse.

\end{abstract}

%% file: 1_intro.tex
\section{Introduction}

The {\em low out-degree orientation problem} is to assign a direction
to each edge of an undirected graph $G=(V,E)$ while minimizing the
maximum out-degree of the resulting directed graph $\overline{G}$. 
Formally, an orientation of an undirected graph is called an
$x$-orientation if every vertex has out-degree at most $x$.
This problem
is fundamentally linked to {\em arboricity}, a widely-used measure of the sparsity of a graph. 
In particular, the optimal orientation is within one of the arboricity.
Here we are interested in approximate solutions.   In the sequential static setting,\footnote{By static, we mean the standard setting, where given a single graph, the algorithm must output an orientation. This is in contrast to the dynamic setting, where the edges of the graph are being inserted and deleted.} a 2-approximation can be found
in linear time~\cite{MatulaBeck1983}.
In the parallel static setting, a
$(2 + \epsilon)$-approximation can be found in expected linear work and polylogarithmic depth \cite{barenboim2008sublogarithmic}. 

{\em Dynamic low out-degree orientations}, which maintain an
orientation as the graph is modified over a series of edge insertions
and deletions were first studied in the seminal work of Brodal and
Fagerberg~\cite{BF99}. The motivation for their work was the
problem of constructing deterministic data structures for answering
adjacency queries on sparse graphs. In the subsequent years, dynamic orientation has
emerged as a fundamental building block for constructing dynamic
algorithms for many other graph problems, including dynamic vertex cover~\cite{peleg2016dynamic}, coloring~\cite{henzinger2020explicit}, independent sets~\cite{onak2020fully}, and  matchings~\cite{bernstein2015fully, bernstein2016faster, he2014orienting}.

Brodal and Fagerberg's algorithm~\cite{BF99} is extremely simple and intuitive,
and can be shown to maintain an $O(c)$-orientation with $O(\log n)$
amortized number of edge reorientations (flips) per update, where $c$ is
the maximum arboricity over the sequence of updates.\footnote{Note that for out-degree orientation, the number of flips is a measurement of recourse. We count an edge insert or delete itself as a flip as well.}
Intriguingly, their algorithm is optimal in the sense that it is $O(1)$-competitive with the number of
reorientations made by {\em any} online or offline algorithm,
regardless of its running time. 

In the subsequent years, a rich body of work has significantly improved our
understanding of low out-degree orientations in the sequential dynamic
setting, and on the tradeoff between maximum out-degree, update time, and
amortized vs. worst-case bounds~\cite{he2014orienting, kowalik2007adjacency, christiansen2022fully, BB20, chekuri2024adaptive}.\footnote{When we say worst-case, we mean no amortization, but that randomization is permitted. Thus, one can have expected worst-case bounds, or worst-case bounds whp.}
Some of these worst-case algorithms function by finding a directed path from a high (or maximum) out-degree vertex to a low out-degree vertex, and flipping all edges along this path: this has the net effect of decreasing the high vertex's out-degree by 1, increasing the low vertex's out-degree by 1, and not changing the out-degree of any of the intermediate vertices on the path. 
However, this approach can lead to high runtime \cite{Kopelowitz2013Orientation, SW20, chekuri2024adaptive, ghaffari2025}.
Currently, for $O(c)$-orientation, the best amortized bounds 
are still achieved by Brodal and Fagerberg's original
algorithm, and the best worst-case bound comes from an
$O(c + \log n)$ orientation with $O(\log n)$ time per update~\cite{BB20}, which implies an $O(c)$ orientation with $O(\log^2 n)$ time per update~\cite{chekuri2024adaptive}.

There has been growing interest in the {\em parallel batch-dynamic setting},
where the updates arrive as {\em batches}, and the goal is to
maintain a good quality orientation in low {\em work} (the total
number of operations performed by the algorithm) and {\em \depth{}} (the
length of the longest chain of sequential dependencies, sometimes called depth).   Typically, the goal is to achieve work-efficient algorithms and \depth{} that is polylogarithmic.   A parallel batch-dynamic algorithm is considered \emph{work efficient} if it does only a constant factor more work per edge than the best sequential dynamic algorithm.
We note that the parallel batch-dynamic setting has received considerable interest in recent years, with progress being made on batch-dynamic algorithms for a number of fundamental graph problems~\cite{GT24, acar2019batchconnect, yesantharao2021parallel,acar2020changeprop, dhulipala2019parallel, ghaffari2025, BB25a, ZKG+25, wang2020closest}.
Batching operations eliminates the serial bottleneck imposed by traditional
sequential dynamic algorithms, and thus
allows applications to exploit more parallelism. Parallel batch-dynamic
algorithms have also emerged as a useful tool for constructing efficient
(static) parallel algorithms~\cite{ghaffari2023nearly,anderson2021parallel}.

In the parallel batch-dynamic setting, Liu et al.~\cite{liu2022parallel} present a $(4 + \epsilon)c$ orientation algorithm,  where each batch update does amortized $O(\log^2 n)$ work per edge (in the batch) with high probability (\whp{}), and polylogarithmic \depth{} \whp{}, and 
Ghaffari and Koo~\cite{ghaffari2025} gave an algorithm for $(2+\epsilon)c$ orientation, with $O(\log^9 n)$ worst-case work
per edge \whp{} and polylogarithmic \depth{} \whp{}. Note these works are solving in two respects a more difficult problem: their max out-degree bounds are with respect to the arboricity of the current graph (i.e., adaptive arboricity) rather than against a fixed upper bound, and they also maintain a more general data structure that can give an approximate $k$-core decomposition in the same time bounds. However, neither result is work-efficient in the setting where $c$ is a fixed upper bound on the arboricity, which is of consistent theoretical interest \cite{onak2020fully,solomon20improved,Kopelowitz2013Orientation,dhulipala2021hierarchical, BF99,BB20,kaplan18dynamic}. 
Liu et al. is a $O(\log n)$ factor off of the amortized algorithm of Brodal and Fagerberg, and Ghaffari and Koo are many log factors different from Berglin and Brodal's worst-case result. 
These recent results leave open the question of whether it is possible to develop parallel work-efficient algorithms for maintaining an $O(c)$-orientation in the fixed upper bound setting. 

In this paper, we develop the first work-efficient amortized parallel batch-dynamic algorithm and the first work-efficient worst-case parallel batch-dynamic algorithm in this regime. We summarize prior and our amortized results in Table \ref{tab:resultsAmortized} and the worst-case results in Table \ref{tab:results}. We note that using reductions from prior work, our results lead to better maximal matching and coloring algorithms. We give these improved bounds in Section \ref{sec:apps}.

First, building on Brodal and Fagerberg's amortized algorithm \cite{BF99},
we present an algorithm that maintains an $O(c)$ orientation and supports batches of updates each of which requires $O(\log n)$ expected amortized work per edge. 
Our amortized algorithm is optimal, inheriting the $O(1)$-competitive guarantee of Brodal and Fagerberg's amortized algorithm~\cite{BF99}, and improves by a logarithmic factor relative to Liu et al.'s work bounds~\cite{liu2022parallel}. Our amortized algorithm uses the static orientation routine of Barenboim and Elkin on a certain subgraph to bound the maximum out-degree and adapts the potential framework of Brodal and Fagerberg to bound the work \cite{barenboim2008sublogarithmic,BF99}. 

Second, building on Berglin and Brodal \cite{BB20}, we design an $O(c+\log n)$-orientation algorithm with $O(\log n)$ expected work worst-case per edge update. 
Our worst-case algorithm is work-efficient, matching Berglin and Brodal's runtime bounds, and greatly reducing the exponent on the logarithm relative to Ghaffari and Koo's work bound~\cite{ghaffari2025}. 
Our worst-case algorithm introduces \emph{skylines} to efficiently flip many edges in parallel. 
To bound the maximum out-degree, building on Berglin and Brodal's analysis \cite{BB20}, we develop a \emph{batch counter game}, prove that the batch counter game has good behavior, and develop a general \emph{domination framework} to connect the counter game to the orientation problem.



As in the sequential algorithm of Berglin and Brodal~\cite{BB20}, we have a tradeoff curve, where we can reduce the number of flips by permitting a worse orientation quality.
When stating our main theorems in Section \ref{sec:overview}, we implicitly consider a specific choice of parameters for clarity; the general bounds are provided in Section \ref{sec:dom}.

All of our algorithms are polylogarithmic \depth{} per batch. The only randomized component of our algorithms is a linear-work semisort \cite{gu2015top}.  Thus, by replacing it with a deterministic parallel sort with $O(n \log n)$ work, our algorithms also achieve a deterministic work bound with a logarithmic overhead in work, and our depth bound is deterministic rather than \whp{}.

\begin{table}
    \centering
    \begin{tabular}{c|c|c|c|c}
        Max out-degree & Work & Parallel? & Adaptive? & Citation   \\
            \hline

        $(2+\epsilon)c$ & $O(b \log n)$ am det &  & & \cite{BF99} \\
        $(4+\epsilon)c$ & $O(b \log^2 n)$ am \whp{} & \checkmark & \checkmark & \cite{liu2022parallel} \\
        $(6+\epsilon)c$ & $O(b \log n)$ am exp, $O(b \log^2 n)$  am det & \checkmark & & Thm \ref{thm:amortized_main} \\

    \end{tabular}
    \caption{Comparison between our amortized algorithm's bounds and prior work. Amortized is abbreviated as am, expected as exp, and deterministic as det.}
    \label{tab:resultsAmortized}
\end{table}
\begin{table}
    \centering
    \begin{tabular}{c|c|c|c|c}
         Max out-degree &Work & Parallel? &  Adaptive? & Citation   \\
            \hline

         $(1+\epsilon) c + 2$ & $O(b \epsilon^{-6} \log^3 n \log c)$ det & & \checkmark & \cite{chekuri2024adaptive}  \\

        $(2+\epsilon)c$ & $O(b \epsilon^{-22} \log^9 n)$ \whp{} & \checkmark & \checkmark &
        \cite{ghaffari2025} \\

        $O(c)$ & $O(b \log^2 n)$ det & & &  \cite{BB20} \\
        $O(c)$ & $O(b \log^2 n)$ exp, $O(b \log^3 n)$ det & \checkmark & & Thm \ref{thm:boundsOc} \\
        $O(c+\log n$) & $O(b \log^2 n \log c)$ det & & \checkmark & \cite{chekuri2024adaptive} \\
        $O(c+\log n$) & $O(b \log n)$ det & &  & \cite{BB20} \\
        $O(c+\log n$) & $O(b \log n)$ exp, $O(b \log^2 n)$ det & \checkmark &  & Thm \ref{thm:boundsOclogn} \\
        $O(c\log n$) & $O(b \sqrt{\log n})$ det  &  &  & \cite{BB20}\\
        $O(c\log n$) & $O(b \sqrt{\log n})$ exp, $O(b  \log^{1.5} n )$ det & \checkmark &  & Thm \ref{thm:boundsOcsqrtlogn}\\

       $O(c \log^2 n)$ & $O(b)$ det & & & \cite{BB20} \\
       $O(c \log^2 n)$ & $O(b)$ exp, $O(b \log n)$ det & \checkmark & &  Thm \ref{thm:boundsclognsquared} \\ 
            \end{tabular}
    \caption{Summary of our worst-case results and prior work. Expected is abbreviated as exp, and deterministic as det. }
    \label{tab:results}
\end{table}

%% file: new_tech_overview.tex
\section{Technical Overview \label{sec:overview}}


In this paper, we revisit two of the fastest sequential algorithms, one classic and one recent, for maintaining 
low out-degree orientations and study how to parallelize them~\cite{BF99,BB20}. In the rest of this section, we 
describe the challenges of parallelizing these algorithms and our solutions. For our worst-case algorithm description, we summarize notation used in Table \ref{tab:techOverviewNotation}.

\subsection{Amortized Low Out-Degree Orientation}
\myparagraph{The Brodal-Fagerberg Algorithm}
In their seminal work, Brodal and Fagerberg~\cite{BF99} proposed a simple dynamic algorithm for maintaining an $O(c)$-orientation with optimal bounds, where $c$ is an upper-bound on the maximum arboricity of the graph over the sequence of updates. Their algorithm works as follows: 
For edge deletions, the algorithm simply deletes the edge. For edge insertions, it orients the new edge arbitrarily out of one of its endpoints. If this causes the vertex's out-degree to be excessive, i.e., exceed a $O(c)$ \emph{cutoff threshold}, the algorithm flips all of its outgoing edges and repeats this process for any other vertex, one by one, whose out-degree becomes excessive.

They analyze their algorithm by comparing the number of edge flips performed over a sequence of updates to the number of edge flips performed by an optimal offline algorithm for $\delta$-orientation that sees the entire sequence of operations in advance. Following their terminology, we call an edge \emph{good} if it agrees with the offline orientation and \emph{bad} otherwise. 
They define a potential function equal to the number of bad edges in the graph. Edge insertions and edge flips performed by the offline strategy can increase the potential. However, if a vertex has many out-edges, most of them must be bad (since a vertex can have at most $\delta$ good out-edges). Thus, by flipping all of them, the algorithm can only decrease the overall potential.
Hence, by a potential-based analysis, the number of flips performed by the algorithm over the entire sequence of updates is bounded by the number of flips performed by the offline strategy.
Finally, they show that there exists an offline strategy that maintains an $O(c)$-orientation with $O(\log n)$ flips.
Interestingly, they show that the amortized number of flips performed by this algorithm is $O(1)$-competitive with the number of flips made by any other sequence of flips that maintains a bounded out-degree, and thus the algorithm achieves optimal work up to constant factors, as a single flip takes constant work. 

\myparagraph{Parallelizing the Brodal-Fagerberg Algorithm: Challenges and Solution}
Unfortunately, parallelizing this algorithm appears difficult due to the chains of sequential dependencies inherent in its repeated-flipping approach.
For instance, consider Figure \ref{fig:spanbreakerBF}. In each step, we flip one vertex below $4c$, but these flips cause another vertex's out-degree to rise to $4c$.  Even if we flipped all vertices with high out-degree in each step, operating on Figure \ref{fig:spanbreakerBF} still has high \depth{}, because a vertex does not become high out-degree until its neighbor flips.

\begin{figure}
\centering
\includegraphics[width=.4\textwidth]{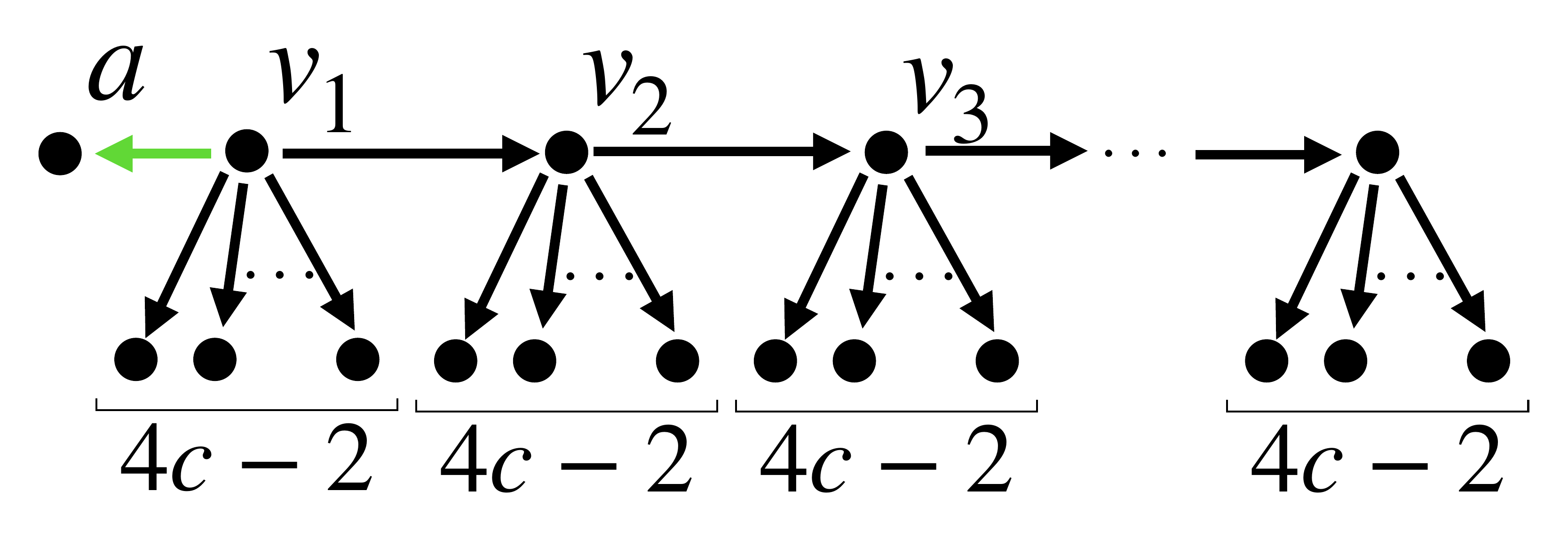}
\caption{Example graph with high \depth{} for Brodal-Fagerberg's Algorithm. Let the cutoff for flipping be $4c$. Note that the black edges can all be inserted in any order without the algorithm performing any flips because the out-degree bound is not violated. When the green edge $(v_1,a)$ is inserted, $v_1$ flips, then $v_2$ in the next step (now having $4c$), so on and so forth. Thus there will be $\Omega(t)$ iterations of the algorithm, and $t$ can be as high as $\Theta(\frac{n}{c})$. Note that similar examples arise even if the edges are not inserted in an arbitrary orientation direction.}
\label{fig:spanbreakerBF}
\end{figure}
Our solution is the following simple strategy: collect all of the outgoing edges of the vertices that currently exceed the \emph{cutoff threshold} (same as before), and run a parallel static orientation algorithm on the local subgraph induced by these edges to orient them. Parallel static orientation can be done in linear expected time and polylog depth using Barenboim and Elkin's algorithm  \cite{barenboim2008sublogarithmic}.\footnote{See Section \ref{sec:Appendix} for a formal lemma statement and proof adapted to the batch-dynamic setting.}

This guarantees that the out-degrees of these vertices are now bounded by the threshold of the static algorithm (i.e., \emph{static threshold}), whereas the rest of the vertices, in the worst case, have their out-degrees bounded by the sum of the cutoff threshold and the static threshold of~$O(c)$.

Our parallel algorithm only requires a \emph{single round} of static orientation on a local affected subgraph to fix the orientation.
The static orientation step effectively reduces the dependencies, but at the cost of a slightly worsened bound (by constant factors) on the out-degree. At the same time, our approach of using static orientation yields the same asymptotic potential reduction as Brodal and Fagerberg's strategy. 
If a high-degree vertex enters the static orientation, it will necessarily leave with far fewer out-edges than it started with, meaning that many flips occurred off of this vertex---enough for a constant amount of potential to be lost per flip. 

Our amortized algorithm result is stated in Theorem~\ref{thm:amortized_main} and proved in Section~\ref{sec:amortized}. Since the only randomized component in our algorithms is a semisort \cite{gu2015top}, we achieve deterministic work and depth bounds by switching from a randomized semisort to a deterministic sort after a fixed time budget (see Lemma \ref{lem:sort} for details). We also show that the algorithm is asymptotically optimal in terms of the number of flips performed, as stated in Theorem~\ref{thm:amortized_optimal}, since it satisfies the same property as the sequential algorithm; i.e., it performs asymptotically the same number of flips as any algorithm that maintains an $O(c)$ orientation.

\begin{restatable}{theorem}{amortizedmain}\label{thm:amortized_main}
    Given an arboricity $c$ preserving sequence of batch updates on a graph $G$ and $\epsilon \in (0,2]$, Algorithm~\ref{alg:amortized} maintains a $(6+\epsilon)c$-orientation, and, for an input batch of size $b$,
    \begin{itemize}
        \item performs amortized $O(\epsilon^{-1}b)$ edge flips for batch insertions, and amortized $O(\epsilon^{-2}b \log n)$ edge flips for batch deletions,
        \item supports batch insertions in amortized expected $O(\epsilon^{-2}b)$ work and $O(\epsilon^{-1}\log^2 n)$ \depth{}, and batch deletions in amortized expected $O(\epsilon^{-3}b\log n)$ work and $O(\log n)$ \depth{}.
        \item supports batch insertions in deterministic amortized $O(\epsilon^{-2}b \log n)$ work and $O(\epsilon^{-1}\log^2 n)$ \depth{}, and batch deletions in deterministic amortized $O(\epsilon^{-3}b \log^2 n)$ work and $O(\log n)$ \depth{}.

    \end{itemize}
\end{restatable}

\begin{restatable}{theorem}{amortizedoptimal}\label{thm:amortized_optimal}
Algorithm \ref{alg:amortized} performs an optimal number of flips (up to constant factors) and runs in optimal expected work.
\end{restatable}

\myparagraph{Comparison with Kaplan and Solomon} Kaplan and Solomon adapt Brodal and Fagerberg's algorithm to the dynamic distributed setting \cite{kaplan18dynamic}. Both our and their algorithm are amortized, and find an affected subgraph then statically orient it. However, their choice of affected subgraph is very different, and  to extend their approach to the parallel batch-dynamic setting would be difficult if feasible, and would require new technical ideas.  

\subsection{Worst-Case Low Out-Degree Orientation}

\input{techaltalt}

\begin{table}[]
    \centering
    \begin{tabular}{c|c}
    Symbol & Meaning \\
    \hline
         $k$ & number of flips of $k$-Flips algorithm  \\
         $F_v$ & front of vertex out-edge queue \\
         $B_v$ & back of vertex out-edge queue \\
        $p(v)$ & potential of vertex $v$ \\
            P(1$\rightarrow$3) & properties of potential function \\

         $\epsilon$ & small additive error in potential function \\
         $\delta$ & max out-degree of offline orientation strategy \\
         $\sigma$ & worst-case number of flips of offline orientation strategy  \\
         $\beta$ & $6\delta \epsilon$ \\
         $\eta$ & $1 + 1/\epsilon + 2\sigma$ \\
         $T$ & threshold of skyline \\
         $\hat x$ & max counter in Berglin Brodal counter game \\
         $S$ & typically, a skyline \\
         $c'$ & rounding of skyline \\
         $Y$ & starting weight in batch counter game \\
         $H$ & cap in batch counter game \\
         R(1$\rightarrow$3) & rules of batch counter game \\
         $p$ & potential state before update \\
         $u$ & counter state before update \\
         $r$ & potential state after edge updates and offline flips \\
         $q$ & potential state after algorithm run \\
         $w$ & counter state we build to dominate $q$ that is reachable from $u$ \\
         J(1$\rightarrow 3$) & conditions of $(H,T)$-bounded \\
    \end{tabular}
    \caption{Notation used in worst-case algorithm section of technical overview.}
    \label{tab:techOverviewNotation}
\end{table}

%% file: techaltalt.tex

We now turn to maintaining an orientation with worst-case guarantees, building on the sequential algorithm of Berglin and Brodal, which we call the $k$-Flips algorithm~\cite{BB20}. We first recall their algorithm and the analysis bounding its maximum out-degree. 

\myparagraph{The \boldmath$k$-Flips Algorithm \cite{BB20}} The algorithm has a single parameter, $k$. On each insert or delete operation, it repeats the following $k$ times: flip an out-edge incident to a (current) max out-degree vertex. The runtime of this algorithm is $O(k)$. Depending on the choice of parameters, $k$ will be between constant and $O(\log n)$. 

To determine which edge is flipped, they store the out-edges in a queue, flipping edges from the front of the queue, and adding them to the back.
By using a queue, they avoid repeatedly flipping the same edge back-and-forth between two max-degree vertices. In the analysis (but not the implementation), the queue for vertex $v$ will be separated into a \emph{front} $F_v$ and a \emph{back} $B_v$: when the front is empty, the back is copied into the front.

\subsubsection{Analyzing the \texorpdfstring{\boldmath$k$}{k}-Flips Algorithm}

Berglin and Brodal's analysis relies on two tools: a potential function and a counter game, which we describe in turn. The purpose of potentials is to measure that progress is being made when edges are being flipped, even when the max out-degree is not immediately decreasing. The purpose of the counter game is to show that potential does not excessively accumulate on a single vertex.



\myparagraph{The potential function} First, fix an offline strategy $\kappa$ that
maintains a $\delta$-orientation and makes at most $\sigma$ flips per update across an online sequence of edge updates. In the Brodal Fagerberg amortized analysis, an offline strategy with $\delta=O(c)$ and $\sigma=O(\log n)$ worst-case was used; Berglin and Brodal's result achieves different runtime to max out-degree tradeoffs depending on the offline strategy chosen. Like before, an edge is said to be \emph{good} if it is oriented the same as in $\kappa$ and \emph{bad} otherwise. Note that since $\kappa$ is a $\delta$-orientation, each vertex has at most $\delta$ good out-edges.

The potential $p(v)$ of a vertex $v$ is the sum of the potentials of its out-edges, and the potential of an out-edge is dependent on whether it is in the front or back of a queue and whether it is good or bad compared to the offline strategy, as shown in Table \ref{tab:BBpotentials}. The parameter $\epsilon$ is set depending on the choice of offline algorithm; for Berglin and Brodal's main result (their $O(c+\log n)$-orientation), $\epsilon=\frac{1}{\log n}$. This potential function satisfies three important properties.

\begin{table}[]
    \centering
    \begin{tabular}{c|c|c}
         & Good & Bad \\
         \hline
        Front & $1+2\epsilon$ & $1$ \\
        Back & $1-\epsilon$ & $(1+\epsilon)$ for first $3\delta$, $1$ for rest \\
    \end{tabular}
    \caption{Potential values of an edge, as presented in \cite{BB20}. Recall that $\delta$ is the orientation quality of the offline algorithm, and $0 < \epsilon \le 1$.}
    \label{tab:BBpotentials}
\end{table}

\begin{enumerate}[topsep=3pt,itemsep=2pt,label=(\textbf{P\arabic*})]
    \item A vertex's potential should approximate a vertex's out-degree: $d^+(v) - \delta\epsilon \le p(v) \le d^+(v) + 5\delta\epsilon$
    for every $v$. Thus the maximum-potential and maximum-out-degree vertices differ by at
    most $\beta := 6\delta\epsilon$. 
    \item  Flipping an edge from a vertex of out-degree at least $4\delta$
    lowers the total potential in the graph by at least $\epsilon$. 
    \item  An edge update, together with $\kappa$'s flips, raises the total
    potential by at most $\eta\epsilon$, where $\eta := 1 + 1/\epsilon + 2\sigma$. 
\end{enumerate}

\myparagraph{Counter game} The game is played on $n$ counters, one per vertex. Let $u(v)$ be the weight of counter $v$. The weights are nonnegative and sum to a fixed total $X$. Let $\hat x$ denote the largest weight. A single move does the following:
\begin{enumerate}
\item selects a counter $v$ whose weight is within $\beta + 2=6\delta \epsilon+2$ of $\hat x$
\item removes up to $\hat x - (\beta+2)$ weight from $v$
\item adds the removed weight to other counters of its choosing, keeping $X$ fixed. 
\end{enumerate} 

 They prove that, from counters of starting
weight $O(\delta)$, no sequence of moves raises the largest weight above
$O(\delta + (\delta\epsilon + 1)\log n)$.

\myparagraph{Applying the counter game} Say that a set of counters \emph{dominates} an orientation if every counter weighs at least the
potential of its vertex. 
Then the max counter weight bounds the maximum potential and thus the maximum out-degree. 
Berglin and Brodal argue that the counters dominate the dynamic orientation given by the $k$-flips algorithm. 
Note that following an edge update, some vertices may have potential greater than their counter weight, and other vertices will have potential less than their counter weight. The domination argument shows that counter weight can be moved from places of surplus to places where it is needed. 

This leads to a maximum out-degree bound of $O(\delta + (\delta\epsilon + 1)\log n)$. 
Plugging in the offline strategy of Brodal-Fagerberg gives $\delta=O(c),\sigma=O(\log n),\epsilon=\frac{1}{\log n}$ yielding an $O(c+\log n)$-orientation with $O(\log n)$ worst-case time. Using a standard technique found in, for example, Chekuri et al. \cite{chekuri2024adaptive}, this yields an $O(c)$-orientation with $O(\log^2 n)$ worst-case runtime.\footnote{See Lemma \ref{lem:boostOrientation} for an explicit statement and proof.}

An important property of the $k$-flips algorithm that enables this domination argument is that the max out-degree only increases by 1 during an update step. This is true because an edge is always flipped from a max out-degree vertex.

\subsubsection{Why is Adapting the \texorpdfstring{\boldmath$k$}{k}-Flips Algorithm to the Batch Setting Hard?}

The sequential worst-case algorithm immediately gives a parallel dynamic algorithm (parallelizing a single update), because it only performs $k=O(\log n)$ flips, and thus runs in $O(\log n)$ \depth{}. 
However, in the parallel {\it batch-dynamic} setting, applying the sequential algorithm on a batch of size $b$ would yield $\Omega(b \log n)$ \depth{}. 
Therefore, we would like to find a way to perform many edge flips at once to lower the \depth{} for batch updates.
However, we run into four difficulties when flipping many edges simultaneously:
\begin{enumerate}[label=(\arabic*),topsep=3pt,itemsep=3pt,parsep=0pt,leftmargin=20pt]
    \item We must decide which edges to flip. We cannot choose arbitrary edges to flip, because we need a net reduction in potential.
    \item In the $k$-flips algorithm, the choice of later edge flips strongly depends on the earlier flips, since the maximum out-degree vertex changes with each flip. This dependence is important in the analysis because it bounds the max out-degree increase.
    However, because we must pick many edges to flip at once, we seemingly lose this bound.
    \item  In a batch update, we will need for many vertices to flip out-edges and transfer potential at once.
    We must either demonstrate that there exists a sequence of moves of Berglin and Brodal's sequential counter game that dominates a batch update, or develop a new counter game and domination argument more suitable to the batch setting.

        \item Data structures that are trivial in the sequential setting, if naively adapted, become expensive in the parallel batch-dynamic setting, contributing an additional $O(\log n)$ factor to the work. In particular, maintaining an out-edge queue is trivial in the sequential setting but naively adds an $O(\log n)$ to the work in the parallel batch-dynamic setting.

\end{enumerate}

\subsubsection{Our Algorithm}

Before presenting our algorithm, we must define a core primitive, skylines. Informally, we define a \defnn{skyline} to be a subset of out-edges of high out-degree vertices, where the number of out-edges contributed by each vertex increases with its out-degree. We show an example of a skyline in Figure \ref{fig:skyline}. In this language, the $k$-Flips algorithm flipped a skyline of size 1; by flipping larger skylines, we achieve greater parallelism. Informally, the \emph{threshold} of a skyline is the out-degree below which no out-edges are taken, and we say a skyline has \emph{sufficient height} if the threshold is at least $4\delta$.\footnote{For efficiency reasons in the data structures, the actual definition is slightly more complicated, see Section \ref{sec:skyline} for details.}

\begin{figure}
\hspace{.15\textwidth}
    \begin{minipage}{.35\textwidth}
    \includegraphics[width=\textwidth]{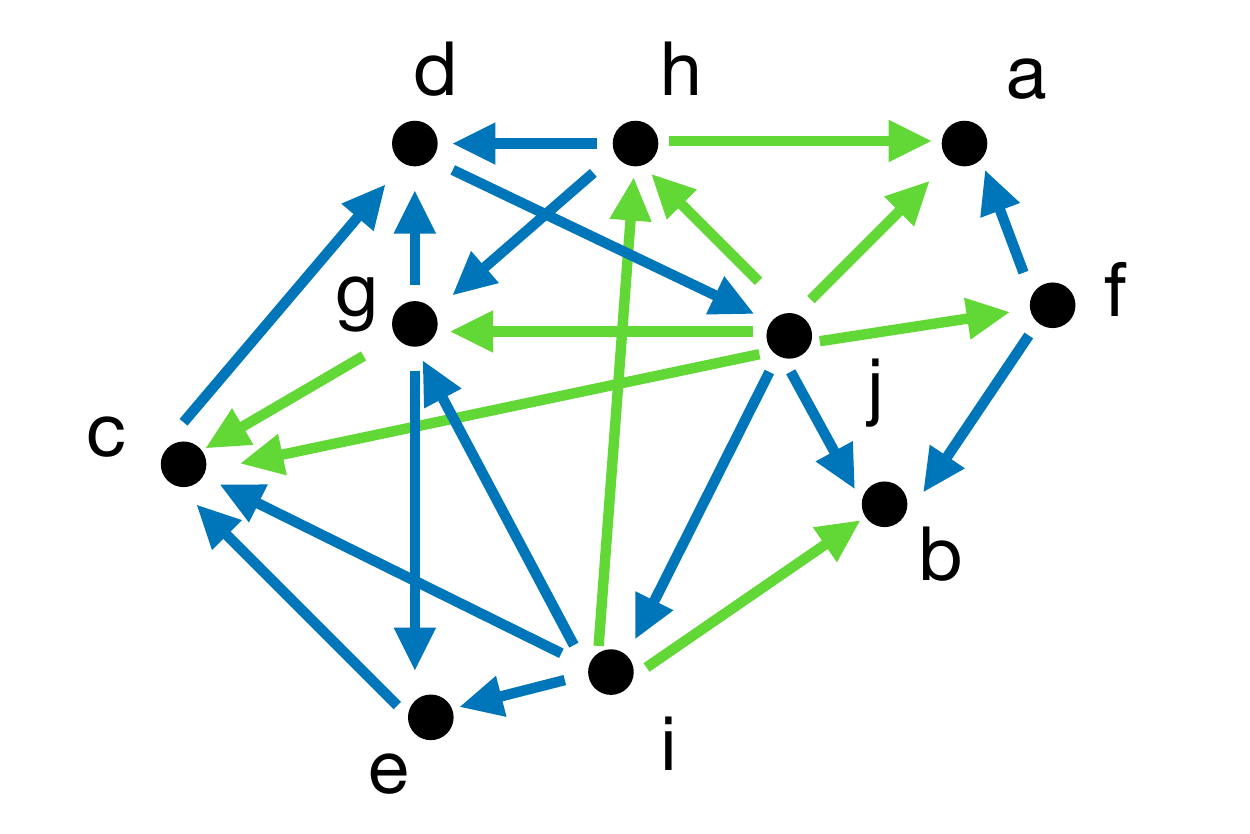}

   \end{minipage}
    \begin{minipage}{.45\textwidth}
    \includegraphics[width=\textwidth]{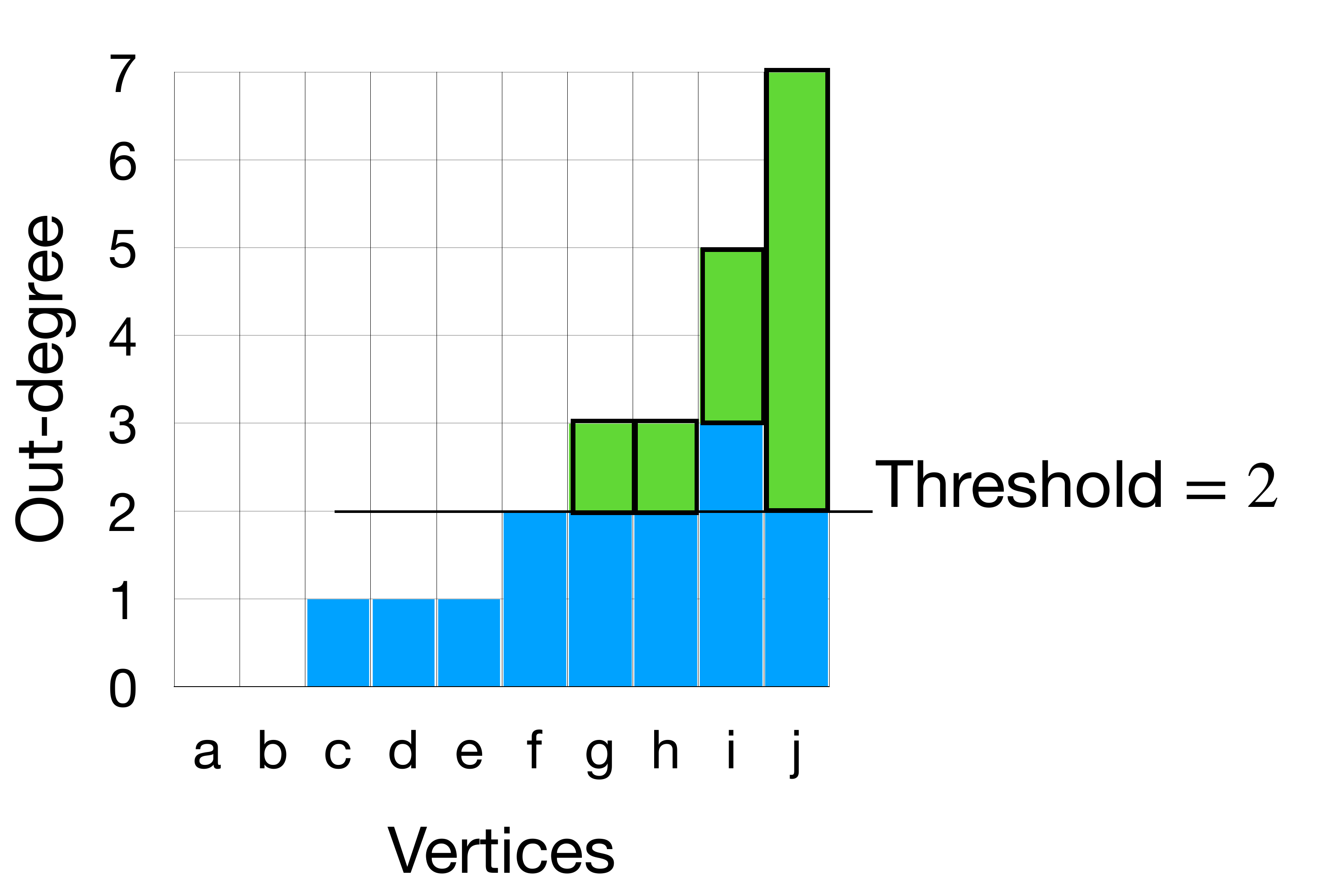}
    \end{minipage}
    
    \caption{Example graph orientation and skyline. A skyline is a subset of edges coming from high-degree vertices where the amount of edges taken depends on the vertex's out-degree. On the left, we show an orientation, where we have colored the edges in an example skyline of size 9 green. On the right we show a bar graph of the vertex out-degrees. The height of the green portion of each column equals the number of out-edges from this vertex placed in the skyline. Looking at the bar graph, note that the bottom of the green bars is (approximately) level. The threshold ($T$) of this skyline is 2.}
    \label{fig:skyline}
\end{figure}

Now we can give our algorithm. Our parallel algorithm will process an update batch $B$ as follows. First, the deletions in $B$ are removed from the graph. The insertions in $B$ are statically $3c$-oriented and inserted into the graph. 

The algorithm then runs $2\eta=2(1+1/\epsilon+2\sigma)$ iterations, each collecting a skyline of size $\frac{|B|}{2}$ and testing its height. If the skyline has sufficient height, the iteration flips it. Otherwise, the iteration performs a parallel static $3c$-orientation on that skyline and terminates early. If all $2\eta$ skylines had sufficient height, the procedure collects one more skyline $S$ of size $|B|/2$, removes $S$ from the graph and repeats on $S$ as a new insertion batch.

We achieve the following bounds.

\begin{restatable}{theorem}{mainCpluslogn} \label{thm:boundsOclogn}
     There exists a parallel batch-dynamic algorithm that for a batch update of size $b$, does $O(b\log n)$ reorientations deterministically, $O(b \log n)$ expected work, $O(b \log^2 n)$ deterministic work, $O(\log^4 n)$ depth, and maintains a $O(c + \log n)$-orientation, where $c$ is a fixed upper bound on the arboricity of the graph over the update sequence, and the updates are given by an adaptive adversary.
\end{restatable}

Upon seeing this algorithm, two natural questions arise. We briefly respond to each of them before discussing the analysis.

\begin{enumerate}
    \item {\it {\bf Why is this algorithm correct?} In particular, given that the initial skyline insertion can increase the max out-degree by $3c$, how is a counter game able to bound the maximum out-degree?} Our algorithm has the property that, if during a call of the algorithm $v$ loses out-edges and has $T$ left, then $v$ will not have significantly greater than $T$ out-edges by the end of the call. This weaker property is sufficient to apply the counter game to our algorithm.
    
    \item  {\it If skylines can be flipped in linear expected time, then the work of this algorithm is $O(|B| \eta)$ in expectation and the depth is polylogarithmic. But can skylines be extracted and flipped efficiently?} Flipping skylines in linear expected time is nontrivial. The formal definition of skylines is built so that there exist data structures that can implement them efficiently. 
\end{enumerate}

\subsubsection{Our Analysis} 

\mbox{}\\
\myparagraph{Skylines revisited} Skylines are a suitable choice for edges to flip because they satisfy the following two properties:

\begin{enumerate}[leftmargin=*]
    \item We show in Lemma \ref{lem:skylineFlipRelease} that flipping a skyline of size $x$ and sufficient height releases at least $x\epsilon$ potential. This is the parallel equivalent of a single edge flip releasing potential. 
\item We show in Lemma \ref{lem:skyline-thresholds-ij} that flipping several skylines of the same size leads to stable thresholds: if $S_1, \ldots, S_k$ are skylines of a common size $x$ with rounding $c'$, each collected after the previous one is flipped, and $T_i$ is the threshold of $S_i$, then
\begin{equation}\label{eqn:tech-stability}
        T_i \le \min_{j < i} T_j + c'
        \qquad\text{for all } i.
\end{equation}
This is a weaker parallel analog to the max out-degree only increasing by 1 per edge update in the $k$-flips algorithm. 
\end{enumerate}

\myparagraph{A batch counter game} 
As in Berglin and Brodal's analysis, we bound the maximum out-degree by tracking the
movement of potential with a counter game: we call our version the batch counter game. We find our counter game more natural to apply to the batch setting because it offers greater flexibility of counter movement.

The game involves $n$ counters, where each counter $x$ holds a non-negative real weight $u(x)$. The game is parameterized by a starting weight $Y$ (initially $u(x) = Y$ for all $x$) and a cap $H$. Given a current counter state (configuration) $u$, a move in the game is defined by a threshold $T$ and a new state $u'$, denoted as $u \to_T u'$, subject to the following rules:

\begin{enumerate}
    \item[(R1)] Counters cannot lose weight if their weight is below the threshold, and counters above the threshold cannot drop below it. Formally: $\forall x, u'(x) \ge \min(u(x), T).$    \item[(R2)] The total weight of the counters cannot increase: $\sum_x u'(x) \leq \sum_x u(x).$
    \item[(R3)] No counter may exceed the threshold by more than the cap: $\forall x, u'(x) \le T + H.$
\end{enumerate}

\begin{figure}
    \centering
    \includegraphics[width=.6\textwidth]{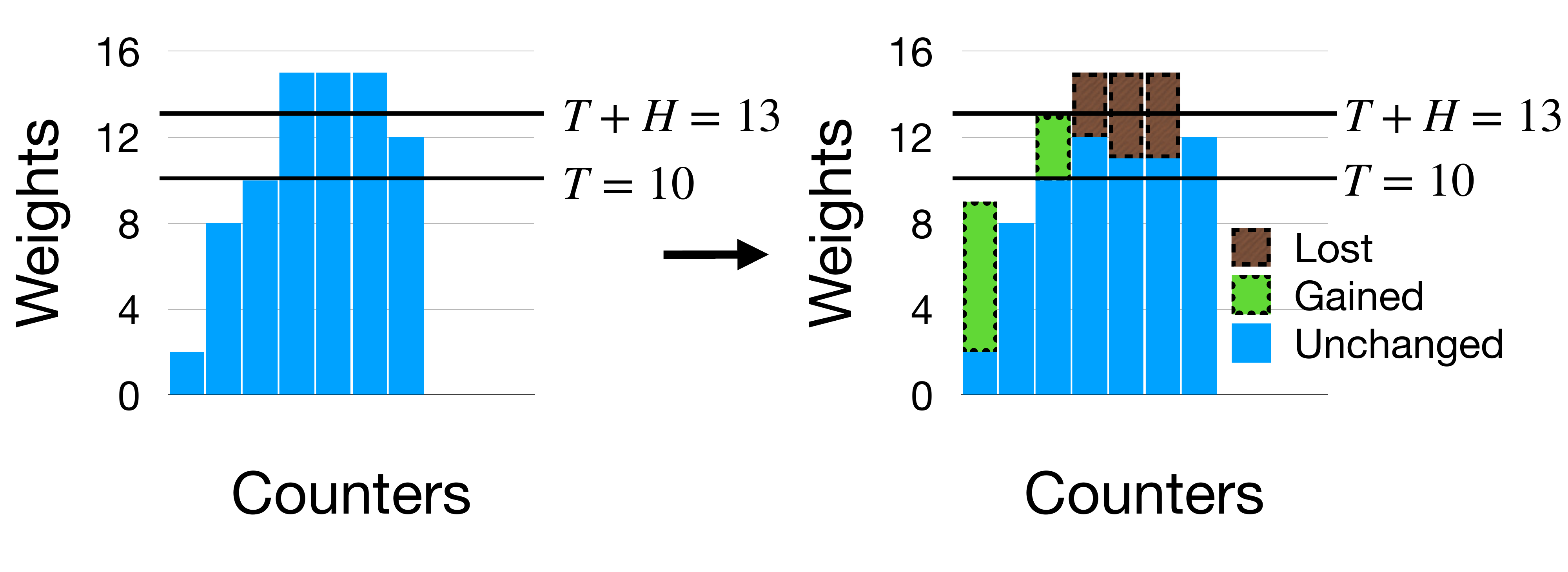}
    \caption{Example of counter movements. Here, $T=10$ and $H=3$. Note that the total amount taken from counters above $T$, which is 11, is more than the amount added to counters, which is 10. Note that all counters end up with weight under $T+H=13$, and that no counter below $T$ loses weight.}
    \label{fig:counterMoves}
\end{figure}

We illustrate a valid move in Figure \ref{fig:counterMoves}. We say that $(u_1,u_2,\dots,u_f)$ is a valid game sequence if for each step $t \in [f-1]$, there exists some $T_t$ such that $u_t \to_{T_t} u_{t+1}$. In Section \ref{sec:counter-game}, we prove the following theorem. 

\begin{restatable}{theorem}{counterLimit}
Suppose that all counters start at weight $Y > H$. The max counter weight during any valid game sequence is $O(Y + H\log n)$.
\label{thm:counterLimit}
\end{restatable}


\myparagraph{Domination} Following the structure of Berglin and Brodal, to complete the argument, we need for the counter game weights to be greater than (dominate) the vertex potentials. We design a general framework for domination, and then show our algorithm plugs into this framework. 

Suppose a configuration (set of counters) $u$ dominates the current potential $p$. The next batch $B$ is handled by a \emph{call} of the algorithm: applying $B$ and letting the offline strategy $\kappa$ perform its flips first carries the potential to a state we call $r$, and the call's response then carries it to a final potential named $q$. The step from $p$ to $r$ injects at most
$|B|\eta\epsilon$ of potential (Lemma~\ref{lem:potential-increase-per-update}), and $u$ may no longer dominate $r$. We cannot simply raise $u$ to cover the gaps without increasing the total weight and breaking rule (R2). The call must instead release potential
as it produces $q$, so that a new configuration $w$ can be built that dominates $q$ and is reachable from $u$ by a legal counter move. 

We develop the definition of \emph{$(H,T)$-bounded} to encapsulate the properties we need for an algorithm to have for the counter game to apply. An algorithm call is $(H,T)$ bounded if the following three conditions hold:

\begin{enumerate}[topsep=3pt,itemsep=2pt,label=(\textbf{J\arabic*})]
    \item $q(v) \ge \min\{r(v),\,T-\delta\epsilon\}$ for every $v$.
    \item $\sum_v q(v) - \sum_v r(v) \le -|B|\eta\epsilon$.
    \item $T \le \max_v d^+(v) \le T+H$ in the output orientation.
\end{enumerate}

Informally, our algorithm satisfies condition (J1) because we only flip skylines, (J2) is satisfied because we flip sufficiently many skylines, and (J3) because we fully remove the final skyline taken in each recursive call. We prove formally that our algorithm is $(H,T)$ bounded in Lemma \ref{lem:wcHT}.

\myparagraph{Bounding an $(H,T)$ bounded algorithm} In Theorem \ref{thm:generic-outdegree}, we show that an algorithm dominated by the counter game does achieve the desired max out-degree bound. Thus, it remains to show that any $(H,T)$ bounded algorithm is dominated by the counter game. 
Concretely, we show, in Lemma \ref{lem:domination-transfer}, that $u \to_{T-\delta\epsilon} w$ is a legal move with cap $H+\beta$. Roughly speaking, the condition (J1) keeps a vertex whose potential decreases from falling below
the threshold $T-\delta\epsilon$, fulfilling counter game rule (R1). Condition (J2) ensures the net potential change is nonpositive, satisfying the total-weight rule (R2). And (J3), with the additive
$\beta$ gap between out-degree and potential from (P1), caps the new counters at
$T+H+\beta$, which gives (R3).


\myparagraph{Efficient Data Structures} If we were willing to pay an additional $O(\log n)$ cost in the work, we could easily use binary search trees for all of our data structures. However, binary search trees are more powerful than we need: by designing data structures with weaker guarantees, we can achieve work-efficiency.

As a building block for our other data structures, we develop a deterministic parallel batch-dynamic bag data structure with constant work per update (see Section \ref{sec:bag}). This bag supports append, delete from pointers, peek (view $k$ elements), and size, but does not support contains queries. 

Instead of maintaining a queue of vertex out-edges, we observe that partial FIFO order is sufficient for the analysis, and so can maintain a weaker but faster data structure we call a pannier (see Sections \ref{sec:dataStructures} and \ref{sec:potential-setup-main} for more discussion). 

The most expensive part of finding a skyline is determining which vertices have high out-degree: naively this would require a sorted list. We design a roughly sorted list (RSL) that will return the vertices with highest out-degree with some additive imprecision: this error term we call the rounding $c'$ of the RSL (a fully sorted list would yield no error, $c'=1$). An RSL can be maintained with constant work per update in our setting. The imprecision in the RSL leads to imprecision in the skylines, but this can be handled by statically orienting the insert batch and by slightly increasing the cap $H$ used in the $(H,T)$-bounded proof. See Sections \ref{sec:dataStructures} and \ref{sec:appendix_skyline} for more details on the RSL.

%% file: 3_background.tex
\newpage 
\section{Preliminaries}\label{sec:prelims}

\subsection{Model and Primitives}
\myparagraph{Parallel Model}
We use the work-span (or work-depth) model for binary fork-join parallelism with atomics to analyze parallel algorithms~\cite{CLRS,BlellochF0020,GJS21}.
The model assumes a set of \thread{}s that share memory.
A \thread{} can \forkins{} two child \thread{s} that run in parallel.
When all children complete, the parent \thread{} continues.
The \defnn{work} $W$ of an algorithm is the total number of instructions and
the \defnn{span} (depth) $S$ is the length of the longest sequence of dependent instructions.
Computations can be executed using a randomized work-stealing
scheduler in practice in $W/P+O(S)$ time with high probability on $P$ processors~\cite{BL98,ABP01,gu2022analysis}. The atomic operations Test and Set (TS) and Compare and Swap (CAS) are permitted. 

There are several variants of work-span models, including the CRCW PRAM~\cite{KarpR90}.   Importantly, however, many of
these models are robust with respect to each other in terms of work.   In particular, it is possible to cross simulate
these models with the same asymptotic work up to constants, at least assuming randomization~\cite{BlellochF0020,GJS21}. 
The \depth{} among models, however, can differ by a logarithmic factor (e.g., simulating n-way forking of a PRAM with binary forking requires a logarithmic increase in \depth{}).   We therefore state precise asymptotic bounds
 for work, but focus just on polylogarithmic \depth{}.

\myparagraph{Parallel Batch-Dynamic Algorithms}
Our algorithms operate in the parallel batch-dynamic setting. In this setting, we have a sequence of \emph{batch} updates $\mathcal{B}=\{B_i\}$ (edge insertions/deletions) and a sequence of $t+1$ undirected graphs $G_0,G_1,\ldots,G_t$, where $G_0$ is empty, and each $G_i$ was generated from $G_{i-1}$ with a batch edge update $B_i=E_{i-1} \triangle E_i$. After each batch update on the graph, we must update the data structure we are maintaining (i.e. the orientation) to the new graph.
We let $b_i=|B_i|$. When clear, we refer to a particular batch update in the sequence as $B$, and let $b=|B|$. Note that the batches can have different sizes. We assume that a batch update from the user will not insert edges already in the graph, nor delete edges not present, and will have no duplicate edges. We assume that batch deletes give the pointer to the edge they are deleting.

 Similar to design goals for efficient 
parallel algorithms, the gold standard for parallel batch-dynamic algorithms is
to design algorithms that are {\em work-efficient} (i.e., asymptotically perform 
the same work as the best sequential dynamic algorithms that processes the updates 
in a batch one-at-a-time) and have low (ideally poly-logarithmic) \depth{} in the
worst case.


Our adaptive adversary is able to see the out-degree orientation we have chosen after each time step, and may pick the next batch updates based on this knowledge. However, our adaptive adversary may not see random bits that we have not yet requested. This permits us to use a randomized semisort as a subroutine.

We use the notation $[x]$ to mean $\{1,2,\ldots,x\}$.

\myparagraph{Parallel primitives}
The following parallel procedures are used throughout the paper.
\defnn{Scan} (prefix sum) takes as input an array $A$ of length $n$, an associative
binary operator $\oplus$, and an identity element $\bot$ such that
$\bot \oplus x = x$ for any $x$, and returns the array
$(\bot, \bot \oplus A[0], \bot \oplus A[0] \oplus A[1], \ldots, \bot \oplus_{i=0}^{n-2} A[i])$
as well as the overall sum, $\bot \oplus_{i=0}^{n-1} A[i]$.
Scan can be done in $O(n)$ work and $O(\log n)$ \depth{} (assuming $\oplus$
takes $O(1)$ work)~\cite{jaja1992parallel}.
\defnn{Reduce} takes an array $A$ and a binary associative function
$f$ and returns the sum of the elements in $A$ with respect to $f$.
\defnn{Filter} takes an array $A$ and a predicate $f$ and returns a new array
containing $a \in A$ for which $f(a)$ is true, in the same order as in $A$.
Reduce and filter can both be done in $O(n)$ work and $O(\log n)$ \depth{}
(assuming $f$ takes $O(1)$ work). Performing a \defnn{map} operation over the elements of an array, i.e., applying some function $f$ to each element, can be done in the same work and \depth{} bounds assuming $f$ takes $O(1)$ work.

Given a sequence $A$ of entries, each associated with a \emph{key}, a \defnn{semisort} reorders $A$ such that all
entries with the same key are consecutive. Note that the keys may remain out of order, so a semisort is weaker than a full sort.  There are classic algorithms for semisort with $O(n)$ expected
work and $O(\log n)$ \depth{} \whp{}~\cite{gu2015top, BlellochF0020}.\footnote{We say $g(n,c) \le O(f(n))$ \emph{with high probability in $n$} (\whp{}) if there exist constants $k,c_0$ such that for $c \ge c_0$ there exists $n_0$ such that for $n \ge n_0$, we have that $g(n,c) \le k c f(n)$ with probability at least $1-\frac{1}{n^c}$.} 
Using a semisort, we can also implement a \defnn{removeDuplicates} (remove all repeated elements in an array) and a \defnn{groupBy} (given a sequence of pairs, group them by their first element) in the same work and \depth{}.
Deterministically, sorting $n$ elements can be done via a parallel merge sort in $O(n\log n)$ work and $O(\log n)$ \depth{} in our model~\cite{goodrich2023optimal}. 

Note that we can run a randomized semisort with a work and depth budget in mind: if the semisort exceeds this budget (either exceeds a work budget of $O(n \log n)$ or a depth budget of $O(\log n)$), we then scrap the results and run the deterministic sort. We therefore get two bounds, a bound in expectation and a bound deterministically, for the same algorithm. This means that in expectation, the sort will do $O(n)$ work, but that the sort will never exceed $O(n \log n)$ work. Furthermore, the depth is deterministic, because of our depth budget.  

\begin{lem} \label{lem:sort}
    We can semisort a list of $n$ keys in $O(n)$ work in expectation such that deterministically the work is $O(n \log n)$ in the worst scenario. The depth will be deterministic $O(\log n)$. 
\end{lem}

Note that various operations must be semisorted (or sorted) before occurring, e.g., semisorting a batch of edge inserts before batch inserting in each vertex's bag. We do not explicitly write these underlying semisorts (or sorts if deterministic) in the pseudocode for clarity, but they do contribute to our runtime bounds.

\subsection{Parallel Data Structures \label{sec:dataStructures}}

We will use the following data structures in our algorithms. 
Note that we choose to 1-index (rather than zero-index) all data structures throughout this work, except for the roughly sorted list. All arrays will be indexed inclusive on both ends. For an array $A[i:j]$, if $i > j$, then this returns the empty set, not an error, by convention. 

\myparagraph{Bag} A bag is a container structure that must support \textsf{batchInsert}$(L)$, \textsf{batchDelete}$(L)$, \textsf{batchPop}$(b)$, \textsf{size}, and \textsf{peek}$(b)$ (output an arbitrary $b$ elements of the bag, if they exist) operations, where $L$ is an array of arbitrary length (the length can differ between different calls) and $b$ is an arbitrary value. Note that a bag does \emph{not} support a contains nor isMember query. Note that we require peek to be deterministic and to have no side effects, so two consecutive calls to peek will return the same values. However, a call to batchInsert, batchDelete, or batchPop is allowed to completely reorder the bag.

We will use bags to store out-edges in various places, including in the amortized algorithm and within the static orientation. In the pseudocode, we will implicitly convert arrays to bags and back when necessary. 

We can implement a bag with constant deterministic time (no amortization or randomness) per operation using a linked list with skew numbers \cite{MYERS83,okasaki1999purely}, see Section \ref{sec:bag} for the full data structure and proofs. Note that a hash table would be constant amortized expected time and is thus not sufficient for our purposes, as we want a data structure that is worst-case and deterministic.

\begin{restatable}{lem}{bag} \label{lem:bag} There exists a deterministic bag with maximum capacity $poly(n)$ that for a batch of $b$ operations requires $O(b)$ work and $O(\log n)$ \depth{}.\end{restatable} 



\myparagraph{Pannier} A pannier is a pair of bags, a front bag and a back bag. The size of the pannier is the sum of the sizes of the bags. When peeking or popping from the pannier, the front bag is utilized completely before the back bag is accessed. 
Note that the pannier is, at a high level, similar to the batched-queue \cite{MH81, BURTON82}, though the access patterns are less restrictive for the pannier.
Although a pannier is just two bags, we give a pannier a special name because we use the pannier in very specific ways in our algorithms and analysis (essentially as a weakened queue). We further motivate the pannier in Section \ref{sec:potential-setup-main}. 



\begin{lem} \label{lem:tsb}
Implemented with deterministic bags, for a batch $b$ of updates, a pannier with maximum capacity $poly(n)$ takes $O(b)$ work and $O(\log n)$ \depth{}. 
\end{lem} 

\myparagraph{Roughly Sorted List} To store the relative ordering of vertices by out-degree, we define a roughly sorted list (RSL). An RSL stores $n$ key-value pairs $(v,d)$. In our use of the RSL, the keys will be vertices, and the value will be the vertex's out-degree. An RSL supports two (batch) operations, update and prefix. The update$(v,d')$ operation updates $v$'s value to $d'$. Informally, the prefix$(w)$ operation will return an array of the $w$ keys with the highest values, up to some additive error. 

Formally, let $c'$ be the rounding of the RSL and $M \le O(\log n)$ be the budget of the RSL. Define the operation $r(x,c')=\lfloor \frac{d^+(x)}{c'} \rfloor c'$ (for round down to a multiple of $c'$). On calling prefix, a roughly sorted list will return an array of keys $S$ of size $w$ such that for all keys $u,v$, 
\begin{multline} (r(\text{val}(u),c') > r(\text{val}(v),c') \land v \in S \implies u \in S \lor \text{val}(v) > Mc' \lor \\ (\text{val}(u)=r(\text{val}(u),c') \land r(\text{val}(u),c') = r(\text{val}(v),c') + c') \end{multline}

That is, if $u$ has rounded value strictly higher than $v$, then if $v$ is included, $u$ must also be included. But this requirement does not hold if both $u$ and $v$ have values that are sufficiently high: if keys have very high value, then any keys can be returned. 

We achieve the following work and \depth{} bounds for this data structure.

\begin{lem} \label{lem:rsl}
A roughly sorted list (RSL) with budget $M \in O(\log n)$ takes $O(w)$ work and $O(\log n)$ \depth{} to output a prefix of size $w$. To batch update the values of $w$ keys, the RSL takes $O(w)$ work in expectation and $O(w \log w)$ work deterministically. The RSL will take $O(\log n)$ \depth{}. The RSL has maximum capacity $poly(n)$.

\end{lem} 

\begin{proof}
A roughly sorted list can be implemented with $M+2$ bags. Bags $i \in [1,M]$ hold the vertices with values in range $((i-1)c',ic']$. Bag $0$ holds the vertices with value 0, and bag $M+1$ holds all of the vertices with value greater than $Mc'$. We call bag $M+1$ the high bag or special bag. Inserts/deletes to the RSL are first semisorted by the target bag then applied to the appropriate bag. To update the values of vertices, they are removed, the value changed, then reinserted. 

We keep a machine word of length $M+2$, which is feasible because $M \le O(\log n)$. We mark a bit with a one if the corresponding bag is nonempty, and with a zero otherwise. Using standard bit tricks, we can determine the nonempty bag with largest id less than $i$, for any $i \in [0,M+1]$. Because at most $w$ bags change empty status, we can set the corresponding bits in this machine word in $O(w)$ work and $O(\log w)$ \depth{}.

For a prefix call, a running prefix sum will be computed from highest to lowest bag until the requisite number of vertices is found. This prefix sum, importantly, can skip empty bags, by leveraging bit instructions on the machine word.

Then, all vertices will be peeked from the highest bags, and an arbitrary number of vertices (enough to fill the remainder) will be taken from the lowest bag still selected in the prefix sum. The prefix sum will take $O(\min(w,M))=O(\log n)$ \depth{}. The work is bounded by the output times of the bags, which is constant per item examined. Note that no semisort is needed for taking a prefix from an RSL.
\end{proof}

Note that a roughly sorted list saves a log factor compared to a fully sorted list implemented with balanced binary search trees, at the cost of being ``roughly sorted" instead of completely sorted. The imprecision in roughly sorting will eventually lead to a slightly worse constant in our worst-case algorithms' max out-degree bound (in particular, in the bounds in Lemma \ref{lem:skyline-thresholds-ij}, which is used in Lemma \ref{lem:update-call-bounded-reinsertion}, which then plugs into machinery we generate to give the out-degree bound).

\subsection{Edge orientations}

Consider an undirected graph $G=(V,E)$. 
The vertices are stored in an array. The edges are stored in two places. There is a bag containing all edges, and each vertex stores a bag of its incident edges (both its out and in edges). Additionally, if the graph is oriented, each vertex keeps a pannier (or a queue, depending on the algorithm) of its out edges. Each edge knows its endpoint vertices and, if oriented, the direction it points (tail points to head). Each vertex knows its degree and out-degree. 

\begin{definition}[Arboricity] 
The arboricity $c$ of a graph $G=(V,E)$ is  
$$\max_{S \subseteq V,|S| \ge 2} \left\lceil \frac{|E(S)|}{|S|-1} \right\rceil.$$ 
By the Nash-Williams theorem (\cite{NashWilliams1961EdgeDisjoint}), the arboricity is also the minimum number of edge-disjoint spanning forests needed to cover a graph.
\end{definition}

Arboricity is the notion of graph density that we use in this work. We let $c$ be an upper bound for the arboricity of $G$ over the sequence of batch updates.


\begin{table}[h]
\centering
\begin{tabular}{c|c}
Symbol & Meaning \\
\hline
$n$ & number of vertices, $|V|$ \\
$m$ & number of edges, $|E|$ \\
$c$ & arboricity \\
$B$ & batch of edge updates \\
$b$ & $|B|$ \\
$\delta^+(x)$ & out edges of vertex $x$ \\
$\delta^+(X)$ & out edges of all vertices in set $X$ \\
$\delta_A^+(x)$ & $\delta^+(x) \cap A$ \\
$d^+(x)$ & $|\delta^+(x)|$ \\
$d^+_A(x)$ & $|\delta_A^+(x)|$  \\

$A \triangle B$ & symmetric difference of sets $A$ and $B$ \\
$[x]$ & $\{1,2,\ldots,x\}$ \\
$r(x,c')$ & rounding function: $\lfloor x/c' \rfloor c'$ \\
$M$ & budget of the roughly sorted list \\
$c'$ & rounding of the roughly sorted list \\
\end{tabular}
\caption{General notation}
\label{tab:notation}
\end{table}

\subsection{Static orientation algorithm \label{sec:staticOrientation}}

For computing a static orientation, we use the algorithm of Barenboim and Elkin~\cite{barenboim2008sublogarithmic}, which was designed for the distributed setting, adapted to our parallel setting; see Algorithm~\ref{alg:staticorient}. Given a bag of directed edges, the algorithm produces a $(2+\epsilon)\cdot c$ orientation in low work and \depth{}, as shown in the following lemma. We present the pseudocode and a proof of runtime bounds in the appendix. 

\begin{restatable}{lem}{staticLem}[Adapted from \cite{barenboim2008sublogarithmic}] Given a bag of directed edges $E$, and $\epsilon \in (0,2]$, Algorithm~\ref{alg:staticorient} outputs a bag of directed edges such that no vertex has out-degree more than $(2+\epsilon)\cdot c$, where $c$ is the arboricity of the graph induced on $E$. The algorithm runs in $O(\epsilon^{-1} |E|)$ expected work and $O(\epsilon^{-1} \log^2 |E|)$ \depth{}. The algorithm also runs in $O(\epsilon^{-1} |E| \log |E|)$ deterministic work.
\label{lem:static-orient}
\end{restatable}

%% file: 4_amortized_algo.tex
\section{Optimal Amortized Parallel Batch-Dynamic Orientation}\label{sec:amortized}

In this section, we present an amortized parallel batch-dynamic algorithm (Algorithm~\ref{alg:amortized}) for maintaining a low out-degree orientation of a dynamic graph under batch updates. In particular, we prove the following theorems.

\amortizedmain*

\amortizedoptimal*

\subsection{Algorithm Description}
Our approach adapts the classic sequential algorithm of Brodal and Fagerberg~\cite{BF99} to the parallel batch-dynamic setting. For a batch of edge insertions, the algorithm first adds the new edges, assigning them an arbitrary orientation. A ``repair'' step is then performed to correct the out-degrees of vertices violating the degree bound. However, instead of iteratively flipping all outgoing edges of violating vertices as in \cite{BF99}, we observe that performing a single static orientation on the local subgraph induced by the violating vertices and all their out-edges is sufficient and yields equivalent asymptotic bounds.

More specifically, we define two degree bounds: a \emph{cutoff threshold} $\threshcutoff$ and a \emph{static threshold} $\threshstatic$. After inserting the batch of edges, the algorithm performs a parallel static $\threshstatic$-orientation on the subgraph induced by vertices with out-degree greater than $\threshcutoff$ and all their out-edges. This guarantees that each of these violating vertices has its out-degree reduced to at most $\threshstatic$, while ensuring the out-degree of any other vertex is bounded by $\threshcutoff + \threshstatic$.
Batch deletions are handled straightforwardly: edges in the batch are simply removed from the graph as this cannot increase any out-degrees.

The full batch-dynamic update algorithms are detailed in Algorithm~\ref{alg:amortized}. 
Implementation-wise, in addition to storing vertices in $V$, we maintain two bags, $\vlow$ and $\vhigh$, that partition the vertices based on their out-degree: vertices with out-degree greater than $\threshcutoff$ are placed in $\vhigh$, and those with out-degree at most $\threshcutoff$ are placed in $\vlow$. Whenever a vertex's out-degree changes, we update its bag membership accordingly.

\begin{algorithm}[ht]
\caption{Optimal Amortized Algorithm}\label{alg:amortized}

\function{\textsc{Update}$(\overline{G},B,\epsilon)$}{
  \If {$B$ is insertion batch} {
    Arbitrarily orient the edges in $B$ and add them to the graph \;
    Update $\vlow$ and $\vhigh$\label{line:amor_high}\;
    \tcp{Statically orient out-edges of $\vhigh$}
    Let $E' = \delta^+(\vhigh)$\label{line:amor_high_edges}\;
    Run \textsc{StaticOrientation}$(E',c,\epsilon)$ and update $E$, $\vlow$, $\vhigh$.\label{line:amor_ins_upd}

  }
  \Else {
    Remove $B$ from graph \;
    Update $\vlow$ and $\vhigh$\label{line:amor_del_upd}\;
  }

}

\end{algorithm}

\subsection{Analysis}

\begin{table}[]
    \centering
    \begin{tabular}{c|c}
    Symbol & Meaning \\ \hline
         $\threshcutoff$ & cutoff threshold (high vs. low vertices) \\
         $\threshstatic$ & static orientation quality \\
         $\threshoffline$ & offline orientation quality \\
    \end{tabular}
    \caption{Selected notation used in Section \ref{sec:amortized}.}
    \label{tab:notationAmortized}
\end{table}
Our analysis largely follows that of Brodal and Fagerberg, though with minor modifications for our setting \cite{BF99}. We analyze our algorithm's performance against an optimal offline strategy using a potential function.
Consider a sequence of $t$ arboricity $c$ preserving batch updates $S=(B_1,B_2,\ldots,B_t)$. Let $G_i = (V, E_i)$ be the graph after Algorithm~\ref{alg:amortized} has processed the first $i$ batches, with $G_0$ as the empty graph. 
For comparison, we consider an offline strategy that maintains a $\threshoffline$-orientation, producing the graph $G_i^* = (V, E_i^*)$ after processing the first $i$ batches. Let $r^*$ be the total number of reorientations performed by this offline strategy over the entire sequence $S$. 

Similar to Brodal and Fagerberg \cite{BF99}, we define an edge $e=(u\to v) \in E_i$ to be \emph{bad} if its orientation differs from its orientation in $G_i^*$ (i.e., $e \notin E_i^*$). Otherwise, the edge is \emph{good}. Define the following potential function:
\begin{align*}
\Phi_i := |\{e \in E_i : e \notin E_i^*\}|,
\end{align*}
i.e., $\Phi_i$ counts the number of bad edges in $E_i$. Initially, $\Phi_0 = 0$ because the initial graph is empty.

During the repair step for batch $B_i$, let $\Phi_i^\downarrow(v)$ denote the decrease in potential resulting from the reorientation of the out-edges of vertex $v$ by the static orientation step. 
The total decrease in potential for the batch during this step is given by $\Phi_i^\downarrow = \sum_{v\in V}\Phi_i^\downarrow(v)$. 
We now bound the size of the subgraph sent to static orientation by Algorithm~\ref{alg:amortized} in terms of the potentials.

\begin{lemma}[Related to Lemma 1 in \cite{BF99}] \label{lem:amor_subgraph_size}
Assuming $B_i$ is a \bIns, let $E_i'$ denote the set of edges computed in Line~\ref{line:amor_high_edges} of Algorithm~\ref{alg:amortized}. Then,
\begin{align*}
    |E_i'| \le \frac{\threshcutoff+1}{\threshcutoff+1-2\threshoffline-\threshstatic}\Phi_i^\downarrow.
\end{align*}
\end{lemma}

\begin{proof}[Proof of Lemma \ref{lem:amor_subgraph_size}]
Let $\vhigh_i$ denote the set $\vhigh$ after adding the input batch of edges $B_i$ at Line~\ref{line:amor_high} of Algorithm~\ref{alg:amortized}.
    Consider a vertex $v\in \vhigh_i$. By definition, $d_i^+(v) \ge \threshcutoff+1$. Among the $d_i^+(v)$ out-edges, at most $\threshoffline$ of them are good before the repair step, implying that at most $\threshoffline$ edges become bad after the repair step. Of the at least $d_i^+(v)-\threshoffline$ bad edges, at least $d_i^+(v)-\threshoffline-\threshstatic$ edges are flipped by the repair step, implying that at least $d_i^+(v)-\threshoffline-\threshstatic$ edges become good. Therefore, the decrease in potential corresponding to vertex $v$ is,
    \begin{align*}
        \Phi_i^\downarrow(v) &\ge d_i^+(v)-\threshoffline-\threshstatic-\threshoffline \\
        &\implies d_i^+(v) \le \Phi_i^\downarrow(v) + 2\threshoffline + \threshstatic.
    \end{align*}
    Summing over all vertices in $V_i^\text{high}$, we get
    \begin{align*}
        (\threshcutoff+1)|V_i^\text{high}| &\le \sum_{v\in V_i^\text{high}}d_i^+(v) \le \Phi_i^\downarrow + (2\threshoffline+\threshstatic)|V_i^\text{high}|\\
        &\implies |V_i^\text{high}| \le \frac{\Phi_i^\downarrow}{\threshcutoff+1-2\threshoffline-\threshstatic}
    \end{align*}
    Finally, we have
    \begin{align*}
        |E_i'| = \sum_{v\in V_i^\text{high}} d_i^+(v) &\le \Phi_i^\downarrow + (2\threshoffline+\threshstatic)|V_i^\text{high}|\\
        &\le \frac{\threshcutoff+1}{\threshcutoff+1-2\threshoffline-\threshstatic}\Phi_i^\downarrow.
    \end{align*}
\end{proof}

Given this, we now bound the total number of flips performed by Algorithm~\ref{alg:amortized} over the sequence $S$. The following lemma is also similar to Lemma~1 in \cite{BF99}.

\begin{lemma}\label{lem:amor_work_bound}
    Suppose we have an arboricity $c$ preserving sequence of $t$ batch updates $\mathcal{B} = (B_1,B_2,\ldots,B_t)$, where the total number of edge insertions and deletions are $t_\text{ins}$ and $t_\text{del}$, respectively. Then, for $\threshoffline, \threshcutoff,\threshstatic > c$ with $\threshcutoff>2\threshoffline+\threshstatic-1$, Algorithm~\ref{alg:amortized} sends at most 
    \begin{align*}
        (t_\text{ins}+r^*)\frac{\threshcutoff+1}{\threshcutoff+1-2\threshoffline-\threshstatic}
    \end{align*}
    edges to be statically oriented, assuming there exists an offline strategy that maintains a $\threshoffline$-orientation and performs at most $r^*$ reorientations when processing the batches in sequence $S$.
\end{lemma}

\begin{proof}[Proof of Lemma \ref{lem:amor_work_bound}]
    Observe that edge insertions and edge reorientations performed by the offline strategy (for both insertions and deletions) are the only sources of increase in the potential. Edge deletions and the repair step in \bIns do not increase the potential. Thus, the overall decrease in potential is bounded by the total increase in potential across all batches, which is at most $t_\text{ins}+r^*$. Therefore, we have
    \begin{align*}
    \sum_{i=1}^t\Phi_i^\downarrow \le t_\text{ins}+r^*.
    \end{align*}

    The number of edges sent by Algorithm~\ref{alg:amortized} to be statically oriented is at most $\sum_{i=1}^t |E_i'|$.
    Therefore, by Lemma~\ref{lem:amor_subgraph_size}, we have
    \begin{align*}
        \sum_{i=1}^t|E_i'| &\le \sum_{i=1}^t \frac{\threshcutoff+1}{\threshcutoff+1-2\threshoffline-\threshstatic}\Phi_i^\downarrow\\
        &= \frac{\threshcutoff+1}{\threshcutoff+1-2\threshoffline-\threshstatic}\sum_{i=1}^t\Phi_i^\downarrow\\
        &\le \frac{\threshcutoff+1}{\threshcutoff+1-2\threshoffline-\threshstatic}(t_\text{ins}+r^*).
    \end{align*}
\end{proof}

Note that we have an extra $-\threshstatic$ term in the denominator above compared to \cite{BF99}. This is because we do not flip all edges being repaired; instead, we use static orientation. Due to this term, we end with a slightly worse degree bound.

Brodal and Fagerberg present an offline strategy that performs few flips: 

\begin{lemma}[Lemma~3, \cite{BF99}] \label{lem:offline}
Let $\threshoffline > c$. There exists an offline strategy that maintains a $\threshoffline$-orientation and performs $\lceil\log_{\threshoffline/c}(n)\rceil$ reorientations on deletions and no reorientations on insertions. There also exists an offline algorithm that performs no reorientations on deletes and $\lceil\log_{\threshoffline/c}(n)\rceil$ reorientations on insertions.
\end{lemma}

Putting everything together, we can prove Theorems~\ref{thm:amortized_main} and \ref{thm:amortized_optimal}.

\amortizedmain*

\begin{proof}[Proof of Theorem \ref{thm:amortized_main}]
Let $\epsilon'=\epsilon/6$, and set $\threshoffline=(1+\epsilon')c$, $\threshstatic=(2+\epsilon')c$ and $\threshcutoff=2\threshoffline+\threshstatic+2\epsilon'c$. 
Then, the orientation quality will be 
\begin{align*}
    \threshcutoff + \threshstatic = 2\threshoffline+2\threshstatic+2\epsilon'c = (6+\epsilon) c.
\end{align*}

By Lemma \ref{lem:amor_work_bound} and \cite[Lemma~3]{BF99}, the total number of edges that participate in static orientation, across all batches, is no more than 
\begin{align*}
    O\left(\left(t_\text{ins}+t_\text{del} \log_{1+\epsilon'} n\right) \cdot\frac{\threshcutoff+1}{\threshcutoff+1-2\threshoffline-\threshstatic}\right) &= O\left(\left(t_\text{ins}+t_\text{del}\log_{1+\epsilon'} n\right)\cdot \frac{(4+5\epsilon')c+1}{2\epsilon'c + 1}\right) \\
    &= O\left(\epsilon^{-1}t_\text{ins}+\epsilon^{-2}t_\text{del} \log n\right). 
\end{align*}

Note that even though no flips are performed on batch deletions, in our amortized charging scheme we assign flips and work to delete steps. Since edges are flipped only during static orientation, the amortized number of flips performed for a batch of size $b$ is $O(\epsilon^{-1}b)$ for batch insertions, and $O(\epsilon^{-2}b\log n)$ for batch deletions. 

With randomized primitives, static orientation takes expected work $O(\epsilon^{-1}|E|)$ on a set of edges $E$. Therefore, the total expected work associated with static orientations is $O(\epsilon^{-2}t_\text{ins}+\epsilon^{-3}t_\text{del} \log n)$. We account for this work by charging each batch insertion $O(\epsilon^{-2})$ work per edge in the batch, and each batch deletion $O(\epsilon^{-3}\log n)$ work per edge in the batch. The work for updating the vertex bags $\vlow$ and $\vhigh$ after static orientations (i.e., line~\ref{line:amor_ins_upd}) is also proportional to the size of subgraph sent for static orientation. Therefore, we charge batch insertions $O(\epsilon^{-1}b)$ and batch deletions $O(\epsilon^{-2}b \log n)$ to account for this work. 

We charge the work of inserting or deleting edges, as well as the work of updating the bags before static orientation, directly to the update step in which this work occurs. Inserting or deleting edges from the graph takes expected $O(b)$ work. Updating the bags $\vlow$ and $\vhigh$ (i.e., lines~\ref{line:amor_high} and \ref{line:amor_del_upd}) also takes expected $O(b)$ work for both batch insertions and deletions. 

Thus, the overall amortized expected work bounds are $O(\epsilon^{-2}b)$ and $O(\epsilon^{-3}b\log n)$ for batch insertions and deletions, respectively. The \depth{} for batch insertions is dominated by the \depth{} of the static orientation algorithm, which is worst-case $O(\epsilon^{-1}\log^2n)$ \whp{}, whereas for batch deletions it is worst-case $O(\log n)$ \whp{} to delete the edges and update the bags.

The bounds using a deterministic sort are obtained analogously. In particular, we incur an additional factor of $O(\log n)$ in the work.



\end{proof}

\amortizedoptimal*

\begin{proof}[Proof of Theorem \ref{thm:amortized_optimal}]
The optimality in terms of the number of flips follows from Lemma~\ref{lem:amor_work_bound}, where we show that the amortized number of flips performed by Algorithm~\ref{alg:amortized} is proportional to the 
number of flips performed by any sequence of flips that maintains a bounded out-degree (i.e., $\threshoffline$). Thus, it is $O(1)$-competitive to any algorithm maintaining an $O(c)$ orientation. Furthermore, Algorithm~\ref{alg:amortized}, when implemented with randomized primitives, runs in work linear in the number of flips performed, and is therefore optimal, because at least constant work must be performed per flip. 
\end{proof}


%% file: 5_algorithm.tex
\section{Worst-Case Algorithm}

In this section, we present our work-efficient worst-case algorithm and prove its work and depth bounds. We will analyze the max out-degree of the algorithm in later sections.

\subsection{Skylines \label{sec:skyline}}

First, we must introduce skylines, a core primitive we will use. Skylines will allow us to flip multiple edges from high out-degree vertices in parallel.

Informally, a skyline is a set of out-edges, where the number of out-edges taken from a vertex $v$ increases with $v$'s out-degree, see Figure \ref{fig:approx_skyline} for an example. 
In Section \ref{sec:appendix_skyline}, we show how to extract a skyline of size $x$ in $O(x)$ expected work and polylog depth. 
We defer a more complete discussion of skylines in order to more promptly present our orientation algorithms.


\begin{definition}\label{def:skyline-by-threshold}[Skyline by threshold] A subset $S \subseteq E'$ is called a \emph{skyline} above threshold $T$ with rounding $c'$ if $c' | T$ and 
\begin{equation*}
    \forall v \in V : \quad S \cap \delta^+(v) = Q_v[1 : d^+(v)-T],\footnotemark%
\end{equation*}
\end{definition} 

\footnotetext{Recall that by our inclusive 1-indexing,  $Q_v[1:1]$ is the first out-edge in $v$'s out-edge container.}

\begin{definition}\label{def:approx-skyline-by-size}[Skyline by size] Fix $x$ and choose the (unique) integer threshold $T$, where $c' | T$, such that $C_T \ge x$ but $C_{T+c'} < x$. A subset $S \subseteq E'$ is called a skyline of rounding $c'$ and size $x$ if
\[ \forall v \in V : \quad S \cap \delta^+(v) = Q_v[1 : \max(d^+(v)-T-c',0)+\rho_v] \] 
where $\rho_v \in \mathbb{Z}[0,c']$, $\rho_v \le d^+(v)-T$ and $\sum_{v \in V} \rho_v = x - C_{T+c'}$.  We say that the threshold $T$ was induced by the choice of size $x$. 
\end{definition}

\begin{definition}\label{def:skyline-demand}[Demand]
In a skyline $S$, the \emph{demand} on vertex $v$ is the number of out-edges of $v$ present in the skyline, $|S \cap \delta^+(v)|$. 
\end{definition}

\begin{definition}\label{def:sufficient-height}[Sufficient Height] 
Given a skyline $S$ with threshold $T$ (induced or natural), we say $S$ has \textit{sufficient height} with respect to $\delta$ if $T \ge 4\delta$. 

\end{definition}

\begin{figure}
    \centering
    \begin{minipage}{.35\textwidth}
    \includegraphics[width=\textwidth]{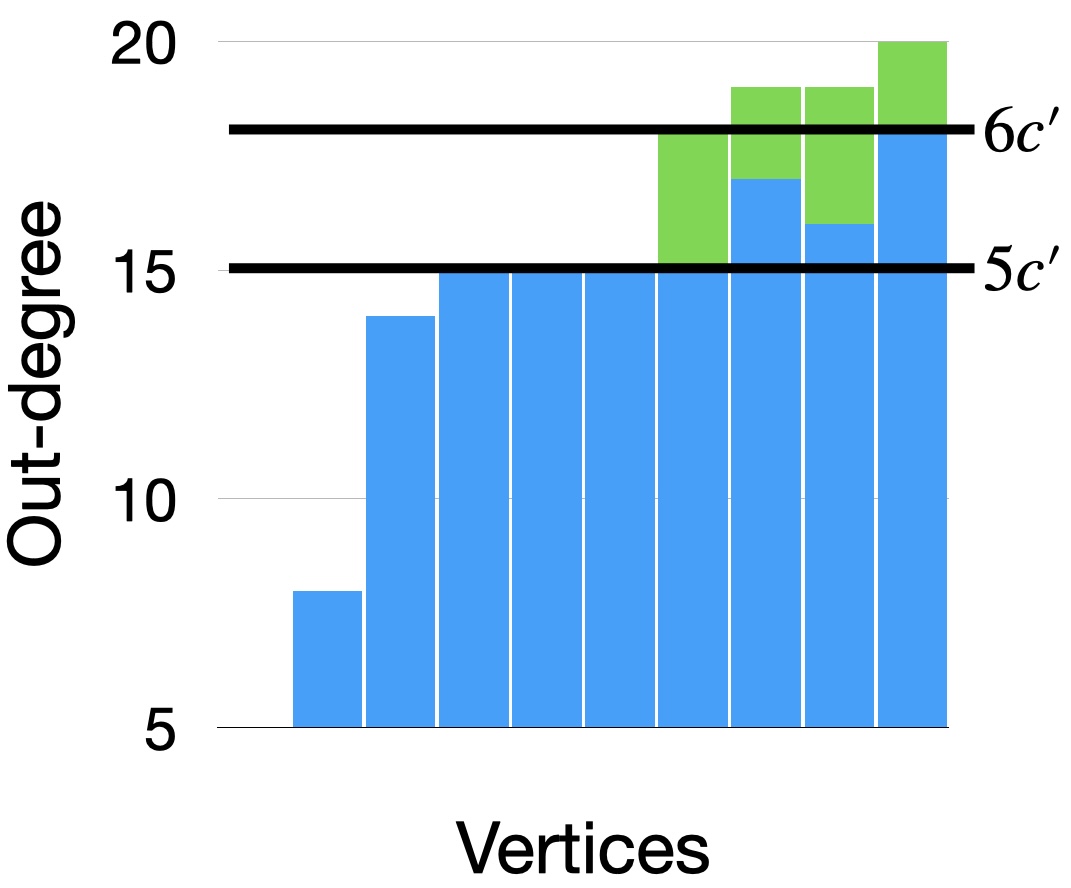}
    \end{minipage}
    \hspace{.1in}
    \begin{minipage}{.5\textwidth}\caption{Example of finding a skyline of size $10$, where $c'=3$. Note that we take all edges above $6c'$ but that we made an arbitrary choice for the edges between $5c'$ and $6c'$.}
    
    \label{fig:approx_skyline}
    \end{minipage}
\end{figure}

When we say that we ``flip'' a skyline of size $x$, we usually refer to the combined actions of finding a skyline of size $x$, removing it from the graph, flipping the edges, and reinserting the skyline back into the graph. Flipping a skyline is the core action our worst-case algorithms take to reduce the maximum out-degree of the graph.
For our worst-case algorithms, it will be important that when we take several skylines of the same size in a row, the overall change in the threshold induced by these skylines is small, as expressed by the following lemma (proof in Section \ref{sec:appendix_skyline}).

\begin{lemma}
  Let $S_1,S_2,\dots,S_k$ be skylines each with rounding $c'$ each of size $x$ where $S_i$ is collected after flipping $S_{i-1}$. Additionally, denote the corresponding thresholds $T_1,T_2,\dots,T_k$. 

  Then for each $1 \le i \leq k$, we have that 
    \[T_i \leq \min_{j < i} T_j + c'.\]

    \label{lem:skyline-thresholds-min}
    
\end{lemma}

\subsection{Algorithm}

Like Berglin and Brodal, our algorithm has parameters $\delta,\sigma,$ and $\epsilon$ \cite{BB20}.
We set $\delta, \sigma, \epsilon$ to be values satisfying the following conditions: 
\begin{enumerate*}[label=(\arabic*)]
    \item there exists an offline $\delta$-orientation strategy $\kappa$ of $\mathcal{B}$ making at most $\sigma$ flips per update per batch in the worst case, and
    \item $0 < \epsilon < 1$.
\end{enumerate*}
Additionally, define $\eta = 1 + 1/\epsilon + 2\sigma$. Note that in all cases, we have $1 \le \eta \le O(\log n)$. 
Our algorithm will achieve the following bounds. 

\begin{thm}
    Algorithm \ref{alg:wealg} maintains an $O(\delta + (c' + \delta\epsilon + 1)\log n)$-orientation of $G$ under a sequence of batch updates $\mathcal B$ given by an adaptive adversary, where $\Omega(\frac{\delta}{\log n}+\delta \epsilon+1) \le c' < \delta$ is the rounding of the skylines used (and of the roughly sorted list of vertex out-degrees). For any batch of edge updates $B$, the algorithm makes $O(\eta)$ flips per edge update deterministically worst case and uses:
    \begin{itemize}
        \item $O(b\eta)$ expected work and $O(b \eta \log b)$ deterministic work.
        \item $O(\eta \log^2 n \log b)$ \depth{}.\footnote{Note that the \depth{} upper bound is loose when $|B|$ is small; this can be traced to the proof of Lemma \ref{lem:demands}. The focus of this paper is on work-efficiency rather than on optimizing the logarithmic power in the depth.}
       
    \end{itemize}

    \label{thm:wealg}
\end{thm}

Our worst-case algorithm, Algorithm \ref{alg:wealg}, proceeds as follows. In the edge case that the number of edges in the graph is less than $O(|B|)$, we will apply the batch update then statically orient the entire graph. Otherwise, $B$ is applied (inserted/deleted) into/from the graph (the inserts are statically oriented).\footnote{Note that the static orientation of the inserts is technically necessary because of the vertex roughly sorted list, see Section \ref{sec:appendix_skyline} for details.} Next, we will flip $2\eta$ skylines of size $\frac{|B|}{2}$ (rounding up to the next integer size if $|B| \mod 2 = 1$). If $|B| \ge 4$, we then take an additional skyline $B'$ of size $\frac{|B|}{2}$ and remove it from the graph, and then recurse with $B'$ as the new batch (insert) update.


\begin{algorithm}[ht!]\caption{Worst-Case Algorithm}\label{alg:wealg}

\function{\textsc{reducePotential}$(\overline G, s, t)$} {

\RepTimes{$t$}{
        $S \gets \text{skyline of size } s$  \\
        \If{$S$ has sufficient height}{
            Flip $S$.
        } \Else{
            $3c$-orient $S$ and \Return \textsc{true}     }
    }
    \Return \textsc{false}
}
\function{\textsc{process}$(\overline{G},B,s,t)$} {
   Apply the deletes in $B$ to $\overline{G}$ \\
   $3c$-orient the inserts in $B$ \\
   Apply the inserts in $B$ to $\overline{G}$ \\
   $a \gets \textsc{reducePotential}(\overline{G},s,t)$ \\
   \If {$a$ or $|B| < 4$} {\Return $\overline{G},\emptyset $\\ }
   $S \gets $ skyline of size $s$ \\
   Remove $S$ from $\overline{G}$ \\
   \Return $\overline{G},S$
}

\function{\textsc{Update$(\overline G,B)$}}{

    $\overline{G},B' \gets \textsc{process}(\overline{G},B,\frac{|B|}{2},2\eta)$ \\
    
   \If{$|B'| > 0$} {
   Update$(\overline{G},B')$ \\}

}
\end{algorithm}

\subsection{Work and Depth}

\begin{lemma} \label{lem:wcWorkDepth}
      For any batch of edge updates $B$, Algorithm \ref{alg:wealg} makes $O(\eta)$ flips per edge update deterministically worst case. It uses $O(b\eta)$ expected work and $O(b \eta \log b)$ deterministic work, and $O(\eta \log^2 n \log b)$ \depth{}. 
       \end{lemma}

Statically orienting the insertions in the batch update has $O(b)$ expected work and $O(\log^2 b)$ depth by Lemma \ref{lem:static-orient} \cite{barenboim2008sublogarithmic}. 
By Lemma~\ref{lem:skylineFlipCost}, flipping a skyline of size $s$ has $O(s)$ expected work and $O(\log^2 n)$ depth. In the $i^{th}$ call to \textsc{Update}, we flip $2\eta$ skylines of size $\frac{|B|}{2^i}$, for $O(\frac{|B|}{2^i} \eta)$ expected work. Therefore, the total expected work is $\sum_i O(\frac{|B|}{2^i} \eta) = O(|B| \eta)$. Since the only randomized component is the semisort, and all semisorts occur on arrays of size at most $b$, by Lemma \ref{lem:sort}, we also have $O(|B| \eta \log b )$ deterministic work.  

Now consider the depth. A $\textsc{ReducePotential}(\overline{G},s,t)$ call will have $O(t \log^2 n)$ depth by Lemma~\ref{lem:skylineFlipCost} and Lemma~\ref{lem:skylineFlipStaticCost}. Since static orientation is $O(\log^2 b)$ depth on a set of size $b$ (Lemma \ref{lem:static-orient}), the depth of \textsc{process} is dominated by the depth of \textsc{ReducePotential}. Since the batch size halves each recursive call and $t=2\eta$, there will be at most $O(\log b)$ recursive calls, for $O(\eta \log^2 n \log b)$ depth overall.

%% file: potentialsMain.tex
\section{Analysis of potentials}\label{sec:potential-setup-main}

In this section, we introduce potentials, as constructed by Berglin and Brodal \cite{BB20}. Potentials will be a core tool used in analyzing the maximum out-degree of our orientation algorithms.

\myparagraph{Motivation for Potentials} Consider viewing the orientation problem as a load balancing task where each edge is a task that must be assigned to one of its two endpoints. From this point of view, flipping an edge is reassigning a task to the other suitable vertex, and flipping out-edges away from vertices with high out-degree corresponds to reducing the workload of overwhelmed vertices. Intuitively, flipping out-edges away from high vertices is ``good" for load balancing the graph. But how do we actually measure this progress, especially since our max out-degree may temporarily increase during the process of flipping edges? 


Berglin and Brodal are able to measure progress by giving vertices potential values. Informally, the potential of a vertex is its out-degree, plus or minus some small additive error. The potentials are defined in a careful way so that flipping an edge from a high out-degree vertex reduces slightly the sum of all vertex potentials (signifying progress).

\begin{table}[h]
\begin{center}
\begin{tabular}{c|c}
Symbol & Meaning \\
\hline
$\sigma$ & worst case number of flips in offline strategy chosen \\
$\epsilon$ & penalty for front vs back bag placement, $0 < \epsilon < 1$ \\
$F_v$ & front edges of $v$ \\
$B_v$ & back edges of $v$ \\
$p(v)$ & potential of vertex $v$ \\
\end{tabular} 
\end{center}
\caption{Potential-related notation} 
\end{table}

Formally, for a given sequence of update batches $\mathcal{B} = \{B_i\}$, let $\delta, \sigma, \epsilon$ be values satisfying the following conditions: 
\begin{enumerate*}[label=(\arabic*)]
    \item there exists an offline $\delta$-orientation strategy $\kappa$ of $\mathcal{G}$ making at most $\sigma$ flips per update per batch in the worst case, and
    \item $0 < \epsilon < 1$.
\end{enumerate*}
We consider an edge to be ``good" if it agrees with the offline orientation and ``bad" otherwise. 
Front refers to the front bag in the vertex's pannier (that holds the out-edges), and back refers to the back bag. 
Then the potential value of an edge is given by Table \ref{tab:BBpotentials}, and the potential value $p(v)$ of a vertex $v$ is the sum of the potentials of its out-edges. 
We say an edge is flipped \textit{from} $v$ and \textit{to} $u$ if it was removed from $F_v$ and inserted into $B_u$. 



\myparagraph{Queues vs Panniers} In Berglin and Brodal's work, a vertex's out-edges were stored in a strict FIFO queue, and the front and back queues were analytical tools (not present in the algorithm). 
This does not function well in the parallel batch-dynamic setting for two reasons:

\begin{enumerate}
    \item Using a naive parallel batch-dynamic queue (for example, a balanced binary search tree) would cost $O(\log n)$ per update, compared to constant for the sequential queue. 
    \item  Since statically orienting a skyline will flip some edges and not others, the edges that must be flipped are not a prefix of the queue. A queue structure that only removes from the front and inserts to the back cannot support static orientation. 
\end{enumerate}

Instead, we use a pannier $Q_v$ to store the out-edges of vertex $v$. Let $F_v$ denote the front bag in the pannier for vertex $v$, and $B_v$ the back bag. In our analysis, the front will correspond to the front bag and the back to the back bag. Thus, in our work, front and back are built into the algorithm, and not just analytical concepts. Note that a pannier is weakly ordered (it does not have a strict FIFO order). This makes the pannier not only more efficient to maintain, but allows us to cleanly handle cases like static orientation where only some edges need to be flipped. In particular, note that we can remove an edge from $F_v$ then reinsert the edge to $F_v$ without changing the potential.

\myparagraph{Potential Lemmas} Because the underlying data structure (panniers) and access patterns (skyline) of our algorithms are different, we must reprove various useful properties about potentials. We state them here and provide intuition as to what function they will serve, but defer the proofs to the appendix, since the proofs largely follow from Berglin and Brodal's analysis \cite{BB20}. 

Lemma \ref{lem:potential-degree-relation} formalizes the notion that a vertex's out-degree is close (within additive error) to its potential value. 
Let $\beta = 6\delta\epsilon$ be the \textit{resolution} of the system, which is the range in the potential value: we will use $\beta$ extensively in later sections. 
Lemma \ref{lem:potential-increase-per-update} shows that the amount of potential increase that can come from an edge update is bounded by $1 + \epsilon + 2\sigma \epsilon$.
Lemma \ref{lem:skylineFlipRelease} shows that when we flip a skyline (of sufficient height), we do get at least $\epsilon$ potential decrease per edge in the skyline. We thus will need to do $\eta := \frac{1 + \epsilon + 2\sigma \epsilon}{\epsilon}$ flips to counteract the $1 + \epsilon + 2\sigma \epsilon$ potential increase.
Finally, Lemma \ref{lem:skylineStaticRelease} formalizes that since all edges that enter a skyline static orientation were eligible to be flipped, running static orientation does not worsen our potential total.

\begin{restatable}{lem}{lemBeta}[Lemma 1 in \cite{BB20}]\label{lem:potential-degree-relation}
    For any vertex $v$, $d^+(v) + 5\delta\epsilon \geq p(v) \geq d^+(v) - \delta\epsilon$.
\end{restatable}

\begin{restatable}{lem}{lemPotentialIncrease}[From analysis in \cite{BB20}]\label{lem:potential-increase-per-update} For a single edge update, the sum of potential inserted is at most $1+\epsilon+2\sigma\epsilon$. \label{lem:edge-potential} \end{restatable}

\begin{restatable}{lem}{lemSkylinePotential}[Skyline version of Lemma 4 in \cite{BB20}]
    Let $S$ be a skyline of size $x$, rounding $c'$, and induced threshold $T$ of sufficient height. Then flipping $S$ releases at least $x\epsilon$ potential.  
    \label{lem:skylineFlipRelease}
\end{restatable}

\begin{restatable}{lem}{lemStaticPotential}
    Let $S$ be a skyline of size $x$, rounding $c'$, and induced threshold $T$ of sufficient height. Then statically orienting $S$ does not increase net potential.
        \label{lem:skylineStaticRelease}

\end{restatable}

%% file: appendix_batch_counter.tex
\section{Batch-Counter Game \label{sec:counter-game}}

In this section, we define the batch counter game and show its bounds. The batch counter game is the second core tool we will need for the max out-degree analysis.

Berglin and Brodal use a counter game to bound the max out-degree of their orientation algorithm \cite{BB20}. 
However, their counter game was designed with the sequential dynamic setting in mind. Based on their game, we design a batch-counter game which more naturally applies to the parallel batch-dynamic setting while maintaining the same bound. Note that our counter game could be directly applied to their sequential orientation algorithm to give the same bounds; our counter game dominates their counter game. 

Consider a function $u$ that maps vertices (of which there are $n$) to nonnegative, real counter weights. The weights of these counters will eventually correspond to the potentials of vertices in the graph. This game will consist of a single, adversarial player. On each turn, the player may adjust counter weights, with some restrictions. In this section, we show that in the game, the max counter weight is bounded. In Section \ref{sec:dom}, we will show that the game is powerful enough to simulate the changes in potential in our algorithm.

\subsection{Game}

Let $u : V \rightarrow \mathbb{R}$ give the counter weights (which are nonnegative). 
Fix $H > 0$. 
We define a legal $T$-move between counter weights $u$ and $u'$, denoted $u \rightarrow_T u'$, where $T > 0$, $u$ gives the old counter states, and $u'$ gives the new counter states, if the following holds:

\begin{enumerate}
\item[(R1)] $\forall v \in V, u'(v) - u(v) \ge \textrm{min}(0,T-u(v))$ (If you are below the threshold, you cannot lose weight, and if you are above the threshold, you cannot be reduced below $T$)
\item[(R2)] $\sum_{v \in V} u(v) \ge \sum_{v \in V} u'(v)$ (total counter weight does not increase).
\item[(R3)] $\forall i \in V, u'(i) \le T + H$ (All counters cannot go above threshold + cap).\footnote{Note that a single counter can increase its weight by more than $H$ during a $T$-move, so long as its new value is at most $T+H$.}
\end{enumerate}
Informally, a legal action on $u$ is taking all of the counter weight above $T$ and rearranging (or removing) it however we want, except making sure the new max is no more than $T + H$.

We say that $u_1,u_2,\ldots,u_f$ is a valid game sequence if for all $i \in [f-1]$, there exists $T_i$ such that $u_i \rightarrow_{T_i} u_{i+1}$. When we refer to time $t$, we mean counter state $u_t$, and that threshold $T_t$ is used to transition $u_t$ to $u_{t+1}$. 


\subsection{Limited counter growth}
We claim the following (that the counter game does not build up too much weight on any counter).

\counterLimit*

First, note that choosing a threshold below $Y-H$ would reduce all counters to weight below the starting weight, and so would be a strictly worse position than the starting position. We therefore assume that all thresholds are at least $Y-H$.
Next, we show that we can assume that the chosen thresholds are strictly increasing.

\begin{lemma} Consider a game sequence $(u_i)_{i=1}^f$ with associated thresholds $(T_i)_{i=1}^{f-1}$. There exists a game sequence $(u'(i))_{i=1}^x$ with associated thresholds $(T_i')_{i=1}^{x-1}$ where $T_t' < T_{t+1}'$ for all $t \in [x-1]$ ($u_f$ can be reached by
a sequence of $T_t$'s, where each $T_t$ is strictly increasing). \label{lem:increasingThresholds} \end{lemma}


\begin{proof}
Consider the first time that the threshold dips: $T_1, T_2, T_3$, where $T_1 < T_2$ but $T_2 \ge T_3$.\footnote{Note that by $T_1$ we mean $T_{s+1}$ for some time $s$. We shorten the notation for clarity.}
Associated are counter states $u_1,u_2,u_3,u_4$. We will claim that $u_2 \rightarrow_{T_3} u_4$.

Formally: we have 3 requirements to show.
\begin{enumerate}
\item[(R1)] We must show that for all $i \in V$, that $u_4(i) - u_2(i) \ge \textrm{min}(0,T_3-u_2(i))$.
\begin{enumerate}[leftmargin=*]
\item[Case 1:] $u_2(i) \le T_3$ (under the bar). We must show that $u_4(i) \ge u_2(i)$. Observe that
$u_2(i) \le T_3 \le T_2$, so $u_3(i) \ge u_2(i)$ by R1 on $u_2 \rightarrow_{T_2} u_3$.
\begin{enumerate}[leftmargin=*]
\item[Case 1a:] $u_3(i) \le T_3$. Then $u_3(i) \le u_4(i)$, so $u_2(i) \le u_4(i)$.
\item[Case 1b:] $u_3(i) > T_3$. Then $u_4(i) \ge T_3 \ge u_2(i)$.
\end{enumerate}
\item[Case 2:] $u_2(i) > T_3$ (over the bar). We must show that $u_4(i) \ge T_3$.
\begin{enumerate}[leftmargin=*]
\item[Case 2a:] $u_2(i) \le T_2$. Then $u_3(i) \ge u_2(i)$. Since $u_3(i) \ge u_2(i) > T_3$, we have
that $u_4(i) \ge T_3$ by R1 on $u_3 \rightarrow_{T_3} u_4$.
\item[Case 2b:] $u_2(i) > T_2$. Then $u_3(i) \ge T_2 \ge T_3$. If $u_3(i)=T_3$, then $u_4(i) \ge u_3(i) = T_3$. 
Otherwise, $u_3(i) > T_3$, so we have that $u_4(i) \ge T_3$.
\end{enumerate}
\end{enumerate}
\item[(R2)] Since $\sum u_2(i) \ge \sum u_3(i)$ and $\sum u_3(i) \ge \sum u_4(i)$, we have
$\sum u_2(i) \ge \sum u_4(i)$, as desired.
Therefore $u_2 \rightarrow_{T_3} u_4$.
\item[(R3)] Unaffected by skipping an intermediate state (only involves $T_3$ and $u_4$), so is fine.
\end{enumerate}
\end{proof}

\begin{definition} Let $M_t := \max(u_t)$ be the max counter value at time $t$.\footnote{Before skyline flip i with threshold $T_i$.} \end{definition}

\begin{lemma} Suppose that $(T_t)$'s are (strictly) increasing. Then $M_i - T_i \le H$.\footnote{Remark: it is possible that $M_{i+1} < M_i$, this is fine, though usually we do expect the max to
go up over time.}$^,$\footnote{Remark: From R3, we have that $M_{i+1} \le T_i + H$. This lemma is saying something different,
that $M_i \le T_i + H$ (that the threshold you pick must be rather close to the current max). This is
forced by the thresholds being strictly increasing.} \label{lem:maxThresholdCloseness}
\end{lemma}
\begin{proof}
Since all thresholds are at least $Y-H$, for time $i=1$, $M_1 = Y$, and $T_1 \ge Y-H$, so $M_1 - T_1 \le Y - (Y-H) = H$.
Consider time $i > 1$. By R3, we have $M_{i}-T_{i-1} \le H$. Since $T_{i-1} < T_{i}$ by assumption,
this means that $M_{i} - T_{i} < M_{i} - T_{i-1} \le H$. \end{proof}

\begin{definition} Let $w(t,l)$ denote the weight above $l$ at time $t$, and let $w_t$ denote the weight above
$T_t$ at time t. 
That is, $w(t,l) = \sum_{i \in [n]} max(0,u_t(i)-l)$. Abusing notation, also write $w(u,l)$
for the weight above $l$ in the counter state $u$. \end{definition}
The weight is the counter mass that we have to play with. 
We will note that the weight goes down over time, and so the game has increasingly limited ability to do anything bad. 

Note that our Lemmas \ref{lem:weightsDown} and \ref{weightsDownLong} loosely correspond to Berglin and Brodal's Lemma 11, and our Lemma \ref{lem:weightsHalve} to their Lemma 12 \cite{BB20}. However, whereas Berglin and Brodal have fixed levels to measure the accumulation of weight, we must use the thresholds (given by the adversary) as milestones. Because the thresholds cannot be too far apart (by Lemma \ref{lem:closeThresholds}), using the thresholds in this way is not harmful to our bound. 

\begin{lemma} Let $u \rightarrow_T u'$. Let $L \le T$. Observe that $w(u,L) \ge w(u',L)$. \label{lem:weightsDown}\end{lemma}
\begin{proof}
Let $I_{h,h}$ be the counters who start high, stay high ($u \ge L$, $u' \ge L$). $I_{l,l}$ be counters who
start low, stay low ($u < L$, $u' < L$). $I_{l,h}$ be counters who start low, end high ($u < L, u' \ge L$).
Because $L \le T$, by R1, there are no counters who start high and end low.
Observe by R2 that $0 \ge \sum_i u'(i) - \sum_i u(i)$, which is equal to
$\sum_{i \in I_{h,h}} u'(i) - u(i) + \sum_{i \in I_{l,l}} u'(i) - u(i) + \sum_{i \in I_{l,h}} u'(i) - u(i)$. Observe
that $ w(u',L)-w(u,L) = \sum_{i \in I_{h,h}} u'(i) - u(i) + \sum_{i \in I_{l,h}} u'(i) - L$. Since $u < L$ for all
low $u$, we have that this is no more than $\sum_{i \in I_{h,h}} u'(i) - u(i) + \sum_{i \in I_{l,h}} u'(i) -
u(i) \le -\sum_{i \in I_{l,l}} (u'(i) - u(i)) \le 0$. Therefore $w(u,L) \ge w(u',L)$. \end{proof}

\begin{lemma} Suppose $(T_t)$'s increasing. Suppose $t \le t'$. Then $w(u_t,T_t) \ge w(u_{t'},T_t)$. (weights go
down over time) \label{weightsDownLong} \end{lemma}
\begin{proof} This is because skylines can't bring up weight. All of the skylines past $T_t$ are either
rearranging the weight above $T_t$, or throwing some down, never bringing up.\footnote{Remark: Note that it is not true that $w(t,l) \ge w(t+1,l)$ for every $l$.} Thus, repeatedly applying Lemma \ref{lem:weightsDown} (which we can since the thresholds are
increasing) yields the result. \end{proof}

\begin{lemma} Suppose that the thresholds are strictly increasing over time, $j > i$, and $T_j - T_i \ge 2H$. Then $w_j \le \frac{1}{2} w_i$.\footnote{Where $w_j = w(u_j,T_j)$.} (weight goes down by half with jump of $\Theta(H)$ in threshold value) \label{lem:weightsHalve} \end{lemma} 

\begin{proof} 
Suppose there are $v$ counters with height at least $T_j$ at time $j$. Since $M_j \le T_j + H$ at time $j$ (by
Lemma \ref{lem:maxThresholdCloseness}), we have that $w_j \le Hv$.
At time $j$, observe that $w(u_j,T_i) \ge (T_j-T_i) v \ge 2Hv$. By Lemma \ref{weightsDownLong}, we have $w(u_j,T_i) \le w(u_i,T_i)$. Putting it together, we have $w_j \le Hv \le \frac{1}{2} 2Hv \le \frac{1}{2} w(u_j,T_i) \le \frac{1}{2} w(u_i,T_i)=\frac{1}{2} w_i$.\end{proof}

\begin{lemma} Let $(u_i,T_i)_{i=1}^f$ be a valid game sequence. If there exists $x \in [f-2]$ such that $T_{x+1}-T_x
> H$, then $u_{x+1}=u_f$. \label{lem:closeThresholds} \end{lemma}
\begin{proof}
Suppose that $T_{x+1}-T_x > H$. 
Then $T_{x+1} > T_x + H \ge M_{x+1}$ (by R3).  Consider the move $u_{x+1} \rightarrow_{T_{x+1}} u_{x+2}$. By R1, no counter below the threshold loses weight. Since all counters are below the threshold, no counter loses weight. Therefore, no counter can gain
weight, so the counter state is unchanged ($u_{x+2}=u_{x+1}$). Further thresholds are above
$T_{x+1}$, and so also cannot change the counter states.  \end{proof}

\begin{proof}[Proof of Theorem \ref{thm:counterLimit}] Consider the game $(u_i,T_i)_{i=1}^f$. We can assume that the $(T_i)$'s are strictly
increasing and that $T_{i+1} - T_i \le H$ by Lemmas \ref{lem:increasingThresholds} and \ref{lem:closeThresholds}. Let $y_1 = 1$, and inductively define $y_j$ to be the smallest value of $z$ for which $T_z-T_{y_{j-1}} \ge 2H$. Because each $T_i$ differs by no more than $H$, note that $T_{y_j}-T_{y_{j-1}} \le 3H$. We are taking a subsequence of threshold choices with
sufficiently different height. Applying Lemma \ref{lem:weightsHalve} to each jump, we have that the weight above
the threshold goes down by half each time, so there are at most $O(\log n)$ jumps of thresholds.
Each threshold jump gives an $O(H)$ increase, so the max counter has value no more than $O(Y+H \log n)$. \end{proof}

%% file: actualAnalysis.tex
\section{Domination Framework}\label{sec:dom}

The goal of this section is to establish a general framework for bounding the maximum out-degree of our dynamic orientation algorithms. We proceed as follows. First, in Section~\ref{subsec:analysis-prelim}, we formally define the notion of an $(H,T)$-bounded call. Then in Section~\ref{subsec:domination-step}, we demonstrate that any $(H,T)$-bounded call corresponds to a legal move in the counter game while preserving a domination invariant. Building on this correspondence, Section~\ref{subsec:global-outdegree} proves our main structural result: for any update procedure whose calls are either $(H,T)$-bounded or trivially produce orientations of maximum out-degree $O(\delta)$, the global out-degree is bounded by
\begin{align*}
\max_{x} d^+_{\overline G}(x)
=
O\left(\delta + \bigl(H+\delta\epsilon+1\bigr)\log n\right). \tag{Theorem \ref{thm:generic-outdegree}}
\end{align*}
Finally, we apply this theorem to our worst-case algorithm.

Throughout this section, we fix an online update algorithm $A$ and a sequence of update batches $\{B_i\}_{i\ge 1}$. Let $\overline G_0$ be the empty orientation on $n$ vertices and define $\overline G_i := A(\overline G_{i-1}, B_i)$.  Let $p_i$ denote the vertex-potential function associated with $\overline G_i$. Recall $\beta := 6\delta\epsilon$ and
$\eta := 1 + 1/\epsilon + 2\sigma$.

\subsection{Preliminaries}\label{subsec:analysis-prelim}
We associate one counter with each vertex. For an orientation $\overline G$ with potential function $p$, we say a counter state $u$ \emph{dominates} $p$ if $u(v)\ge p(v)$ for every vertex $v$. For a counter state $u$ and a potential function $p$, define the vertexwise gaps
\[
g^+(u,p)(v) := \max\{p(v)-u(v),0\},
\qquad
g^-(u,p)(v) := \max\{u(v)-p(v),0\},
\]
and their totals
\[
\operatorname{Gap}^+(u,p) := \sum_v g^+(u,p)(v),
\qquad
\operatorname{Gap}^-(u,p) := \sum_v g^-(u,p)(v).
\]
For two potential functions $p,q$, define the vertexwise change and increase
\[
\delta(p,q)(v) := q(v)-p(v),
\qquad
\delta^+(p,q)(v) := \max\{\delta(p,q)(v),0\},
\]
and their totals
\[
\Delta(p,q) := \sum_v \delta(p,q)(v),
\qquad
\Delta^+(p,q) := \sum_v \delta^+(p,q)(v).
\]

\begin{lemma}\label{lem:gap-after-applying-batch}
Let $u$ dominate $p$. Then for any potential function $r$,
\[
\operatorname{Gap}^+(u,r)\le \Delta^+(p,r).
\]
\end{lemma}
\begin{proof}
Let $P:=\{v \in V:\ r(v)>u(v)\}$. Then
\[
\operatorname{Gap}^+(u,r)=\sum_{v\in P}(r(v)-u(v))
=\sum_{v\in P}\bigl(p(v)-u(v)\bigr)+\sum_{v\in P}\bigl(r(v)-p(v)\bigr).
\]
Since $u$ dominates $p$, the first sum is nonpositive. Therefore,
\[
\operatorname{Gap}^+(u,r)\le \sum_{v\in P}\delta^+(p,r)(v)\le \Delta^+(p,r).
\]
\end{proof}

Consider an invocation of $A$ on an input orientation $\overline G$ and batch $B$. 
If $A$ is recursive, then during this invocation it may generate a sequence of
recursive calls on progressively smaller batches. We index these calls as follows.
Let $\overline G^{(0)} :=\overline G$ and $B^{(0)}:=B$. For $t=0,1,\dots,m-1$,
the $t$-th call runs $A(\overline G^{(t)}, B^{(t)})$, produces an output
orientation $\overline G^{(t+1)}$. The invocation may delete a batch of edges $B^{(t+1)}$ (which may be empty). The invocation may then terminate, or recurse on $B^{(t+1)}$ as an insertion batch. Let $\mathcal{G}$ be the update stream including these recursive calls.

For each $t\in\{0,\dots,m-1\}$, write:
\begin{itemize}
    \item $p^{(t)}$ for the potential function of $\overline G^{(t)}$;
    \item $r^{(t)}$ for the potential function after applying $B^{(t)}$ to
    $\overline G^{(t)}$. If $t=0$, we also let $\kappa$ perform flips.
    \item $q^{(t)}$ for the potential function of the output orientation
    $\overline G^{(t+1)}$.
\end{itemize}

By Lemma~\ref{lem:potential-increase-per-update}, applying $B^{(t)}$ increases total
potential by at most $|B^{(t)}|\eta\epsilon$, i.e.,
\begin{equation}\label{eqn:inj-bound}
\Delta^+(p^{(t)}, r^{(t)}) \le |B^{(t)}|\eta\epsilon.
\end{equation}

We classify calls into two types.
\begin{definition} \label{def:trivial}
A call $A(\overline G^{(t)}, B^{(t)})$ is \emph{trivial} if 
\[ \max_v d^+_{\overline G^{(t+1)}}(v) \le 8 \delta. \]
\end{definition}

\begin{definition}\label{def:HT-bounded}
Fix parameters $H>0$ and $T\ge 0$.
A call $A(\overline G^{(t)}, B^{(t)})$ is \emph{$(H,T)$-bounded} if:
\begin{enumerate}
    \item for every vertex $v$,
    \[
    q^{(t)}(v)\ \ge\ \min\{r^{(t)}(v), T-\delta\epsilon\};
    \]
    \item $\Delta(r^{(t)}, q^{(t)}) \le -|B^{(t)}|\eta\epsilon$;

    \item the output orientation satisfies
    \[
    T \ \le\ \max_v d^+_{\overline G^{(t+1)}}(v) \ \le\ T+H.
    \]
\end{enumerate}
\end{definition}

\subsection{Connecting bounded move with counter game}\label{subsec:domination-step}

The next lemma shows that each $(H,T)$-bounded call corresponds to a legal move in
the counter game (with a slightly larger cap), while maintaining domination.

\begin{lemma}\label{lem:domination-transfer}
Suppose $u$ dominates $p^{(t)}$ and $A(\overline G^{(t)},B^{(t)})$
is $(H,T)$-bounded. Define
\[
T^* := T-\delta\epsilon
\qquad\text{and}\qquad
H^* := H+\beta.
\]
Then there exists a counter state $w$ such that:
\begin{enumerate}
    \item $w$ dominates $q^{(t)}$; and
    \item $u \to_{T^*} w$ is a legal $T^*$-move in a counter game with cap $H^*$.
\end{enumerate}
\end{lemma}
\begin{proof} 

Write $p:=p^{(t)}$, $r:=r^{(t)}$, $q:=q^{(t)}$, and $B:=B^{(t)}$ for brevity. By Lemma~\ref{lem:gap-after-applying-batch} and \eqref{eqn:inj-bound},
\begin{equation}\label{eqn:gap-u-r}
\operatorname{Gap}^+(u,r)\le \Delta^+(p,r)\le |B|\eta\epsilon.
\end{equation}

Define $v$ by
\[
v(x) := u(x) + \delta(r,q)(x) = u(x) + (q(x)-r(x)).
\]
Then for every vertex $x$,
\[
q(x)-v(x)=r(x)-u(x),
\]
and therefore
\begin{equation}\label{eqn:gap-v-q}
\operatorname{Gap}^+(v,q)=\operatorname{Gap}^+(u,r).
\end{equation}
Moreover,
\[
\sum_x v(x)
=
\sum_x u(x) + \Delta(r,q)
\le
\sum_x u(x) - |B|\eta\epsilon,
\]
where the inequality uses item~(2) of Definition~\ref{def:HT-bounded}.

Let $\hat q:=\max_x q(x)$ and define $w$ by
\[
w(x) :=
\begin{cases}
q(x) & \text{if } q(x)>v(x),\\
\min\{v(x),\hat q\} & \text{if } q(x)\le v(x).
\end{cases}
\]
Then $w(x)\ge q(x)$ for all $x$, hence $w$ dominates $q$, and $\max_x w(x)=\hat q$.

It remains to verify that $(u,w)$ is a legal $T^*$-move with cap $H^*$.

\paragraph{(R1).}
We show $w(x)\ge \min\{u(x),T^*\}$ for every $x$.
Since the subcall is $(H,T)$-bounded, item~(1) gives
\begin{equation}\label{eqn:q-floor}
q(x)\ge \min\{r(x),T^*\}.
\end{equation} 
Let $\hat x$ be a maximum out-degree vertex in the output $\overline G^{(t+1)}$. By item~(3), $d^+(\hat x)\ge T$, and by Lemma~\ref{lem:potential-degree-relation},
\begin{equation}
\hat q \ge q(\hat x) \ge d^+(\hat x)-\delta\epsilon \ge T-\delta\epsilon = T^*.
\label{eq:qhatGreaterT}
\end{equation}
Now fix $x$. We consider two cases.

\begin{enumerate}
    \item $u(x)\le T^*$. We show $w(x)\ge u(x)$.
        \begin{enumerate}
        \item If $r(x)\ge u(x)$, then by \eqref{eqn:q-floor},
        $w(x)\ge q(x)\ge \min\{r(x),T^*\}\ge u(x)$.
        \item If $r(x)<u(x)\le T^*$, then $r(x)<T^*$ and \eqref{eqn:q-floor} implies
        $q(x)\ge r(x)$, hence $v(x)=u(x)+q(x)-r(x)\ge u(x)$. We case on comparing $q(x)$ to $v(x)$. 
        \begin{enumerate}
            \item If $q(x) > v(x)$, then $w(x) = q(x) \ge v(x) \ge u(x)$.
            \item If $q(x) \le v(x)$, then $w(x) = \min\{v(x),\hat q\}$. Since $\hat q\ge T^*\ge u(x)$,
        we have $w(x)=\min\{v(x),\hat q\}\ge u(x)$. 
        \end{enumerate}

        \end{enumerate}
    \item $u(x) > T^*$. We show $w(x)\ge T^*$.
        \begin{enumerate}
        \item If $r(x)\ge T^*$, then $w(x)\ge q(x)\ge \min\{r(x),T^*\}=T^*$ by
        \eqref{eqn:q-floor}.
        \item If $r(x)<T^*$, then $q(x)\ge r(x)$ by \eqref{eqn:q-floor}, so
        $v(x)=u(x)+q(x)-r(x)\ge u(x)>T^*$, and $\hat q\ge T^*$ by (\ref{eq:qhatGreaterT}). Thus
        $w(x)=\min\{v(x),\hat q\}\ge T^*$.
        \end{enumerate}
\end{enumerate}

This proves (R1).

\paragraph{(R2).} First note that
\[
\sum_x q(x) - \sum_x v(x)
    = \sum_{x : q(x) > v(x)} (q(x) - v(x))
       -
       \sum_{x : q(x) \le v(x)} (v(x) - q(x)) \\
    = \operatorname{Gap}^{+}(v, q) - \operatorname{Gap}^{-}(v, q).
\]
Hence,
\begin{equation}\label{eqn:gap-sum}
\sum_x q(x)
    = \sum_x v(x)
      + \mathrm{Gap}^{+}(v, q)
      - \mathrm{Gap}^{-}(v, q)
\end{equation}
Note that we can equivalently express $w(x)$ as $ \min(q(x) + g^-(v,q)(x),\hat{q}).$ Then, by (\ref{eqn:gap-sum}) 
\begin{equation*}\sum_x w(x) = \sum_x \min(q(x) + g^-(v,q)(x), \hat{q}) \leq \operatorname{Gap}^-(v,q) +  \sum_x q(x) = \sum_x v(x) + \operatorname{Gap}^+(v,q).\end{equation*}
Finally, using \eqref{eqn:gap-v-q}, \eqref{eqn:gap-u-r}, and the definition of $w$, we have
\[\sum_x w(x) \ \le\  \sum_x v(x) + \operatorname{Gap}^+(v,q)
\ =\ \sum_x v(x) + \operatorname{Gap}^+(u,r)
\ \le\ \sum_x u(x)-|B|\eta\epsilon + |B|\eta\epsilon
\ =\ \sum_x u(x).
\]
Thus (R2) holds.

\paragraph{(R3).}
Since $\max_x w(x)=\hat q$, it suffices to upper bound $\hat q$.
By item~(3) of Definition~\ref{def:HT-bounded},
$\max_x d^+_{\overline G^{(t+1)}}(x)\le T+H$. Applying
Lemma~\ref{lem:potential-degree-relation},
\[
\hat q = \max_x q(x)
\le
\max_x d^+_{\overline G^{(t+1)}}(x) + 5\delta\epsilon
\le
(T+H)+5\delta\epsilon.
\]
Using $T=T^*+\delta\epsilon$ and $\beta=6\delta\epsilon$,
\[
(T+H)+5\delta\epsilon
=
(T^*+\delta\epsilon+H)+5\delta\epsilon
=
T^* + (H+6\delta\epsilon)
=
T^* + H^*,
\]
which is (R3).

Therefore $u\to_{T^*} w$ is a legal move with cap $H^*$, and $w$ dominates $q$.

\end{proof}

\subsection{Global out-degree bound}\label{subsec:global-outdegree}

We now bound the maximum out-degree over the entire execution of $A$.

\begin{theorem}\label{thm:generic-outdegree}
Assume there exists a parameter $H>0$ such that every call to $A$ that occurs while processing the stream $\mathcal G$ satisfies one of the following:
\begin{enumerate}
    \item the call is trivial;
    \item the call is $(H,T)$-bounded for some threshold $T\ge 0$.
\end{enumerate}

Then for every update index $i$,
\[
\max_{v} d^+_{\overline G_i}(v)
=
O\left(\delta + \bigl(H+\delta\epsilon+1\bigr)\log n\right).
\]
\end{theorem}

\begin{proof}
Index all calls to $A$ (top-level and recursive) as
$C_1,C_2,\dots,C_M$. Let $\overline H_0:=\overline G_0$, and for each $j\in[M]$
let $\overline H_j$ denote the orientation output by call $C_j$. Let $p_j$ be the potential function of $\overline H_j$. Note that each update output $\overline G_i$ is the output of some call in the list, so it suffices to bound $\max_v d^+_{\overline H_j}(v)$ for all $j$.

Partition $\{1,2,\dots,M\}$ into \emph{maximal epochs}
$\mathcal E=\{[a_\ell,b_\ell]\}_\ell$ of two types:
\begin{itemize}
  \item \textbf{Type I (trivial epoch):} every call $C_j$ with $j\in[a_\ell,b_\ell]$
  is trivial;
  \item \textbf{Type II (bounded epoch):} $C_{a_\ell-1}$ is trivial (or $a_\ell=1$),
  and every call $C_j$ with $j\in[a_\ell,b_\ell]$ is $(H,T_j)$-bounded for some
  threshold $T_j$.
\end{itemize}

\paragraph{Type I epochs.}
If $C_j$ is trivial, then by definition $\max_v d^+_{\overline H_j}(v)=O(\delta)$.

\paragraph{Type II epochs.}
Fix a Type II epoch $[a,b]$. Let $H^*:=H+\beta$.
Since $C_{a-1}$ is trivial (or $a=1$ and $\overline H_0$ is empty), we have
\[
\max_v d^+_{\overline H_{a-1}}(v) = O(\delta).
\]
By Lemma~\ref{lem:potential-degree-relation},
\begin{equation}\label{eqn:epoch-start-potential-new}
\max_v p_{a-1}(v) = O(\delta) + 5\delta\epsilon.
\end{equation}
Define
\[
Y := \max\{O(\delta) + 5\delta\epsilon,\ H^*+1\},
\]
and let $u_{a-1}$ be the uniform counter state $u_{a-1}(v)=Y$ for all vertices $v$.
Then $u_{a-1}$ dominates $p_{a-1}$ by \eqref{eqn:epoch-start-potential-new}.

For each $j\in\{a,a+1,\dots,b\}$, call $C_j$ is $(H,T_j)$-bounded for some $T_j$.
Applying Lemma~\ref{lem:domination-transfer} inductively yields counter states
$u_a,u_{a+1},\dots,u_b$ such that for every $j\in[a,b]$:
\begin{enumerate}
    \item $u_j$ dominates $p_j$, and
    \item $u_{j-1}\to_{T_j-\delta\epsilon} u_j$ is a legal move in a counter game
    with cap $H^*$.
\end{enumerate}
Thus $(u_{a-1},u_a,\dots,u_b)$ is a valid counter game with cap $H^*$ and starting
weight $Y$. By Theorem~\ref{thm:counterLimit},
\[
\max_{a-1\le j\le b}\ \max_v u_j(v)
=
O\bigl(Y + H^*\log n\bigr).
\]
Since $u_j$ dominates $p_j$ for all $j\in[a,b]$, the same bound holds for the
maximum potential:
\[
\max_{a\le j\le b}\ \max_v p_j(v)
=
O\bigl(Y + H^*\log n\bigr).
\]
Finally, Lemma~\ref{lem:potential-degree-relation} implies
$d^+(v)\le p(v)+\delta\epsilon$, hence throughout the epoch,
\[
\max_{a\le j\le b}\ \max_v d^+_{\overline H_j}(v)
=
O\bigl(Y + H^*\log n + \delta\epsilon\bigr).
\]
Using $\beta=6\delta\epsilon$, $H^*=H+\beta$, and $Y=O(\delta + \delta\epsilon + H)$, this simplifies to
\[
\max_{a\le j\le b}\ \max_v d^+_{\overline H_j}(v)
=
O\left(\delta + \bigl(H+\delta\epsilon+1\bigr)\log n\right)
\]
as desired.
\end{proof}

%% file: we_worstcase.tex
\subsection{Applying the Domination Framework}

To prove the max out-degree bound, we must connect this algorithm with the $(H,T)$ bounded framework. A single counter game move will correspond to a single call to process (if the skylines therein all had sufficient height).

\begin{lemma}\label{lem:update-call-bounded-reinsertion}
Consider a call $\textsc{process}(\overline{G},B,s,t)$ that returns $(\overline{G'},B')$. Let:
\begin{itemize}
    \item $b:=|B| > 0$.
    \item $r$ denote the potential after applying $B$ (and allowing $\kappa$'s flips if this is the first call)
    \item  $q$ denote the potential of $\overline G'$
    \item $1 \le c' < \delta$ be the rounding of the RSL (and the skylines) 
    \item $T := \min(\{T(S_j) : S_j \text{ is any skyline collected in this invocation}\})$ 
    \item  $H := 2c' + 2$.
\end{itemize}  

Suppose that
\begin{enumerate}
    \item The skyline data structure correctly returns skylines.
    \item $|B| \eta \le st$.
\end{enumerate}

Then the call is either trivial or $(H,T)$ bounded.\label{lem:wcHT}
\end{lemma}

\begin{proof}

We case on the return value of \textsc{reducePotential}. 

\myparagraph{Case 1} Suppose that \textsc{reducePotential} returned true. This means that a skyline $S$ was of insufficient height $(T(S) < 4\delta)$. Before returning note that \textsc{reducePotential} $3c$-orients $S$. 

For any vertex $v$, at
most $T(S)+c' < 4\delta+c'$ out-edges of $v$ remain outside $S$, while the $3c$-orientation
ensures that $v$ has at most $3c$ out-edges inside $S$. Hence
$d^+(v)\le (4\delta+c')+3c \le 8\delta$ (using $\delta\ge c$), and the call is trivial.

\myparagraph{Case 2} Now suppose that \textsc{ReducePotential} returned false. Let $S_1,\dots,S_t$ be
the skylines flipped in the loop, and let $S_{t+1}$ be the final skyline that is removed (if a skyline is removed). We will verify the three conditions in Definition~\ref{def:HT-bounded}.

\noindent\textbf{(1)}  Note that only out-edges in skylines are removed from a vertex. 

If $v$ never has an edge enter a skyline, then $q(v)\ge r(v)$ and therefore $q(v)\ge \min\{r(v),T-\delta\epsilon\}$.

If $v$ does have an edge enter a skyline $S_j$, by the definition of a skyline of threshold $T(S_j)$, $v$ has at least $T(S_j)$ out-edges not in the skyline. Therefore, after flipping (or removing) $S_j$, $v$ has at least $T(S_j)$ out-edges. Since $T\le T(S_j)$ for all $j$, $v$ remains with at least $T$ out-edges.

\noindent\textbf{(2)}
Each flipped skyline has sufficient height, so by Lemma~\ref{lem:skylineFlipRelease}
flipping $S_j$ releases at least $s\epsilon$ potential. Thus the $t$ flips release
at least $st\epsilon$. Removing the final skyline can only decrease potential. By assumption 2 (in the lemma statement) we have $st \epsilon \ge b\eta \epsilon$, so $\Delta(r,q) \le -b\eta\epsilon$.

\noindent\textbf{(3)}
For the lower bound, since skylines do not remove edges below the threshold, and all thresholds used have positive out-degree above the threshold (before the skyline flip), we have that $\max_v d^+_{\overline G'}(v) \ge T$. 

If the final skyline is not removed, then $|B| < 4$, so $v$ has at most $2$ edges in (what would have been) $S_{t+1}$. If the final skyline is removed, then $v$ has no (remaining) edges in this skyline. Therefore $v$ has at most $T(S_{t+1})+c'+2$ out-edges.
By Lemma~\ref{lem:skyline-thresholds-ij} we have that 
$T(S_{t+1})\le \min_j T_j +c' = T + c'$, so $d^+(v) \le T(S_{t+1})+c'+2 \le T + 2c' + 2$.

All three items of Definition~\ref{def:HT-bounded} hold, so the invocation is $(H,T)$-bounded.
\end{proof}

\begin{theorem}[Orientation quality]\label{thm:wc-maxout}
Algorithm~\ref{alg:wealg} maintains an
\[
O\left(\delta + \left(c' + \delta\epsilon+1\right)\log n\right)
\text{-orientation}
\]
throughout the edge updates.
\end{theorem}

\begin{proof}

Each call to \textsc{process} has $s=\frac{|B|}{2}$ and $t=2\eta$ (for some initial or recursive batch $B$), so $st = |B|\eta$. Therefore, assuming that skylines are correct, by Lemma \ref{lem:wcHT}, each call to process is trivial or $(H,T)$ bounded for $H=2c'+2$. 

Then applying Theorem~\ref{thm:generic-outdegree} gives that $\max_x d^+(x) \le O(\delta + (2c' + 2 + \delta \epsilon + 1) \log n)$. 
Thus by Lemma \ref{lem:skylinesCorrectImpliesCorrect} we conclude that skylines are indeed correct and that this is the max out-degree bound.\end{proof}

\begin{proof}[Proof of Theorem \ref{thm:wealg}]
Using Lemma \ref{lem:wcWorkDepth} and Theorem \ref{thm:wc-maxout}, we immediately conclude Theorem \ref{thm:wealg}.
\end{proof}

\subsection{Tradeoffs}

We plug in various tradeoffs to Algorithm \ref{alg:wealg}. By choosing the offline strategy of Brodal and Fagerberg (Lemma \ref{lem:offline}) \cite{BF99}, we can set $\delta=O(c)$, $\sigma=O(\log n)$, and $\epsilon=\frac{1}{\log n}$, yielding $\eta=O(\log n)$. For the RSL, set $c' = \cp$ and $M = O(\log n)$. Then by Theorem \ref{thm:wealg} we have the following bounds.

\mainCpluslogn*

This can be boosted to an $O(c)$-orientation using a technique of Chekuri et al. (given in Lemma \ref{lem:boostOrientation}) \cite{chekuri2024adaptive}.

\begin{restatable}{theorem}{mainthm}\label{thm:boundsOc}
      There exists a parallel batch-dynamic algorithm that maintains an $O(c)$-orientation. For a batch update of size $b$, the algorithm does $O(b\log^2 n)$ reorientations. The work is $O(b \log^2 n)$ in expectation and $O(b \log^3 n)$ deterministically. The depth is $O(\log^4 n)$. Here, $c$ is a fixed upper bound on the arboricity of the graph across the update sequence, and the edge updates are given by an adaptive adversary.
\end{restatable}

By choosing the offline strategy of He et al. \cite{he2014orienting}, deamortized with Theorem 2 of Berglin and Brodal \cite{BB20}, we can set $\delta = O(c \sqrt{\log n})$ and $\sigma=\sqrt{\log n}$. With $\epsilon=\frac{1}{\sqrt{\log n}}$, we get $\eta=O(\sqrt{\log n})$. We set the rounding of the skylines to be $c'=c$ and note that $M=O(\log n)$ suffices as a budget for the RSL. Thus by Theorem \ref{thm:wealg} we have $O(\delta + (c' + \delta \epsilon + 1) \log n) = O(c \sqrt{\log n} + (c+c \sqrt{\log n} \frac{1}{\sqrt{\log n}} + 1) \log n ) = O(c \log n)$ max out-degree, $\eta = 1 + \frac{1}{\epsilon} + 2 \sigma = O(\sqrt{\log n})$, and so the work is $O(b \sqrt{\log n})$ in expectation. 

\begin{theorem}
    There exists a parallel batch-dynamic algorithm that for a batch update of size $b$, does $O(b \sqrt{\log n})$ reorientations, $O(b \sqrt{\log n})$ expected work, $O(b \log^{1.5} n)$ deterministic work, $O(\log^{2.5} n \log b)$ depth, and maintains a $O(c \log n)$-orientation, where $c$ is a fixed upper bound on the arboricity of the graph across the update sequence, against an adaptive adversary.  \label{thm:boundsOcsqrtlogn}
\end{theorem}

By choosing the offline strategy of Kowalik \cite{kowalik2007adjacency} deamortized with Theorem 2 of Berglin and Brodal \cite{BB20}, we can set $\delta=O(c \log n)$ and $\sigma=O(1)$. Then setting $\epsilon=1$ yields $\eta=O(1)$. Set $c'=c \log n$ to be the rounding of the RSL, and note that $M=O(\log n)$ is large enough to hold all vertices in non-high bags (except in intermediate states). Plugging into Theorem \ref{thm:wealg} yields work $O(b)$ expected, $O(b \log n)$ deterministic work, $O(\log^2 n \log b)$ depth, and a max out-degree bound of $O(c \log n + (c \log n + c \log n (1) + 1)\log n) = O(c  \log^2 n)$.

\begin{theorem}
    There exists a parallel batch-dynamic algorithm that for a batch update of size $b$, does $O(b)$ reorientations,  $O(b)$ expected work, $O(b \log n)$ deterministic work, $O(\log^2 n \log b)$ depth, and maintains a $O(c \log^2 n)$-orientation, against edge updates given by an adaptive adversary, where $c$ is an upper bound on the arboricity of the graph over the update sequence. \label{thm:boundsclognsquared}
\end{theorem}

%% file: appendix_deterministicBag.tex
\section{Deterministic Bag \label{sec:bag}}

With the main result proven, we now will prove that our data structures indeed are efficient. This begins with showing that a deterministic bag can be efficiently maintained under batch updates.

\bag*

\subsection{Skew numbers}

Consider skew binary numbers \cite{MYERS83, okasaki1999purely}. A number is represented as a linked list of nodes with weights. Each weight must be equal to $2^{i+1}-1$ for some $i$. The first two weights may be equal, but otherwise, all weights must be strictly increasing. All natural numbers have a unique skew representation \cite{MYERS83}. Okasaki describes how to increment and decrement skew binary numbers by one \cite{okasaki1999purely}. We extend this to support arbitrary increases and decreases. We give a function for initializing a skew number from a binary number in Algorithm \ref{alg:skew-init}, adding a binary number to a skew number in Algorithm \ref{alg:skew-add}, and subtracting in Algorithm \ref{alg:skew-subtract}. 

\begin{lemma} The function \textsc{skewInit} is correct.
\end{lemma}
\begin{proof}

Let $x$ be the (starting) desired value of the skew number, $x_i$ be the value of $x$ at the iteration of the while loop with value $i$, and $l_i$ be the value of the skew number at the beginning of the iteration with value $i$.

Note that we increase the list $l$ (the skew number) in only two places: when adding two copies of $2^i-1$, and when adding a single copy of $2^i-1$. In both cases, we immediately subtract this amount from $x$. Thus, after each loop iteration we maintain the invariant that $x=x_i+l_i$. Also note that $x_i$ and $l_i$ are nonnegative integers at all times. 

Next, we claim that $2(2^i-1) \ge x_i$, for all $i$. We proceed with induction. After the first iteration, $i$ is $\lceil \log_2 x \rceil+1$, so this is true. Now suppose $2(2^i-1) \ge x_i$. We must show that after the next iteration, $2(2^{i-1}-1) \ge x_{i-1}$. We have three cases.

\begin{enumerate}
    \item If $2(2^i-1)=x_i$, then $x_{i-1}$ will be zero, so this is true. 
    \item If $2^i - 1 > x_i$, then $2^i-2 \ge x_i$, so it follows that $2(2^{i-1}-1) \ge x_i = x_{i-1}$. 
    \item If $2^i - 1 \le x_i$ and $2(2^i-1)\ne x_i$, then $x_{i-1}=x_i-2^i+1$. Since $2(2^i-1) \ne x_i$ by assumption and $2(2^i-1) \ge x_i$ by the inductive hypothesis, we have $x_i < 2(2^i-1)$ and so $x_i \le 2(2^i-1)-1$, so $x_i-2^i+1 \le 2(2^i-1)-1-2^i+1= 2^{i}-2=2(2^{i-1}-1)$, as desired.  
\end{enumerate}
    
\end{proof}

\begin{lemma}
    The function \textsc{skewAdd} is correct.
\end{lemma}

\begin{proof}
Let $x_0$ be the starting value of $x$ and $l_0$ the initial skew number.
Note that $x$ is changed in four locations (decrementing, subtracting $w_1$, another decrement, and setting to 0). In each location, the skew number is increased by the same amount. Therefore, the invariant is maintained that (immediately after) $x$ is changed, $x + l = x_0+l_0$. After the function run, $x_{\text{final}}=0$, so $l_{\text{final}}=l_0+x_0$ as desired. 


Next, note that in every iteration, $x$ is decremented by at least 1. Therefore, the while loop will terminate.

When the while loop terminates, we have that $x < w_1$, where $w_1$ is the first entry in the linked list representation of the skew number. 

Next, we repeatedly decrement $x$ and merge the first two elements of the list until the skew number does not have its first nonzero digit be a two.

Finally, we can append a skew version of $x$ to the list, and the result is a skew number, and the correct skew number.
\end{proof}
\begin{lemma}
    The function \textsc{skewSubtract} is correct.
\end{lemma}

\begin{proof}
Note that in both places where we decrement $x$, we reduce the skew number by the same amount, preserving the constant $l-x$. Since $x$ decrements every iteration we eventually finish. Note that splitting the front of a skew number gives a skew number, and so does removing the front number, so the resulting $l$ is a skew number. 

\end{proof}

\begin{lemma}
    Let the maximum allowable value across all updates be $poly(n)$. On an input $x$, the function skewInit is $O(\log x)$ work and \depth{}. The functions skewAdd and skewSubtract are $O(\log x + \min(x,\log n))$ work and depth.
\end{lemma}

\begin{proof}
We prove each separately. 
\begin{description}
\item[skewInit:] Since we decrement $i$ in each round of the while loop and $i = O(\log x)$, and appending to a linked list is a constant time operation, the bounds follow. 

\item[skewAdd:] Consider the run of the first while loop (starting from line 6). Note that we either add $1$ or $w_1$. Whenever we add $1$, the value of $w_1$ doubles, and the next value we add must be $w_1$. Whenever we add $w_1$, the next value we add must be $1$. We therefore have at most $O(\log x)$ iterations of this loop.

In the second while loop (line 20), every time we add $1$, the value of the smallest node doubles, so we can add $1$ at most $O(\log n)$ times. We can also add $1$ at most $x$ times because the total demand is $x$. The final skewInit call is $O(\log x)$ work and \depth{}. Therefore, the total work and \depth{} is $O(\log x + \min(x,\log n))$.

\item[skewSub:] We can view the iterations of \textsc{skewSubtract} in stages. In the first stage, we remove the front of the list until we find an element larger than the remaining value in $x$. This will take at most $O(\log x)$ subtracts because the front of the list doubles in size after the first remove (the first remove may not double because of the 2 digit at the beginning of some skew numbers). 

From this point onward, whenever we decrement $x$ by 1, write down split$(w)$, where $w$ was the front of the skew list at the time. Whenever we decrement by $w$ for some $w$, write down remove$(w)$. Observe that we will never write down split$(w)$ twice (for the same $w$): because if $x$ was large enough to split both instances of $w$ in a list, we would have removed the first instance of $w$ wholesale instead of splitting it. Similarly, we will never write down remove$(w)$ twice for the same $w$, because if we had the budget to split (into two copies of the same value) and remove both, we would have removed both at the same time earlier.

Because writing down remove$(w)$ corresponds with removing $w$ from $x$, we can write down remove at most $O(\log x)$ times. Because writing down split corresponds with removing 1 from $x$, we can write down split at most $x$ times. Because each split$(w)$ value is different, and the maximum capacity is $poly(n)$, we can write down split at most $O(\log n)$ times. Therefore, the total number of second stage iterations is at most $O(\log x + \min(x,\log n))$, for total work and \depth{} $O(\log x + \min(x,\log n))$.
\end{description}
\end{proof}

\begin{algorithm}
\label{alg:skew-init}


\tcp{Given a natural number $x$, initialize a skew number with value $x$}
\function{\textsc{skewInit}$(x)$}{
 $l \gets $ empty linked list \;
  $i \gets \lceil \log_2 x \rceil+2$\;
  \while{$x > 0$}{
   
    \If{$2(2^i - 1) == x$} {
    add two copies of $2^i - 1$ to front of list $l$  \;  $x \gets 0$ \;
    }

    \Elif{$(2^i - 1) \le x$}{
    add $2^i - 1$ to front of list $l$ \;
    $x \gets x - 2^i + 1$ \;
    }

     $i \gets i - 1$ \;
   
  }
  \Return $l$

}

\caption{Skew numbers: Initialization}

\end{algorithm}

\begin{algorithm}
\label{alg:skew-add}

\tcp{Note that :: means appended}
\function{\textsc{skewAdd}$(l,x)$}{
  \If{isEmpty$(l)$} {
    $l \gets \textsc{skewInit}(x)$ \;
    \Return $l$ 
  }
  $w_1 :: rest \gets l$ \;
  \while{$x \ge w_1$}{
    
    flag $\gets $ True \;
   
    \If{ $|l| > 1$}{
      $w_1 :: w_2 :: rest \gets l$\;
      \If {$w_1 == w_2$}{
        $l \gets (1+w_1+w_2) :: rest$ \;
        $x \gets x-1$ \;
        flag $\gets $ False \;

      }

    }

     \If{flag} {
      $l \gets w_1 :: l$ \;
      $x \gets x - w_1$ \;
      
      }

      $w_1 \gets \textrm{head}(l)$ \;

  }
  \If{$|l| > 1$} {
  $w_1 :: w_2 :: rest \gets l$ \;
  \while {$w_1==w_2$ and $x > 0$}{
  $l \gets (1+ w_1 + w_2) :: rest$ \;
  $x \gets x - 1$ \;
  $w_1 :: w_2 :: rest \gets l$ \;
  }

  }
  
  $l_2 \gets \textsc{skewInit}(x)$ \;
  $l \gets l_2 :: l$ \;
  $x \gets 0$ \; 
  \Return $l$
    
}

\caption{Skew numbers: Skew add}

\end{algorithm}

\begin{algorithm}
\label{alg:skew-subtract}

\tcp{Note that :: means appended}
\function{\textsc{skewSubtract}$(l,x)$}{

  \while{$x > 0$}{
   $w :: rest \gets l$ \;
   \If{$x \ge w$} {
   $x \gets x - w$\;
   $l \gets rest$ \;
   }
   \Else {
     $x \gets x - 1$ \;
     $l \gets \frac{w-1}{2} :: \frac{w-1}{2} : :rest $ \;

   }
   }
  \Return $l$ \;
    
}

\caption{Skew numbers: Skew subtract}

\end{algorithm}

\subsection{Using skew numbers for bag}

Okasaki describes how to make a random-access list using skew numbers with increments and decrements allowed. We create a bag structure (no random-access) with skew numbers with arbitrarily sized appends and deletes. Note that we can convert an array of size $x$ to a complete binary tree and back in $O(x)$ work and $O(\log x)$ \depth{} \cite{jaja1992parallel}.

To each node of the linked list representation of the skew number, we add a pointer to a complete binary tree.

\myparagraph{Peek} To peek from a skew bag, we sequentially prefix sum to see the appropriate number of elements. For all of the numbers except the largest one, we write out the elements into an array with the appropriate offset. For the largest number, it is possible that the associated tree is much larger than the peek value. In this case, we take the tree of all nodes at most depth $d$ (for a shallow value of $d$, such that $2^d$ is slightly larger than the number of nodes we must peek), write these nodes to an array, then take the prefix of appropriate size. 


\myparagraph{Append, Initial} To insert to a skew bag (or during skewInit), to increment, one grabs an element from the batch and makes it the parent of two existing binary trees in the linked list. 
To insert a value $2^{i}-1$, this many elements are taken from the input array and formed into an independent complete binary tree, and associated with a new node in the linked list with value $2^i-1$.

\myparagraph{Append, Optimized} The above version of append adds a $O(\log x)$ factor to the \depth{}, because in each while loop iteration, $O(\log x)$ depth of creating a binary tree is spent. To reduce the \depth{} down to $O(\log x)$, the creation of the binary tree is deferred until the end, where all binary trees are created at once. 


\myparagraph{Delete} We keep a list of trees that we remove from the bag. Whenever we remove a single element (a root), we store this as a tree with one vertex. Thus, after running delete, we have a list of length $O(\log x+\min(x,\log n))$ of removed elements. Some of the elements we wish to delete (and some elements we did not want to delete) will have been removed from the bag, and some are still within the bag. We want to place the elements we accidentally removed into the holes created by the deleted elements still in the bag. 

To do this, each deleted element marks that it is being deleted, and each removed element marks itself as having been removed. Then, filter the deleted pointers list to only include those not removed, and filter the removed elements list by those which should not have been deleted. Now we can map the elements we want to return into pointers of deleted elements still in the bag. 

\myparagraph{Work and Depth} For removing and replacing $x$ elements we pay $O(x)$ work. For the underlying skew number operations, we pay 
$O(\log x + \min(x,\log n))$, for overall $O(x)$ work. The depth is $O(\log x)$ plus the depth of skew number operations, which is $O(\log n)$, for $O(\log n)$ depth.


Note that $O(\log x + \min(x,\log n))$ work is upper bounded by $O(x)$, and that the depth is upper bounded by $O(\log n)$.
Thus we conclude Lemma \ref{lem:bag}.

\subsection{Pannier}

We briefly give the pseudocode for extracting elements (\textsc{batchPop}) from a pannier.

\begin{algorithm}
\tcp{Pop $x$ elements from the pannier $Q$}
\function{\textsc{batchPop}$(Q,x)$}{
\If{$x \le |Q_F|$}{
\Return \textsc{batchPop}$(Q_F,x)$ \;

}
\Else {
    $S_1 \gets \textsc{batchPop}(Q_F,|Q_F|)$ \;
    $Q_F \gets Q_B$, $Q_B \gets \emptyset$ \;
    $S_2 \gets \textsc{batchPop}(Q_F,x-|S_1|)$ \;
    \Return $S_1 \cup S_2$ \;

}

}
\end{algorithm}

%% file: appendix_skyline.tex
\section{Maintaining Skylines\label{sec:appendix_skyline}}

In this section, we show that we can maintain skylines efficiently. Skylines were introduced in Section \ref{sec:skyline}.
Tables~\ref{tab:notation} and \ref{tab:skyline_notation} summarize the relevant notation used throughout this section. 

In order to know which vertices have high out-degree, throughout our algorithms, we will store the vertices in a roughly sorted list (RSL) of rounding $c'$, where the value of a vertex is its out-degree. The choice of $c'$ will be determined by the orientation quality we seek.  

\myparagraph{Finding the demands} Recall that the demands for a skyline are the number of edges we need to take from each vertex. To find the demands for a skyline, note that the skyline contribution in each stratum of the RSL can be split into the contribution from old (higher) vertices (which will always be $c'$ per vertex) and the new (current bag) vertices, as demonstrated by Figure \ref{fig:skylineNewOld}. We sequentially prefix sum ($O(M)$ \depth{}) on the old and new stratum sizes to find the cutoff. Because we can choose edges from the last stratum arbitrarily, we opt to take $c'$ edges from as many vertices as possible, starting with the old vertices. If we need to take any edges from new vertices of this stratum, we use a prefix sum on the edge counts in this bucket to determine which vertices to take edges from. See Algorithm \ref{alg:approx_skyline} for the pseudocode for finding the demands for a skyline. 

\begin{table}[h]
\begin{center}
\begin{tabular}{c|c}
Symbol & Meaning \\
\hline
$\Delta$ & maximum vertex degree in the graph \\
$s_T(v)$ & $d^+(v) - T$ (distance from threshold $T$) \\ 
$n_T$ & $|\{v \in V : d^+(v) \geq T\}|$ (number of vertices above threshold $T$)  \\ 
$C_k$ & $\sum_{j=T+1}^\Delta n_j$ (number of edges strictly above threshold $T$) \\ 
$C_k'$ & $C_k$ after a skyline flip, when ambiguous \\ 
$Q_v$ & set of out edges of vertex $v$, in the pannier order \\ 
$c'$ & the rounding of the skyline \\
$T$ & typically, the threshold of a skyline  \\
\end{tabular} 
\end{center} 
\caption{Skyline-related notation. See Table \ref{tab:notation} for additional general notation.}\label{tab:skyline_notation} 
\end{table}

\begin{algorithm}
\label{alg:skyline-approx-threshold}
\tcp{$P$: the vertex RSL}
\tcp{$c'$: rounding of $P$}
\tcp{$M$: budget of $P$}
\tcp{$x$: desired skyline size}
\tcp{Output: the skyline threshold $T$}
\function{\textsc{getThreshold}$(P,c',M,x)$}{
//$l$ is number of vertices seen, $s$ is the running skyline sum \;
     $s \gets 0, l \gets 0, j \gets 0$ \;
    \For{$ \{i \in [M+1,M,\ldots,1 : |P[i]| > 0 ]$ } {
        \If{$s=0$} {$j \gets i$} 
        \If{$s + lc'(j-i) \ge x$}  
        {\Return $(j-\lceil \frac{x-s}{lc'} \rceil-1)c'$} 
        $s \gets s + l c' (j-i)$ \;
        $j \gets i$ \;
        $l \gets l + |P[i]|$ \;
        
        $Y \gets peek(P[i],\min(|P[i]|,x))$ \;
        $s \gets s + \sum_{v \in Y} (d^+(v) - \min(Mc',r(d^+(v)-1,c')))$ \;
        \If{$s \ge x$}{\Return $(i-1)c'$} }  } \end{algorithm}

\begin{algorithm}
\label{alg:skyline-approx-oldDemands}
\tcp{$P,c',M,x$: same as above} 
\tcp{$T$: threshold for desired skyline}
\tcp{\textsc{initial}: number of edges taken conservatively so far from high vertices}
\tcp{\textsc{highVertices}: the vertices from which we will take our demanded edges}
\tcp{Output: the demands for the skyline}
\function{\textsc{oldDemands}$(P,c',M,x,T,\textsc{initial},\textsc{highVertices})$}{
     left $\gets x-\textsc{initial}$ //number of fill in needed in final level \;
      \textsc{numTake} $\gets \lfloor \frac{left}{c'} \rfloor$ //number of vertices needed to give this fill in \;
      \textsc{lastRem} $\gets \textrm{left} \mod c'$ //remainder to take from last vertex \;
      \parfor{$v \in \textsc{highVertices}$}{
        //demand strictly above last bucket \;
        \textsc{demand}$(v) \gets d^+(v)-T - c'$ \;

      }

      \parfor{$v \in \textsc{highVertices}[1:\textsc{numTake}] $} 
      { 
        //demand in range of last bucket \;
        
        \textsc{demand}$(v) \gets \textsc{demand}(v) + c'$ \;
      
      }
      \If{$\textsc{lastRem} > 0$} {
      \textsc{demand}$(\textsc{highVertices}[\textsc{numTake}+1]) \gets \textsc{demand}$(\textsc{highVertices}[\textsc{numTake}+1]) + \textsc{lastRem} \; }
      \Return \textsc{demand} \;
    
}
\end{algorithm}

\begin{algorithm}
    \label{alg:skyline-approx-newDemands}
\tcp{$x,P,c',M,T,\textsc{highVertices}$: same as above} 
\tcp{\textsc{collect}: number of out-edges of vertices in \textsc{highVertices} counted so far}
\tcp{Output: the demands for the skyline}
\function{\textsc{newDemands}$(P,c',M,x,T,\textsc{collect},\textsc{highVertices})$}{
  $\textsc{left} \gets x-\textsc{collect}$ \;

  \parfor{$v \in \textsc{highVertices}$}{
        //demands for vertices with out-degree strictly above last bucket \;
        \textsc{demand}$(v) \gets d^+(v)-T$ \;
      }
      //$A$ is vertices with out-degree in range $(T,T+c']$ \;
      $A \gets P[\frac{T}{c'}+1]$ \;
      $A \gets peek(A,\min(|A|,\textsc{left}))$ \;
      $S \gets \textsc{prefixSum}(\{d^+(x) - T | x \in A \})$ \;
      $j \gets \textsc{binarySearch}(S,\textsc{left})$ //first index where prefix sum greater or equal to left \;

       \parfor{$v \in A[1:j-1]$}{
        \textsc{demand}$(v) \gets d^+(v)-T$ \;
      }
      \textsc{demand}$(A[j]) \gets \textsc{left}-S[j-1]$ \;
      \Return \textsc{demand} \;

}
\end{algorithm}

\begin{algorithm}
\label{alg:skyline-approx-demands}

\caption{Algorithm for finding an approximate skyline of rounding $c'$ and max bucket index $M$ and size $x$ using roughly sorted list}
\label{alg:approx_skyline}
\tcp{$x,P,c',M$: same as above} 
\tcp{Output: the demands for the skyline}
\function{\textsc{getDemands}$(P,c',M,x)$}{

     $T \gets \textsc{getThreshold}(P,c',M,x)$ \;
     $\textsc{highBag} \gets \frac{T}{c'}+2$ \;
     $\textsc{highVertices} \gets \{v : v \in P[i], i \in [\textsc{highBag},M+1] \}$ \;
     $ \textsc{initial} \gets \textsc{sum}( \{ d^+(v)-T-c' : v \in \textsc{highVertices} \} )$ \;
     $\textsc{collect} \gets \textsc{sum}( \{ d^+(v)-T : v \in \textsc{highVertices} \} )$ \;
     \If{$\textsc{collect} \ge x$} { 
       \Return \textsc{oldDemands}$(P,c',M,x,T,\textsc{initial},\textsc{highVertices})$

     }

     \Else {
       \Return \textsc{newDemands}$(P,c',M,x,T,\textsc{collect},\textsc{highVertices})$
     }
}

\end{algorithm}

\begin{restatable}{lemma}{skylineDemands}
Given an out-degree orientation and a desired skyline size $x$, if $\sum_{v \in \text{highBag}} (d^+(v)-Mc') \le x$ (the total weight of the high RSL bag is upper bounded by $x$), we can find demands that give a skyline of size $x$ in $O(x)$ work and $O(\log^2 n)$ \depth{}.
\label{lem:demands}
\end{restatable}

\begin{proof}[Proof of Lemma \ref{lem:demands}]
Consider finding the demands of a skyline of size $x$, according to Algorithm \ref{alg:approx_skyline}. 

First, consider \textsc{getThreshold}. 
Suppose that in the $i^{th}$ round, we peek $a_i$ vertices. Note that we can assume $a_i > 0$, because we can skip empty bags in the RSL with standard bit tricks, as described in Section \ref{sec:prelims}. 
Note that $\sum_i a_i = x$, so we have total work $O(x)$. 

Because the
  span of peeking $a_i$ vertices is $O(\log n)$, the total \depth{} is $r \log n$, where $r$ is the number of nonzero rounds. Note that $r \le \min(x,M)$, where $M$ is the budget of the underlying RSL, so the span is at most $O(M \log n) \le O(\log^2 n)$. We also have $O(\log n)$ span for the peek on line 8 in \textsc{newDemands}.
Since \textsc{highVertices} is of size at most $x$, we have $O(\log x)$ \depth{} and $O(x)$ work for the remaining functions. 
Thus, the overall work is $O(x)$ and the \depth{} is $O(\log^2 n)$.

By the additive error property of the RSL, the demands collected actually correspond to a skyline, so long as all weight is taken from the high bag (because the vertices in the high bag can have vastly different out-degree and are unsorted). Since  $\sum_{v \in \text{highBag}} (d^+(v)-Mc') \le x$, we do take all weight from the high bag. Therefore, the demands we have do correspond to a skyline.
\end{proof}

We will need the following lemma in order to complete our analysis. Because our algorithms will behave well if skyline flips are correct, and skyline flips are correct if the algorithm behaves well, informally it follows that our algorithms work as intended. The following lemma formalizes this intuition.

\begin{lemma} \label{lem:skylinesCorrectImpliesCorrect}
    Consider a low out-degree orientation algorithm $A$ that uses \textsc{getDemands} to request skylines. 
    
    We call a skyline extraction correct if the demands returned give a skyline (per the skyline definition). Suppose that if \textsc{getDemands} always correctly extracts skylines, then $A$ maintains a maximum out-degree of $D \le Mc' - 3c$. Further suppose that a single invocation $I_A$ of $A$ only will request skylines of size exactly $y := y(I_A)$.
    
    Then $A$ does maintain a maximum out-degree of $D$.
\end{lemma}
\begin{proof}

Before the start of $I_A$, there is no mass in the high bag, because $D \le Mc' - 3c$. After statically orienting and inserting the insert batch, there is still no mass in the high bag, because of the $3c$ buffer between $D$ and $Mc'$. A single skyline flip can place at most $y$ mass inside the high bag. Since all skyline flips are of size $y$, the next skyline flip will include all of this mass, preventing mass from building up in the high bag. Therefore, all skyline flips during $I_A$ will be correct, and so at the end of $I_A$, the high bag returns to empty. Therefore by induction (on the update sequence), $A$ does maintain a maximum out-degree of $D$.
    
\end{proof}

\begin{figure}
\begin{minipage}{.4\textwidth}
\centering
\includegraphics[width=\textwidth]{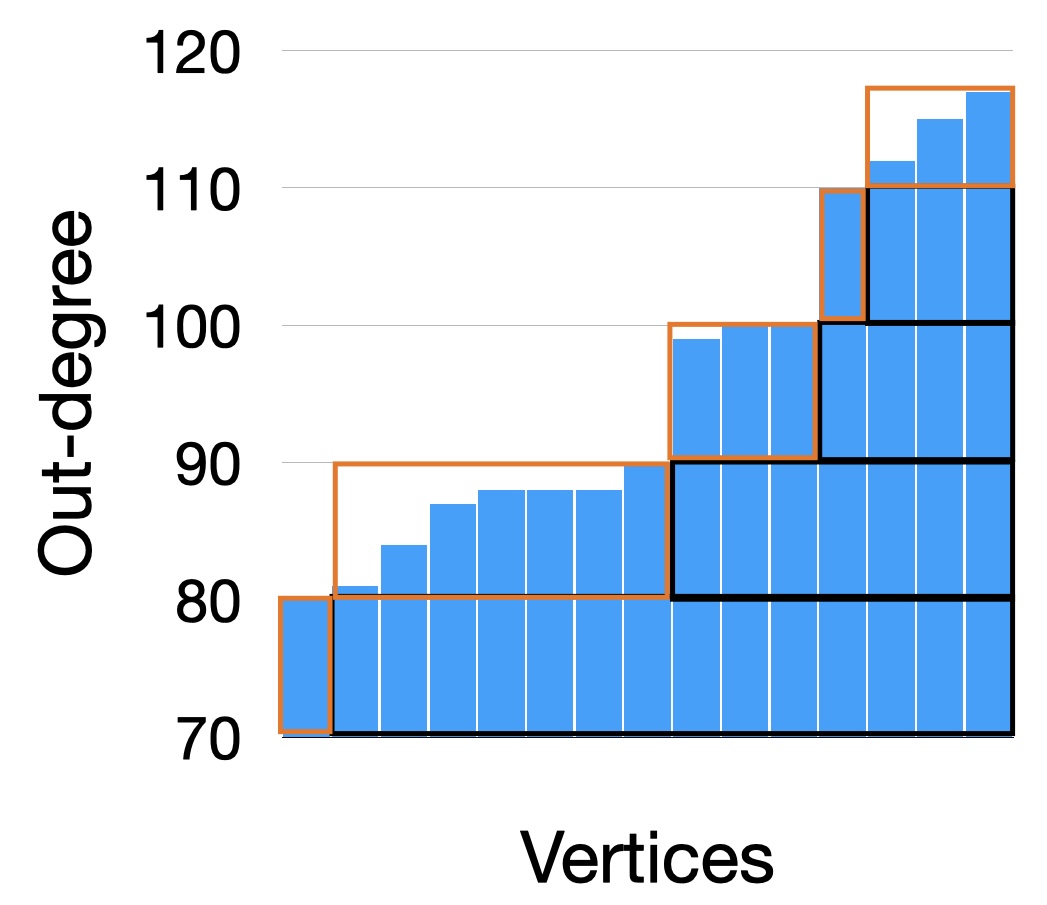}    
\end{minipage}
\begin{minipage}{.55\textwidth}
\caption{New and old vertices from strata perspective of a skyline with rounding 10. Note that the complete skyline for this graph is not shown, only a partial view. The black boxes denote out-degree being added in this stratum from vertices existing at higher strata. The orange boxes denote out-degree being added from vertices who first exist in this stratum. }
\end{minipage}
    \label{fig:skylineNewOld}
    
\end{figure}

\myparagraph{Flipping a Skyline} To actually flip a skyline, we wrap finding the vertex demands with deletes and inserts from the underlying panniers holding the out-edges. Note that using panniers for the out-edges and a roughly sorted list for the vertex out-degrees is critical in saving a $O(\log n)$ factor in efficiency. We want to support two ways to flip edges given a size parameter $x$. In the first way, we find a skyline of size $x$ and flip all edges in the skyline. In the second way, we find a skyline of size $x$, statically orient the skyline, and flip any edges which changed orientation during the static orientation procedure. Note that an edge can keep a field giving its old orientation, and so we can tell if an edge has been flipped or not. Upon moving the edge to its proper bag, we reset this field for future use.

\begin{lemma}
Finding and flipping a skyline of size $x$ takes $O(x)$ expected work and $O(\log^2 n)$ \depth{}. The work bound is also $O(x \log x)$ work deterministically.
\label{lem:skylineFlipCost}
\end{lemma}

\begin{proof}[Proof of Lemma \ref{lem:skylineFlipCost}]

We must find the vertex demands, remove the edges, flip the edges, then reinsert to the proper panniers. Additionally, the roughly sorted list of vertex out-degree must be updated. 

\begin{itemize}
    \item Getting the vertex demands is $O(x)$ work and $O(\log^2 n)$ \depth{} by Lemma \ref{lem:demands}. 
    \item Doing \textsc{batchPop} from a pannier, to extract the edges, is $O(x)$ work and $O(\log n)$ \depth{}, as there are $x$ total edges, by Lemma \ref{lem:bag}. 
    \item The edges must then be flipped and semisorted, for $O(x)$ expected work and $O(\log x)$ \depth{} \whp{}.
    \item Giving the edges to their new owners takes $O(x)$ work and $O(\log n)$ \depth{}  by Lemma \ref{lem:bag}. 
    \item There are $O(x)$ vertices who may change position in the roughly sorted list. Updating these values takes $O(x)$ expected work and $O(\log n)$ \depth{}, by Lemma \ref{lem:rsl}. 
\end{itemize}

Overall, we had $O(x)$ expected work and $O(\log^2 n)$ \depth{} \whp{}. By Lemma \ref{lem:sort}, the cost for semisorting is deterministically $O(x \log x)$, resulting in $O(x \log x)$ deterministic work overall.
\end{proof}

\begin{lemma} 
Finding a skyline of size $x$, submitting the skyline to static orientation, and flipping the edges that change orientation after static orientation, takes $O(x)$ expected work, $O(x \log x)$ deterministic work, and $O(\log^2 n)$ \depth{}. 
 
\label{lem:skylineFlipStaticCost}
\end{lemma}

\begin{proof}[Proof of Lemma \ref{lem:skylineFlipStaticCost}]
   The cost of static orientation itself is $O(x)$ expected work and $O(\log^2 x)$ \depth{}. The other costs incurred are the same as for a standard skyline flip,  so we achieve the same bounds.
\end{proof}

\subsection{Thresholds of skylines}

In this subsection, we show that when flipping many skylines of (crucially) the same size, the induced threshold of later skylines is at most slightly larger than that of earlier skylines. 

\begin{lemma}
    Let $S$ be a skyline with size $x$, rounding $c'$, and induced threshold $T$. Observe that the amount of weight above $L$ for $L \le T$ does not increase after flipping $S$. \label{lem:skylineWeightsDown}
\end{lemma}
\begin{proof}
    When we flip an edge, we either decrease the weight above $L$ by one (if the edge now points from a low degree vertex) or leave the weight unchanged (if the edge now points from a high degree vertex). Thus, flipping a skyline cannot increase the amount of weight above its own threshold. 
\end{proof}

Note that the above lemma corresponds to the counter game Lemma \ref{lem:weightsDown}. Note that this is not necessarily true for $L > T$. For $L > T$, we need a weaker lemma. 

\begin{lemma}
    Suppose we have a skyline $S$ with size $x$, induced threshold $T$, and rounding $c'$. Let $C_L$ denote the weight above $L$ (for $c'|L$, $L > T$) before the flip and $C_L'$ the weight above $L$ after the flip. Then $C_{L}' \le x$. 
 
    \label{lem:skylineWeightsUp}
\end{lemma}

\begin{proof}
View the skyline flip as two steps, deletion and reinsertion. Because $T$ was the induced threshold, we have that $C_{T+c'} < x$. Therefore, all edges above $T+c'$ are in the skyline. Therefore, after deletion, the current weight above $T+c'$ is zero. Inserting $x$ edges back (in the flipped direction) increases the weight above $T+c'$ by at most $x$. Thus $C_{T+c'}' \le x$. Finally, note that $C_{L}' \le C_{T+c'}' \le x$ since higher thresholds have less weight above them.\footnote{Since $L > T$, we have $L \ge T+c'$.} 
    
\end{proof}
\begin{lemma}
    Suppose we flip skylines $S_1, S_2,\ldots,S_r$ each with rounding $c'$ and size $x$. Denote their induced thresholds by $T_1,T_2,\ldots,T_r$. Then  we have that $C_{T_1 + c'} \le x$ before the first skyline flip, and for each $1 \le i \le r$ after flipping skyline $S_i$.
       \label{lem:skylineSizeStable}
\end{lemma}

\begin{proof}

    We proceed with induction. 
    
    Base Case: Before flipping $S_1$, note that $C_{T_1+c'} < x$, otherwise $T_1+c'$ (or higher) would have been the induced threshold.

    Inductive Step: Suppose we have flipped $i$ skylines, and that currently $C_{T_1+c'} \le x$. We are about to flip $S_{i+1}$. We must show that after flipping $S_{i+1}$, that $C_{T_1+c'} \le x$ still. 
    
    Recall that $T_{i+1}$ is the induced threshold by this flip. We case on the value of $T_{i+1}$. 

    When $T_{i+1} \ge T_1 + c'$, by Lemma \ref{lem:skylineWeightsDown} we have that $C_{T_1+c'}' \le C_{T_1+c'} \le x$. When $T_{i+1} < T_1+c'$, by Lemma \ref{lem:skylineWeightsUp}, we have $C_{T_1+c'}' \le x$. 
    
\end{proof}

\begin{lemma}
    Let $S_1,S_2,\dots,S_k$ be skylines each with rounding $c'$ each of size $x$ where $S_i$ is collected after flipping $S_{i-1}$. Additionally, denote the corresponding thresholds $T_1,T_2,\dots,T_k$. Then for each $1 \le i \leq k$, we have that 
    \[T_i \leq T_1 + c'.\]
    \label{lem:skyline-thresholds}
\end{lemma}

\begin{proof}
    We proceed with induction. For $i=1$, this is trivially true. Now suppose it is true for values up to $i$. Consider flipping skyline $S_{i+1}$ of size $x$ with induced threshold $T_{i+1}$ and rounding $c'$. We must show that $T_{i+1} \le T_1 + c'$. 
    
    By Lemma \ref{lem:skylineSizeStable}, we have that $C_{T_1+c'} \le x$ before flipping $S_{i+1}$. Therefore, $T_{i+1} \le T_1 + c'$ (because $x$ weight is needed for this skyline). 
\end{proof}

\begin{lemma}
    Let $S_1,S_2,\dots,S_k$ be skylines each with rounding $c'$ each of size $x$ where $S_i$ is collected after flipping $S_{i-1}$. Additionally, denote the corresponding thresholds $T_1,T_2,\dots,T_k$. Then for each $1 \le j \le i \leq k$, we have that 
    \[T_i \leq T_j + c'.\]
    \label{lem:skyline-thresholds-ij}
\end{lemma}

\begin{proof}
Let $j \in [k]$. Then the result follows from applying Lemma \ref{lem:skyline-thresholds} on the subsequence of skyline flips from $j$ to $k$. 
\end{proof}

This lemma immediately gives Lemma \ref{lem:skyline-thresholds-min}.

%% file: 10_conclusion.tex
\section{Conclusion}

In this paper we have presented two work-efficient algorithms for parallel batch-dynamic low out-degree orientation: one in the amortized setting, and one in the worst-case setting. We highlight three potential directions for future work:

\begin{enumerate}
    \item Can more efficient low out-degree orientations be maintained in the case where the graph is very sparse (i.e., constant arboricity)? For example, Bender et al. gave an algorithm for incremental orientation on forests \cite{bender2021incremental}. 
    \item Kaplan and Solomon gave an amortized algorithm for dynamic orientation in the distributed setting \cite{kaplan18dynamic}; one potential line of work is to see if their approach can be modified for the parallel batch-dynamic setting, and conversely, if the techniques in this paper would apply well to the distributed setting. 
    \item We note that the counter game does not take into account the density constraint given by the arboricity of the graph. Would it be possible to improve the counter game by incorporating this constraint, so that the counter game achieves an $O(c)$ bound instead of $O(c + \log n)$? 
\end{enumerate}

%% file: appendix_prelims.tex
\section{Appendix: Orientation Primitives \label{sec:Appendix}}

\subsection{Static orientation \label{sec:staticOrientationAppendix} }

\begin{algorithm}[h]
\caption{Parallel Algorithm for Static Low Out-Degree Orientation \cite{barenboim2008sublogarithmic}}\label{alg:staticorient}
\tcp{Input: $E$, a directed bag of edges}
\tcp{Input: $c$, the arboricity}
\tcp{Input: $\epsilon$, affects the quality of the orientation}
\tcp{Output: A bag of directed edges $F$ with maximum out-degree $(2+\epsilon)c$}
\function{\textsc{StaticOrientation}$(E,c, \epsilon)$}{
    $F \gets \{\}$ \;
    $E' \gets \{(u,v): (u, v)\in E \lor (v,u) \in E \}$ \; 
    \While{$|E'| > 0$}{
        $E' \gets \textsf{semisort}(E')$, induces $d(v)$\;
        $V' \gets \textsf{removeDuplicates}(\{u:(u,v) \in E'\})$\;
        \parfor{$v$ in $V'$}{
            \If{$d(v) \le (2+\epsilon)\cdot c$}{
                $\textsf{vmark}(v)$\;
            }
        }
        \parfor{$(u,v)$ in $E'$}{
           
            \If{$\textsf{vmarked}(u)$ and ( $\lnot \textsf{vmarked}(v)$ or ($\textsf{vmarked}(v)$ and  $u < v$))}{
                $\textsf{emark}((u,v))$\;
            }
        }
        
        \textsf{batchInsert}$(F, \textsf{filter}(E',\textsf{emarked}))$\;
        \textsf{batchDelete}$(E',(u,v) : \textsf{vmarked}(u) \textrm{ or } \textsf{vmarked}(v))$\;
    }
    \Return $F$
}

\end{algorithm}

\staticLem*

\begin{proof}[Proof of Lemma \ref{lem:static-orient}]
Consider a round of the while loop. Let $|V'|$ be the number of remaining vertices (i.e., vertices with positive degree), and $|E'|$ be the number of remaining edges. Let $|V_{\text{high}}|$ be the number of vertices that have more than $(2+\epsilon)\cdot c$ degree, and thus survive to the next round, and let $E_{\text{high}}$ be the corresponding edges that will survive. 

Observe that an edge can only survive if both of its endpoints are in $V_{\text{high}}$. 
By the arboricity constraint, at most $c |V_{\text{high}}|$ edges survive. At the same time, there are at least $\frac{2+\epsilon}{2} c |V_{\text{high}}|$ edges in the graph (before peeling), because each vertex in $V_{\text{high}}$ carries at least $(2+\epsilon)c$ degree. 
Therefore, $|E'| \ge \frac{2+\epsilon}{2} c |V_{\text{high}}| \ge \frac{2+\epsilon}{2} |E_{\text{high}}|$, so $|E_{\text{high}}| \le \frac{2|E'|}{2+\epsilon}$. 
Therefore, the number of participating edges decreases by a $2/(2+\epsilon)$ fraction in a single round. 

Therefore, by a geometric series sum, the number of participating edges across all rounds is bounded by $O(\epsilon^{-1} |E|)$. Solving $\left(2/(2+\epsilon)\right)^r=1/(2|E|)$ yields that the number of rounds is $O(\epsilon^{-1} \log |E|)$.  

The work is proportional to the total number of participating edges, so by Lemma \ref{lem:sort} the total work is $O(\epsilon^{-1} |E|)$ in expectation and $O(\epsilon^{-1} |E| \log |E|)$ deterministically. Since the maximum capacity of a bag is $O(|E|)$, by Lemma \ref{lem:bag} the \depth{} of bag operations is $O(\log |E|)$. We have $O(\epsilon^{-1} \log |E|)$ rounds, with a semisort and bag operations within each round, leading to overall $O(\epsilon^{-1} \log^2 |E|)$ \depth{}. 
\end{proof}

\subsection{Boosting Orientation Quality}


The following technique is found, for example, in Chekuri et al.  \cite{chekuri2024adaptive}, and shows that an algorithm that maintains an $O(c + \log n)$-orientation can be improved to give an $O(c)$-orientation.

\begin{lemma}[See for example \cite{chekuri2024adaptive}] \label{lem:boostOrientation}
    Given an algorithm $A$ that maintains a $O(c+\log n)$ orientation in $O(b f(n))$ work and $O(d(n))$ depth on a batch update of size $b$, there exists an algorithm $A'$ that maintains a $O(c)$ orientation in $O(b f(n) \log n)$ work and $O(d(n))$ depth.
\end{lemma}

\begin{proof}
Let $G$ be our original graph. Consider the multigraph $G'$, which has $\log n$ copies of each edge in $G$. Each edge in $G$ points to the $\log n$ copies of the edge in $G'$. A batch insert or delete on $G$ will trigger a batch insert or delete on all $\log n$ copies of that edge from $G'$. Although the main body of the paper is written in the language of simple graphs, all of our algorithms and data structures function with a multigraph with minor modification, and the analysis functions unmodified. Note that the arboricity $c'$ of $G'$ is upper bounded by $c \log n$. Use algorithm $A$ to maintain a $O(c' + \log n)$ low out-degree orientation on $G'$: this will cost $O(b f(n) \log n)$ work and the same depth, because $G'$ has a $\log n$ factor more edges than $G$. 

Each edge $(u,v)$ in $G$ will keep a counter denoting how many edges in $G'$ are pointing $u \rightarrow v$. After $G'$ updates, if more than half of the edges in $G'$ point $u \rightarrow v$, $(u,v)$ will point $u \rightarrow v$, and otherwise, the edge will point $v \rightarrow u$. Since each out-edge of $u$ in $G$ corresponds to at least $\frac{\log n}{2}$ out-edges in $G'$, the out-degree in $G$ is bounded by $\frac{O(c' + \log n)}{\log n} = O(c)$.

\end{proof}

%% file: appendix_potential.tex
\section{Appendix: Potentials}\label{sec:potential-setup}

In this section, we prove the lemmas given in Section \ref{sec:potential-setup-main}. Our setup is slightly different, but these proofs largely follow from the techniques of Berglin and Brodal \cite{BB20}.

\begin{lem}[Similar to Lemma 3 in \cite{BB20}] 
\label{lem:queue-swap}
For a vertex $v$, when $|B_v| \ge 4 \delta$, moving all edges from $B_v$ to $F_v$ does not increase $p(v)$. 
\end{lem}
\begin{proof}
    Note that $F_v$ need not be empty. Note that any good edges we move to the front, of which there are at most $\delta$, increase the potential by $3 \epsilon$ each. 

    The first $3\delta$ bad back edges carry $1+\epsilon$ potential (and thus reduce potential by $\epsilon$ when moved to the front), and the rest carry $1$ potential (and so leave the potential unchanged when moved to the front). Since there are at most $\delta$ good edges in $B_v$, there are at least $|B_v|-\delta \ge 3\delta$ bad edges in $B_v$, and so moving all bad edges in $B_v$ to the front yields a $3\delta\epsilon$ reduction. Therefore the potential change is at most zero.
\end{proof}

\begin{lem}[From proof of Lemma 4 in \cite{BB20}]
Flipping an edge in a front bag releases at least $\epsilon$ potential. \label{lem:flip-front}
\end{lem}

\begin{proof}
Observe that if the edge was good, the flipped edge is now bad, and we lose potential (either $\epsilon$ or $2\epsilon$) and if the edge was bad, the flipped edge is now good, so we lose $\epsilon$ potential. \end{proof}

\lemSkylinePotential*

\begin{proof}


Let $u$ be a vertex that has at least one out-edge in $S$. If $\textsc{demand}(u) \le |F_u|$, then we do not reassign the back bag to the front bag in the pannier. If $\textsc{demand}(u) > |F_u|$, then $|B_u| = d^+(u)-|F_u| \ge d^+(u) - \max(d^+(u) - T - c',0) - \rho_u$. When $d^+(u) \ge T + c'$, this simplifies to $d^+(u) - d^+(u) + T + c' - \rho_u \ge T \ge 4\delta$. When $d^+(u) < T + c'$, this simplifies to $d^+(u) - \rho_u \ge T \ge 4\delta$. Therefore, by Lemma \ref{lem:queue-swap}, reassigning the back to be the front does not increase the potential. 

From the point of view of each edge, the edge starts in a front bag when it is flipped. Therefore by Lemma \ref{lem:flip-front} flipping this edge releases at least $\epsilon$ potential. 

Flipping a skyline has two types of movements: moving edges from a back bag to a front bag of the same vertex, and moving from the front bag of one endpoint to the back bag of the other endpoint. The first type of movement leads to a net negative or zero potential change, and the second type of movement releases at least $\epsilon$ potential per edge moved in this way. Therefore, the total net potential release is at least $x \epsilon$.

\end{proof}

\lemStaticPotential*

\begin{proof} By the same argument as in the above proof, reassigning the back to be the front (if needed) does not increase the potential, each edge can be viewed as starting in a front bag when flipped, and each flip releases at least $\epsilon$ potential. 

Because we are statically orienting (rather than doing a wholesale flip), some edges will be removed from a front bag and return to the same front bag (net zero potential change), and others will be removed from a front bag and placed in the back bag of the other endpoint (releasing $\epsilon$ potential). Therefore, the net potential change is at most zero.

\end{proof}

Note that statically orienting (or flipping) a skyline of insufficient height can increase potential, both from the flips and from moving vertices from the back to the front bag. However, this is fine as we do not need to track potential increases while our counter game is not running, see Theorem \ref{thm:generic-outdegree} for details.

\lemPotentialIncrease*

\begin{proof}
    If the update is an insertion the new edge is added to the back bag of some vertex and contributes at most $1 + \epsilon$ potential in the case that it is bad and among the first $3\delta$ bad edges in the back bag. Per update, the offline strategy $\kappa$ makes at most $\sigma$ flips each of which causes an edge to swap classification from good to bad or vice versa. Each reclassification can increase potential by at most $2\epsilon$. Thus in total, the sum of inserted potential is at most $1 + \epsilon  + 2\sigma\epsilon$.
\end{proof}

%% file: appendix_apps.tex
\section{Appendix: Applications \label{sec:apps}}

In the applications, we will need the hash table of Gil et al., which yields the following bounds in the binary forking model. 

\begin{lemma}
    A parallel batch-dynamic hash table can be maintained in $O(b)$ expected work, where $b$ is the size of an update or query. Queries take $O(\log b)$ depth deterministically and updates take $O(\log n \log^* n)$ depth whp in $n$ \cite{gil1991towards}.
\end{lemma}

\subsection{Maximal Matching}
Here we apply our batch-dynamic orientation algorithms to the problem of
maintaining a maximal matching. A \emph{maximal matching} $M$ of a graph
$G=(V,E)$ is a set of edges such that no two edges in $M$ share an endpoint,
and every edge in $E \setminus M$ shares an endpoint with some edge in $M$.
The connection between low out-degree orientations and dynamic maximal
matching appears in many papers in the low-outdegree orientation
literature~\cite{Kopelowitz2013Orientation, neiman2015simple}. The key
idea is that a $c$-orientation allows each unmatched vertex to
efficiently propose matchings to a bounded number of neighbors,
enabling fast searches for ``replacement'' edges for the matching when
a matched edge is deleted.

\myparagraph{Amortized Batch-Dynamic Maximal Matching}
Liu et al.~\cite{liu2022parallel} gave a parallel batch-dynamic
maximal matching algorithm based on their $(4+\epsilon)c$-orientation
algorithm. Their approach maintains, for each vertex $v$, a hash table
$I_v$ of unmatched in-neighbors (vertices $u$ with an edge oriented
from $u \to v$ where $u$ is unmatched). On edge insertions, they run a
parallel static maximal matching~\cite{BFS12, fischer2018tight} on the
subgraph induced by inserted edges between unmatched vertices. On edge
deletions that remove matched edges, they first attempt to rematch the
newly unmatched vertices with their out-neighbors, and then use a
doubling scheme to progressively query more in-neighbors until all
vertices are either matched or have no unmatched neighbors. The work
analysis requires a simple charging argument to amortize the cost of
unsuccessfully checking an in-neighbor that gets matched to another
vertex in the same batch.

Our improved orientation algorithm (Algorithm~\ref{alg:amortized}) can
be directly plugged into this framework. The only change is replacing
their $(4+\epsilon)c$-orientation with our $7c$-orientation, which has
better work bounds. The matching algorithm itself and the analysis
remain unchanged.

\begin{theorem}\label{thm:matching-amortized}
Given a graph $G$ with arboricity at most $c$, there exists a parallel
  batch-dynamic algorithm that maintains a maximal matching $M$ of
  $G$. For a batch $B$ of $b$ edge updates, the algorithm runs in
  amortized expected $O(b(c + \log n))$ work and $O(\log^2 n)$
  \depth{}  \whp{}.
\end{theorem}

\begin{proof}
The proof follows the analysis of Liu et al.~\cite{liu2022parallel},
substituting our orientation bounds. By Theorem~\ref{thm:amortized_main}, our
orientation algorithm performs amortized $O(b \log n)$ expected work per
batch. The matching maintenance costs $O(bc)$ work per batch: processing
out-neighbors costs $O(c)$ per newly unmatched vertex, and the cost of
scanning in-neighbors amortizes to $O(c)$ per vertex that becomes matched
(since each matched in-neighbor has out-degree $O(c)$ and is queried at most
$O(c)$ times before becoming matched). The static maximal matching subroutine
runs in $O(\log^2 n)$ \depth{} \whp{}~\cite{BFS12, fischer2018tight}. Summing the
orientation and matching costs gives $O(b(c + \log n))$ amortized expected
work. The \depth{} is $O(\log^2 n)$ \whp{}, dominated by the static matching and
orientation subroutines.
\end{proof}

\myparagraph{Worst-Case Batch-Dynamic Maximal Matching}
We can similarly obtain a worst-case result by using our worst-case orientation
algorithm (Algorithm~\ref{alg:wealg}).

\begin{theorem}\label{thm:matching-worstcase}
Given a graph $G$ with arboricity at most $c$, there exists a parallel
batch-dynamic algorithm that maintains a maximal matching $M$ of $G$. For a
batch $B$ of $b$ edge updates, the algorithm runs in worst-case expected
$O(b(c + \log n))$ work and $O(\log^4 n)$ \depth{} \whp{}.
\end{theorem}

\begin{proof}
We use the same matching framework as in the amortized case, but with our $O(c + \log n)$-orientation algorithm. As observed in Ghaffari and Koo~\cite{ghaffari2025}, the framework
from Liu et al.~\cite{liu2022parallel} also works in the worst-case setting so
long as the low out-degree orientation can be maintained with worst-case bounds.
The orientation algorithm performs worst-case $O(b \log n)$ expected work per
batch. The matching maintenance costs $O(b(c + \log n))$ worst-case work per
batch, since each vertex has out-degree at most $O(c + \log n)$. The  \depth{} is
$O(\log^4 n)$ \whp{}, dominated by the orientation algorithm.
\end{proof}

\myparagraph{Comparison with Recent Batch-Dynamic Maximal Matching Algorithms}
Ghaffari and Trygub~\cite{GT24} gave a parallel batch-dynamic maximal matching
algorithm with $O(\log^{9} n)$ amortized work per edge and $O(\log^4 n)$ \depth{}
per batch (both \whp{}). More recently, Blelloch and Brady obtain a parallel
batch-dynamic maximal matching algorithm that runs in $O(1)$ expected amortized
work per update, and $O(\log^3 n)$  \depth{} \whp{} \cite{BB25a}. Compared to these amortized results,
our amortized algorithm is arguably simpler and runs in $O(c + \log n)$
amortized expected work per update and $O(\log^2 n)$  \depth{} \whp{}, thus providing a slight
improvement in the \depth{}. 

Ghaffari and Koo~\cite{ghaffari2025} obtained a parallel batch-dynamic maximal matching algorithm
that runs in $O(c + \log^{8} n)$ worst-case work per update and $O(\log^{7} n)$ \depth{}, both
\whp{}. By contrast, our worst-case result is significantly less complex and achieves $O(c
+ \log n)$ expected work per edge in $O(\log^4 n)$ \depth{} \whp{}.

\subsection{Explicit Coloring}\label{sec:coloring}
\mbox{} \\
\myparagraph{Explicit Coloring via Palettes}
Ghaffari and Koo~\cite{ghaffari2025} give a simple palette-based approach for
maintaining an explicit $O(\delta \log n)$-coloring of an input graph with arboricity $c \le \delta$.
The first time a vertex $v$ becomes incident to an edge, 
$v$ is assigned a random \emph{palette} $\mathcal{C}(v)$: each color from a pool
of $C = O(\delta \log n)$ colors is included in $\mathcal{C}(v)$ independently with
probability $\frac{1}{2\delta}$, giving each vertex $\Theta(\log n)$ colors \whp{}. Given a low out-degree orientation where each vertex has at most
$O(\delta)$ out-neighbors, each vertex $v$ selects any color from its palette that
does not appear in any out-neighbor's palette. Since each color appears in
$\mathcal{C}(v)$ with probability $\frac{1}{2\delta}$ and is absent from each
out-neighbor's palette with probability $1 - \frac{1}{2\delta}$, the probability that
a given color is ``good'' (present in $\mathcal{C}(v)$ and absent from all $O(\delta)$
out-neighbors' palettes) is $\Omega(1/\delta)$. With $\Theta(\delta \log n)$ colors in the
pool, at least one good color exists \whp{}. Because the palettes do not change, the algorithm only functions against an oblivious adversary.


When edges are inserted or deleted, only vertices whose out-neighbors change
must recompute their colors. To recolor a vertex $v$, we check each color in
$\mathcal{C}(v)$ against the palettes of $v$'s $O(\delta)$ out-neighbors; checking a
single color takes $O(\delta)$ work and $O(\log n \log^* n)$  \depth{} (each palette is stored in a hash table),
and checking all $O(\log n)$ colors in $\mathcal{C}(v)$ takes $O(\delta \log n)$ work
and $O(\log n \log^* n)$  \depth{} in parallel. 

\myparagraph{Worst-Case Batch-Dynamic Coloring}
We can similarly obtain worst-case results by using our worst-case orientation
algorithm (Algorithm~\ref{alg:wealg}) and various tradeoffs.

\begin{theorem}\label{thm:coloring-worstcase}
Consider a graph $G$ and a sequence of edge updates of total length at most $\text{poly}(n)$, where the arboricity of $G$ is upper bounded by $c$ at all times. There exists a parallel batch-dynamic algorithm that maintains an explicit $O(c \log n)$-coloring of $G$ over this update sequence \whp{} against an oblivious adversary. For a batch $B$ of $b$ edge updates, the algorithm runs in worst-case
$O(b c \log^3 n)$ work \whp{} and $O(\log^4 n)$  \depth{} \whp{}.
\end{theorem}

\begin{proof}
Our worst-case $O(c)$-orientation algorithm (Algorithm~\ref{alg:wealg}) with Theorem \ref{thm:boundsOc} performs at most
$O(b \log^2 n)$ flips in $O(b \log^3 n)$ deterministic work, and each flip causes at most one vertex to have its
out-neighborhood increase. If the palettes are suitable, which is true \whp{}, each such vertex recolors itself in $O(c \log n)$ work and
$O(\log n \log^* n)$ \depth{}. Summing these costs gives $O(b(c \log n) \log^2 n)$ work \whp{}.  The \depth{} is dominated by the \depth{} of orientation, which is $O(\log^4 n)$  \depth{}.
\end{proof}

%% file: refs.bib
@article{MatulaBeck1983,
  author     = {David W. Matula and Leland L. Beck},
  title      = {Smallest-Last Ordering and Clustering and Graph Coloring Algorithms},
  journal    = {Journal of the ACM},
  volume     = {30},
  number     = {3},
  pages      = {417--427},
  month      = jul,
  year       = {1983},
  publisher  = {Association for Computing Machinery},
  address    = {New York, NY, USA},
  doi        = {10.1145/2402.322385},
}

@article{BB20,
author = {Berglin, Edvin and Brodal, Gerth St\o{}lting},
title = {A Simple Greedy Algorithm for Dynamic Graph Orientation},
year = {2020},
issue_date = {Feb 2020},
publisher = {Springer-Verlag},
address = {Berlin, Heidelberg},
volume = {82},
number = {2},
issn = {0178-4617},
url = {https://doi.org/10.1007/s00453-018-0528-0},
journal = {Algorithmica},
month = feb,
pages = {245–259},
numpages = {15}
}

@article{MH81,
title = {Real-time queue operations in pure LISP},
journal = {Information Processing Letters},
volume = {13},
number = {2},
pages = {50-54},
year = {1981},
issn = {0020-0190},
doi = {10.1016/0020-0190(81)90030-2},
author = {Robert Hood and Robert Melville}
}

@inproceedings{chekuri2024adaptive,
  title={Adaptive out-orientations with applications},
  author={Chekuri, Chandra and Christiansen, Aleksander Bj{\o}rn and Holm, Jacob and van der Hoog, Ivor and Quanrud, Kent and Rotenberg, Eva and Schwiegelshohn, Chris},
  booktitle=soda,
  pages={3062--3088},
  year={2024},
  organization={SIAM}
}

@inproceedings{BF99,
	address = {Berlin, Heidelberg},
	author = {Brodal, Gerth St{\o}lting and Fagerberg, Rolf},
	booktitle = {Algorithms and Data Structures},
	title = {Dynamic Representations of Sparse Graphs},
	year = {1999}
}

@inproceedings{peleg2016dynamic,
  title={Dynamic {$(1+\epsilon)$}-approximate matchings: A density-sensitive approach},
  author={Peleg, David and Solomon, Shay},
  booktitle=soda,
  pages={712--729},
  year={2016},
  organization={SIAM}
}

@inproceedings{barenboim2008sublogarithmic,
  title={Sublogarithmic distributed {MIS} algorithm for sparse graphs using {Nash-Williams} decomposition},
  author={Barenboim, Leonid and Elkin, Michael},
  booktitle=podc,
  pages={25--34},
  year={2008}
}

@inproceedings{gu2022analysis,
  title        = {Analysis of Work-Stealing and Parallel Cache Complexity},
  author       = {Gu, Yan and Napier, Zachary and Sun, Yihan},
  booktitle    = apocs,
  pages        = {46--60},
  year         = {2022},
  organization = {SIAM}
}

@article{neiman2015simple,
  title={Simple deterministic algorithms for fully dynamic maximal matching},
  author={Neiman, Ofer and Solomon, Shay},
  journal={ACM Transactions on Algorithms (TALG)},
  volume={12},
  number={1},
  pages={1--15},
  year={2015},
  publisher={ACM New York, NY, USA}
}

@inproceedings{anderson2021parallel,
  title={Parallel Minimum Cuts in {$O (m \log^2 n)$} Work and Low Depth},
  author={Anderson, Daniel and Blelloch, Guy E},
  booktitle={Proceedings of the 33rd ACM Symposium on Parallelism in Algorithms and Architectures},
  pages={71--82},
  year={2021}
}

@inproceedings{liu2022parallel,
  title={Parallel batch-dynamic algorithms for k-core decomposition and related graph problems},
  author={Liu, Quanquan C and Shi, Jessica and Yu, Shangdi and Dhulipala, Laxman and Shun, Julian},
  booktitle=spaa,
  pages={191--204},
  year={2022}
}

@book{jaja1992parallel,
author = {J\'{a}J\'{a}, Joseph},
title = {An Introduction to Parallel Algorithms},
year = {1992},
isbn = {0201548569},
publisher = {Addison Wesley Longman Publishing Co., Inc.},
address = {USA}
}

@book{okasaki1999purely,
  title={Purely functional data structures},
  author={Okasaki, Chris},
  year={1999},
  publisher={Cambridge University Press}
}

@inproceedings{goodrich2023optimal,
  title={Optimal parallel sorting with comparison errors},
  author={Goodrich, Michael T and Jacob, Riko},
  booktitle=spaa,
  pages={355--365},
  year={2023}
}

@InProceedings{wang2020closest,
  author =	{Wang, Yiqiu and Yu, Shangdi and Gu, Yan and Shun, Julian},
  title =	{{A Parallel Batch-Dynamic Data Structure for the Closest Pair Problem}},
  booktitle =	{37th International Symposium on Computational Geometry (SoCG 2021)},
  pages =	{60:1--60:16},
  series =	{Leibniz International Proceedings in Informatics (LIPIcs)},
  ISBN =	{978-3-95977-184-9},
  ISSN =	{1868-8969},
  year =	{2021},
  volume =	{189},
  editor =	{Buchin, Kevin and Colin de Verdi\`{e}re, \'{E}ric},
  publisher =	{Schloss Dagstuhl -- Leibniz-Zentrum f{\"u}r Informatik},
  address =	{Dagstuhl, Germany},
   URN =		{urn:nbn:de:0030-drops-138594},
  doi =		{10.4230/LIPIcs.SoCG.2021.60}
}

@inproceedings{acar2019batchconnect,
  author    = {Umut A. Acar and
               Daniel Anderson and
               Guy E. Blelloch and
               Laxman Dhulipala},
  title     = {Parallel Batch-Dynamic Graph Connectivity},
  booktitle = spaa,
  pagesx= {381--392},
  publisher = {ACM},
  year      = {2019}
}

@inproceedings{bender2021incremental,
  title={Incremental Edge Orientation in Forests},
  author={Bender, Michael A and Kopelowitz, Tsvi and Kuszmaul, William and Porat, Ely and Stein, Clifford},
  booktitle=esa,
  year={2021},
}

@inproceedings{fischer2018tight,
	title={Tight analysis of parallel randomized greedy {MIS}},
	author={Fischer, Manuela and Noever, Andreas},
	booktitle=soda,
	pagesx={2152--2160},
	year={2018},
}

@inproceedings{BFS12,
	title={Greedy sequential maximal independent set and matching are parallel on average},
	author={Blelloch, Guy E. and Fineman, Jeremy T and Shun, Julian},
	booktitle= spaa,
	year={2012},
}

@inproceedings{ghaffari2023nearly,
  title={Nearly work-efficient parallel {DFS} in undirected graphs},
  author={Ghaffari, Mohsen and Grunau, Christoph and Qu, Jiahao},
  booktitle=spaa,
  pages={273--283},
  year={2023}
}

@article{yesantharao2021parallel,
  title={Parallel Batch-Dynamic $k$-d Trees},
  author={Yesantharao, Rahul and Wang, Yiqiu and Dhulipala, Laxman and Shun, Julian},
  journal={arXiv preprint arXiv:2112.06188},
  year={2021}
}

@article{BL98,
	author    = {Robert D. Blumofe and
	Charles E. Leiserson},
	title     = {Space-Efficient Scheduling of Multithreaded Computations},
	journal   = siamjc,
	volume    = {27},
	number    = {1},
	year      = {1998},
}

@Article{ABP01,
	author="Arora, N. S.
	and Blumofe, R. D.
	and Plaxton, C. G.",
	title="Thread Scheduling for Multiprogrammed Multiprocessors ",
	journal=tocs,
	year="2001",
	month="Apr",
	day="01",
	volume="34",
	number="2",
}

@inproceedings{gu2015top,
	title={A top-down parallel semisort},
	author={Gu, Yan and Shun, Julian and Sun, Yihan and Blelloch, Guy E},
	booktitle=spaa,
	year={2015},
}

@incollection{KarpR90,
	author    = {Richard M. Karp and
	Vijaya Ramachandran},
	title     = {Parallel Algorithms for Shared-Memory Machines},
	booktitle = {Handbook of Theoretical Computer Science, Volume A: Algorithms
	and Complexity (A)},
	year      = {1990},
	publisher = {MIT Press},
}

@book{CLRS,
	author    = {Thomas H. Cormen and
	Charles E. Leiserson and
	Ronald L. Rivest and
	Clifford Stein},
	title     = {Introduction to Algorithms (3rd edition)},
	publisher = {MIT Press},
	year      = {2009},
}

@inproceedings{christiansen2022fully,
  title={Fully-Dynamic $\alpha$+ 2 Arboricity Decompositions and Implicit Colouring},
  author={Christiansen, Aleksander BG and Rotenberg, Eva},
  booktitle=icalp,
  pages={42--1},
  year={2022},
  organization={Schloss Dagstuhl--Leibniz-Zentrum f{\"u}r Informatik}
}

@article{kowalik2007adjacency,
  title={Adjacency queries in dynamic sparse graphs},
  author={Kowalik, {\L}ukasz},
  journal={Information Processing Letters},
  volume={102},
  number={5},
  pages={191--195},
  year={2007},
  publisher={Elsevier}
}

@inproceedings{acar2020changeprop,
  author    = {Umut A. Acar and
               Daniel Anderson and
               Guy E. Blelloch and
               Laxman Dhulipala and
               Sam Westrick},
  title     = {Parallel Batch-Dynamic Trees via Change Propagation},
  booktitle = esa,
  pagesx= {2:1--2:23},
  year      = {2020},
}

@inproceedings{gil1991towards,
	title={Towards a theory of nearly constant time parallel algorithms},
	author={Gil, Joseph and Matias, Yossi and Vishkin, Uzi},
	booktitle=focs,
	year={1991},
}

@inproceedings{BlellochF0020,
  author    = {Guy E. Blelloch and
               Jeremy T. Fineman and
               Yan Gu and
               Yihan Sun},
  title     = {Optimal Parallel Algorithms in the Binary-Forking Model},
  year      = {2020},
  booktitle = spaa,
}

@inproceedings{dhulipala2019parallel,
  title={Parallel Batch-Dynamic Graphs: Algorithms and Lower Bounds},
  author={Dhulipala, Laxman and Durfee, David and Kulkarni, Janardhan and Peng, Richard and Sawlani, Saurabh and Sun, Xiaorui},
  booktitle=soda,
  pagesx={1300--1319},
  year={2020}
}

@inproceedings{dhulipala2021hierarchical,
  title={Hierarchical Agglomerative Graph Clustering in Nearly-Linear Time},
  author={Dhulipala, Laxman and Eisenstat, David and {\L}{\k{a}}cki, Jakub and Mirrokni, Vahab and Shi, Jessica},
  booktitle = icml,
  pagesx = {2676--2686},
  year={2021}
}

@inproceedings{GT24,
title = {Parallel Dynamic Maximal Matching},
author = {Mohsen Ghaffari and Anton Trygub},
booktitle = spaa,
year = 2024}

@inproceedings{he2014orienting,
  title={Orienting dynamic graphs, with applications to maximal matchings and adjacency queries},
  author={He, Meng and Tang, Ganggui and Zeh, Norbert},
  booktitle={International Symposium on Algorithms and Computation},
  pages={128--140},
  year={2014},
  organization={Springer}
}

@inproceedings{bernstein2016faster,
  title={Faster fully dynamic matchings with small approximation ratios},
  author={Bernstein, Aaron and Stein, Cliff},
  booktitle=soda,
  pages={692--711},
  year={2016},
  organization={SIAM}
}

@inproceedings{bernstein2015fully,
  title={Fully dynamic matching in bipartite graphs},
  author={Bernstein, Aaron and Stein, Cliff},
  booktitle={International Colloquium on Automata, Languages, and Programming},
  pages={167--179},
  year={2015},
  organization={Springer}
}

@article{onak2020fully,
  title={Fully dynamic {MIS} in uniformly sparse graphs},
  author={Onak, Krzysztof and Schieber, Baruch and Solomon, Shay and Wein, Nicole},
  journal={ACM Transactions on Algorithms (TALG)},
  volume={16},
  number={2},
  pages={1--19},
  year={2020},
  publisher={ACM New York, NY, USA}
}

@article{henzinger2020explicit,
  title={Explicit and implicit dynamic coloring of graphs with bounded arboricity},
  author={Henzinger, Monika and Neumann, Stefan and Wiese, Andreas},
  journal={arXiv preprint arXiv:2002.10142},
  year={2020}
}

@inproceedings{ghaffari2025,
  title={Parallel Batch-Dynamic Coreness Decomposition with Worst-Case Guarantees},
  author={Ghaffari, Mohsen and Koo, Jaehyun},
  booktitle=spaa,
  year={2025}
}

@inproceedings{BB25a,
author = {Blelloch, Guy E. and Brady, Andrew C.},
title = {Parallel Batch-Dynamic Maximal Matching with Constant Work per Update},
year = {2025},
isbn = {9798400712586},
publisher = {Association for Computing Machinery},
address = {New York, NY, USA},
url = {https://doi.org/10.1145/3694906.3743313},
booktitle = spaa,
pages = {429–442},
numpages = {14},
location = {Portland, OR, USA},
series = {SPAA '25}
}

@inproceedings{ZKG+25,
author = {Zhao, Yiwei and Kang, Hongbo and Gu, Yan and Blelloch, Guy E and Dhulipala, Laxman and McGuffey, Charles and Gibbons, Phillip B},
title = {Optimal Batch-Dynamic kd-trees for Processing-in-Memory with Applications},
booktitle = spaa,
year = 2025
}

@article{NashWilliams1961EdgeDisjoint,
	author = {Nash-Williams, C. St.J. A.},
	doi = {10.1112/jlms/s1-36.1.445},
	journal = {Journal of the London Mathematical Society},
	number = {1},
	pages = {445-450},
	title = {Edge-Disjoint Spanning Trees of Finite Graphs},
	volume = {s1-36},
	year = {1961},
}

@InProceedings{Kopelowitz2013Orientation,
author="Kopelowitz, Tsvi
and Krauthgamer, Robert
and Porat, Ely
and Solomon, Shay",
editor="Esparza, Javier
and Fraigniaud, Pierre
and Husfeldt, Thore
and Koutsoupias, Elias",
title="Orienting Fully Dynamic Graphs with Worst-Case Time Bounds",
booktitle="Automata, Languages, and Programming",
year="2014",
publisher="Springer Berlin Heidelberg",
address="Berlin, Heidelberg",
pages="532--543",
isbn="978-3-662-43951-7"
}

@inproceedings{SW20,
author = {Sawlani, Saurabh and Wang, Junxing},
title = {Near-optimal fully dynamic densest subgraph},
year = {2020},
isbn = {9781450369794},
publisher = {Association for Computing Machinery},
address = {New York, NY, USA},
doi = {10.1145/3357713.3384327},
booktitle = {Proceedings of the 52nd Annual ACM SIGACT Symposium on Theory of Computing},
pages = {181–193},
numpages = {13},
location = {Chicago, IL, USA},
series = {STOC 2020},
}

@inproceedings{GJS21,
  title={Atomic power in forks: A super-logarithmic lower bound for implementing butterfly networks in the nonatomic binary fork-join model},
  author={Goodrich, Michael T and Jacob, Riko and Sitchinava, Nodari},
  booktitle=soda,
  pages={2141--2153},
  year={2021},
  organization={SIAM}
}

@article{BURTON82,
title = {An efficient functional implementation of {FIFO} queues},
journal = {Information Processing Letters},
volume = {14},
number = {5},
pages = {205-206},
year = {1982},
issn = {0020-0190},
doi = {10.1016/0020-0190(82)90015-1},
author = {F.Warren Burton}
}

@article{MYERS83,
title = {An applicative random-access stack},
journal = {Information Processing Letters},
volume = {17},
number = {5},
pages = {241-248},
year = {1983},
issn = {0020-0190},
doi = {10.1016/0020-0190(83)90106-0},
author = {Eugene W. Myers}
}

@article{solomon20improved,
author = {Solomon, Shay and Wein, Nicole},
title = {Improved Dynamic Graph Coloring},
year = {2020},
issue_date = {July 2020},
publisher = {Association for Computing Machinery},
address = {New York, NY, USA},
volume = {16},
number = {3},
issn = {1549-6325},
doi = {10.1145/3392724},
journal = talg,
month = jun,
articleno = {41},
numpages = {24}
}

@inproceedings{kaplan18dynamic,
author = {Kaplan, Haim and Solomon, Shay},
title = {Dynamic Representations of Sparse Distributed Networks: A Locality-Sensitive Approach},
year = {2018},
isbn = {9781450357999},
publisher = {Association for Computing Machinery},
address = {New York, NY, USA},
doi = {10.1145/3210377.3210397},
booktitle = spaa,
pages = {33–42},
numpages = {10},
location = {Vienna, Austria},
series = {SPAA '18}
}


%% file: strings.bib
@string{talg = "ACM Transactions on Algorithms (TALG)"}

@string{siamjc = "{SIAM} J. on Computing"}

@string{tocs = "Theory of Computing Systems (TOCS)"}

@string{spaa = "{ACM} Symposium on Parallelism in Algorithms
	and Architectures (SPAA)"}

@string{apocs = "{ACM-SIAM} Symposium on Algorithmic Principles of Computer Systems (APOCS)"}

@string{focs = "{IEEE} Symposium on Foundations of Computer Science (FOCS)"}

@string{stoc = "{ACM} Symposium on Theory of Computing (STOC)"}

@string{soda = "{ACM-SIAM} Symposium on Discrete Algorithms (SODA)"}

@string{esa = "European Symposium on Algorithms (ESA)"}

@string{icalp = "Intl. Colloq. on Automata, Languages and Programming {(ICALP)}"}

@string{icml = "International Conference on Machine Learning (ICML)"}

@string{podc = "{ACM} Symposium on Principles of Distributed Computing (PODC)"}

@string{socg = "International Symposium on Computational Geometry (SoCG)"}
